\documentclass[english]{article}
\usepackage{lmodern}
\usepackage[T1]{fontenc}
\usepackage[latin9]{inputenc}
\usepackage{geometry}
\usepackage{color}
\usepackage{amsthm}
\usepackage{amsmath}
\usepackage{amssymb}
\usepackage{setspace}
\usepackage{esint}

\makeatletter
\numberwithin{equation}{section}
\theoremstyle{plain}
\newtheorem{thm}{\protect\theoremname}[section]
  \theoremstyle{definition}
  \newtheorem{defn}[thm]{\protect\definitionname}
  \theoremstyle{plain}
  \newtheorem{prop}[thm]{\protect\propositionname}
  \theoremstyle{plain}
  \newtheorem{fact}[thm]{\protect\factname}
  \theoremstyle{remark}
  \newtheorem{rem}[thm]{\protect\remarkname}

 \date{}

\makeatother

\usepackage{babel}
  \providecommand{\definitionname}{Definition}
  \providecommand{\factname}{Fact}
  \providecommand{\propositionname}{Proposition}
  \providecommand{\remarkname}{Remark}
\providecommand{\theoremname}{Theorem}

\begin{document}

\title{Time-Inconsistent Mean-Utility \\
Portfolio Selection with Moving Target}

\author{Hanqing Jin%
\thanks{Mathematical Institute and Nomura Centre for Mathematical Finance,
University of Oxford, Woodstock Road, Oxford OX2 6GG and Oxford\textendash{}Man
Institute, Eagle House, Walton Well Road, Oxford OX2 6ED. %
} $\mbox{ \ensuremath{}\ensuremath{}}$and Yimin Yang%
\thanks{$ $Mathematical Institute, University of Oxford, Woodstock Road,
Oxford OX2 6GG.%
}}

\maketitle
\begin{spacing}{3}
\begin{center}
\textbf{\large{Abstract}}
\par\end{center}{\large \par}
\end{spacing}

\begin{onehalfspace}
\textcolor{black}{In this paper, we solve the time inconsistent portfolio
selection problem by using different utility functions with a moving
target as our constraint. We solve this problem by finding an equilibrium
control under the definition given as our optimal control. We derive
a sufficient equilibrium condition for $\mbox{C}^{2}$ utility funtions
and use power functions of order two, three and four in our problem,
and find the respective condtions for obtaining an equilibrium for
our different problems. In the last part of the paper, we also consider
use another definition of equilibrium to solve our problem when the
utility function we use in our problem is $x^{-}$ and also find the
condtions for obtaining an equilibrium.}
\end{onehalfspace}

\newpage{}\tableofcontents{}\newpage{}

\section{Introduction}

\begin{singlespace}
\noindent Stochastic control is now a mature and well established
subject of study. Here we quote as an introduction a summary of the
studies in time inconsistent control problems from Hu, Jin and Zhou\cite{key-1}
as follows. In the study of stochastic control, though not explicitly
stated at most of the times, a standard assumption is time consistency,
which is a fundamental property of conditional expectation with respect
to a progressive filtration. As a result of that, an optimal control
viewed from today will remain optimal viewed from tomorrow. Time consistency
provides the theoretical foundation of the dynamic programming approach
including the resulting HJB equation, which is in turn a pillar of
the modern stochastic control theory. However, there are overwhelmingly
more time inconsistent problems than their time consistent counterparts.
Hyperbolic discounting Ainslie\cite{key-7} and continuous-time mean\textendash{}variance
portfolio selection model Basak Chabakauri\cite{key-8}, Zhou and
Li\cite{key-9} provide two well-known examples of time inconsistency.
Probability distortion, as in behavioral finance models Jin and Zhou\cite{key-10},
is yet another distinctive source of time inconsistency. Motivated
by practical applications especially in mathematical finance, time
inconsistent control problems have recently attracted considerable
research interest and efforts attempting to seek equilibrium controls
instead of optimal controls. At a conceptual level, the idea is that
a decision the controller makes at every instant of time is considered
as a game against all the decisions the future incarnations of the
controller are going to make. An \textquotedblleft{}equilibrium\textquotedblright{}
control is therefore the one that any deviation from it at any time
instant will be worse off. Taking this game perspective, Ekeland and
Lazrak\cite{key-11} approach the deterministic time inconsistent
optimal control, and Bjork and Murgoci\cite{key-3} extends the idea
to the stochastic setting, derive an albeit very complicated HJB equation,
and apply the theory to a dynamic Markowitz problem. Yong\cite{key-12}
investigate a time inconsistent deterministic linear quadratic control
problem and derive equilibrium controls via some integral equations.
However, the study of time inconsistent control is still in its infancy
in general.\\

\end{singlespace}

\noindent \textcolor{black}{In this paper, we solve the time inconsistent
portfolio selection problem by using different utility functions with
a moving target as our constraint. We solve this problem by finding
an equilibrium control under the given definition as our optimal control.
Firstly, we derive a sufficient equilibrium condition for $\mbox{C}^{2}$
utility funtions. Then we use power functions of order two, three
and four in our problems and find the respective condtions for obtaining
an equilibrium for our different problems. In the last part of the
paper, we consider use another definition of equilibrium to solve
our problem when the utility function that we use in our problem is
$x^{-}$ and we also find the condtions for obtaining an equilibrium
for this problem.}\\

\begin{singlespace}
\noindent The structure of this paper is as follows. In section 2,
we set the problems that we want to study and give the definition
of equilibrium when we use our power utility functions in our problem.
We also prove a sufficient equilibrium condition for $\mbox{C}^{2}$
utility functions in this section. In section 3, we use $h\left(x\right)=\frac{X^{2}}{2}$
as the utility functions in our problem and find the conditions for
obtaining an equilibrium for our problem by setting a deterministic
process taking the position as a Lagrangian multiplier. In sections
4 and 5 we use $h\left(x\right)=-\frac{X^{3}}{3}$ and $\frac{x^{4}}{4}$
in our problem respectively, and we directly transform our problem
into a simpler form and find the conditions for obtaining an equilibrium
for our problem. In section 6 we use $h\left(x\right)=x^{-}$ in our
problem. we give a different definition of equilibrium that we work
on in this section, and as usual we find the conditions for obtaining
an equilibrium. Finally we make some remarks in section 7.
\end{singlespace}

\newpage{}

\section{Problems setting }

\subsection{The notations }

The notations used in this paper are listed as follows
\begin{align*}
\begin{cases}
L_{\mathcal{F}}^{\infty}\left(t,T;\mathbb{R}^{l}\right): & \mbox{the set of essentially bounded \ensuremath{\left\{  \mathcal{F}_{s}\right\} } }_{s\in\left[t,T\right]}\mbox{-adapted processes.}\\
L_{\mathcal{F}}^{2}\left(t,T;\mathbb{R}^{l}\right): & \mbox{the set of \ensuremath{\left\{  \mathcal{F}_{s}\right\} } }_{s\in\left[t,T\right]}\mbox{-adapted processes }\\
 & f=\left\{ f_{s}:t\leq s\leq T\right\} \mbox{ with }\mathbb{E}\left[\int_{t}^{T}\left|f_{s}\right|^{2}ds\right]<\infty\\
L_{\mathcal{G}}^{2}\left(\Omega;\mathbb{R}^{l}\right): & \mbox{the set of random variables \ensuremath{\xi:\left(\Omega,\mathcal{G}\right)\rightarrow\left(\mathbb{R}^{l},\mathcal{B}\left(\mathbb{R}^{l}\right)\right)}}\\
 & \mbox{with }\mathbb{E}\left[\left|\xi\right|^{2}\right]<\infty\\
L_{\mathcal{F}}^{p}\left(\Omega;C\left(t,T;\mathbb{R}^{l}\right)\right): & \mbox{the set of continuous }\mbox{\ensuremath{\left\{  \mathcal{F}_{s}\right\} } }_{s\in\left[t,T\right]}\mbox{-adapted processes }\\
 & f=\left\{ f_{s}:t\leq s\leq T\right\} \mbox{ with }\mathbb{E}\left[\underset{s\in\left[t,T\right]}{\sup}\left|f_{s}\right|^{p}\right]<\infty\\
\left(W_{t}\right)_{t\in\left[0,T\right]}: & \left(W_{t}^{1},\cdots W_{t}^{d}\right)_{t\in\left[0,T\right]}\mbox{a d-dimensional Brownian motion on \ensuremath{\left(\Omega,\mathcal{F},\mathbb{P}\right)}}\\
\mathbb{E}_{t}\left[\cdot\right]: & \mathbb{E}\left[\cdot\left|\mathcal{F}_{t}\right.\right]
\end{cases}
\end{align*}
For any $t\in\left[0,T\right)$ , we consider a pair of portfolio
and wealth process $\left(\pi,X\right)$ satisfying the following
equation
\begin{equation}
\begin{cases}
dX_{s}=\left[r_{s}X_{s}+\pi_{s}^{T}\left(\mu_{s}^{x}-r_{s}\mathbf{1}\right)\right]ds+\pi_{s}^{T}\sigma_{s}dW_{s} & s\in\left[t,T\right]\\
X_{t}=x_{t}
\end{cases}
\end{equation}
where $\mathbf{1}$ is a $d$-dimensional vector of ones, $W$ is
a $d$-dimensional standard Brownian motion, $r$ is the risk free
rate process, $\mu^{x}$ is the drift rate vector process of risky
assets and $\sigma$ is the volatility process of risky assets which
is a $d\times d$ matrix and is assumed to be invertible. Throughout
this paper we assume that $r\in L_{\mathcal{F}}^{\infty}\left(0,T;\mathbb{R}\right)$,
$\mu^{x}\in L_{\mathcal{F}}^{\infty}\left(0,T;\mathbb{R}^{d}\right)$,
$\sigma\in L_{\mathcal{F}}^{\infty}\left(0,T;\mathbb{R}^{d\times d}\right)$
and $\sigma^{-1}\in L_{\mathcal{F}}^{\infty}\left(0,T;\mathbb{R}^{d\times d}\right)$,
which means they are all bounded, in addition, we also assume that
$r$ and $\sigma$ are always deterministic. Actually except in the
case when we use $\frac{x^{2}}{2}$ as our utility function, we would
also assume $\mu^{x}$ to be deterministic. Then the above dynamics
of the wealth process can be written as 
\begin{equation}
\begin{cases}
dX_{s}=\left(r_{s}X_{s}+\pi_{s}^{T}\sigma_{s}\theta_{s}\right)ds+\pi_{s}^{T}\sigma_{s}dW_{s} & s\in\left[t,T\right]\\
X_{t}=x_{t}
\end{cases}
\end{equation}
where $\theta=\sigma^{-1}\left(\mu^{x}-r\mathbf{1}\right)$ is the
market price of risk and it is clear that $\theta\in L_{\mathcal{F}}^{\infty}\left(0,T;\mathbb{R}^{d}\right)$
which means $\theta$ is also bounded.

\subsection{Motivation of our problems}

\paragraph{\textmd{A standard static mean variance portfolio selection problem
with a fixed mean target $l$ is }}

\begin{align}
\underset{\pi}{\min} & \mbox{ \ensuremath{}\ensuremath{}\ensuremath{}\ensuremath{}\mbox{ \ensuremath{}\ensuremath{}\ensuremath{}\ensuremath{}}\mbox{ \ensuremath{}\ensuremath{}\ensuremath{}\ensuremath{}}}\frac{1}{2}Var(X_{T})\nonumber \\
s.t. & \mbox{ \ensuremath{}\ensuremath{}\ensuremath{}\ensuremath{}\mbox{ \ensuremath{}\ensuremath{}\ensuremath{}\ensuremath{}}}\begin{cases}
dX_{s}=\left(r_{s}X_{s}+\pi_{s}^{T}\sigma_{s}\theta_{s}\right)ds+\pi_{s}^{T}\sigma_{s}dW_{s}\\
X_{0}=x_{0}\\
\mathbb{E}\left[X_{T}\right]=l\\
\pi\in U_{ad}^{\pi}=\left\{ \pi\left|\right.\pi\in L_{\mathcal{F}}^{2}\left(0,T;\mathbb{R}^{d}\right)\right\} 
\end{cases}\label{eq:2.3-1}
\end{align}
that is 
\begin{align}
\underset{\pi}{\min} & \mbox{ \ensuremath{}\ensuremath{}\ensuremath{}\ensuremath{}\mbox{ \ensuremath{}\ensuremath{}\ensuremath{}\ensuremath{}}\mbox{ \ensuremath{}\ensuremath{}\ensuremath{}\ensuremath{}}}\frac{1}{2}\mathbb{E}\left[(X_{T}-\mathbb{E}X_{T})^{2}\right]\nonumber \\
s.t. & \mbox{ \ensuremath{}\ensuremath{}\ensuremath{}\ensuremath{}\mbox{ \ensuremath{}\ensuremath{}\ensuremath{}\ensuremath{}}}\begin{cases}
dX_{s}=\left(r_{s}X_{s}+\pi_{s}^{T}\sigma_{s}\theta_{s}\right)ds+\pi_{s}^{T}\sigma_{s}dW_{s}\\
X_{0}=x_{0}\\
\mathbb{E}\left[X_{T}\right]=l\\
\pi\in U_{ad}^{\pi}=\left\{ \pi\left|\right.\pi\in L_{\mathcal{F}}^{2}\left(0,T;\mathbb{R}^{d}\right)\right\} 
\end{cases}\label{eq:2.3-1-2}
\end{align}
which can be written as 
\begin{align}
\underset{\pi}{\min} & \mbox{ \ensuremath{}\ensuremath{}\ensuremath{}\ensuremath{}\mbox{ \ensuremath{}\ensuremath{}\ensuremath{}\ensuremath{}}\mbox{ \ensuremath{}\ensuremath{}\ensuremath{}\ensuremath{}}}E\left[h(X_{T}-l)\right]\nonumber \\
s.t. & \mbox{ \ensuremath{}\ensuremath{}\ensuremath{}\ensuremath{}\mbox{ \ensuremath{}\ensuremath{}\ensuremath{}\ensuremath{}}}\begin{cases}
dX_{s}=\left(r_{s}X_{s}+\pi_{s}^{T}\sigma_{s}\theta_{s}\right)ds+\pi_{s}^{T}\sigma_{s}dW_{s}\\
X_{0}=x_{0}\\
\mathbb{E}\left[X_{T}\right]=l\\
\pi\in U_{ad}^{\pi}=\left\{ \pi\left|\right.\pi\in L_{\mathcal{F}}^{2}\left(0,T;\mathbb{R}^{d}\right)\right\} 
\end{cases}\label{eq:2.3-1-1}
\end{align}
where $h\left(x\right)=\frac{x^{2}}{2}$. \\

\noindent In this paper we want to extend the above problem in the
following two ways. Firstly, we want to extend $h(x)$ to other kinds
of utility functions not just the function of $x^{2}$. Secondly,
we want to extend the static mean target $\mathbb{E}\left[X_{T}\right]=l$
to the moving target $\mathbb{E}_{t}\left[X_{T}\right]=X_{t}e^{\int_{t}^{T}\mu_{s}ds}$
for any $t\in\left[0,T\right)$ where $\mu$ is our required return
process which is assumed to be bounded and deterministic. Based on
these two motivations, for any given utility function $h(x)$, our
aim is to find a portfolio $\pi$ in order to minimize
\begin{equation}
J\left(t,X_{t};\pi\right)=\mathbb{E}_{t}\left[h\left(X_{T}-\mathbb{E}_{t}\left[X_{T}\right]\right)\right]
\end{equation}
at any time $t\in\left[0,T\right)$ with $X_{t}=x_{t}$. And in order
to satisfy the moving target requirement, we want to have a constraint
on our control $\pi$ s.t. 
\begin{equation}
\mathbb{E}_{t}\left[X_{T}\right]=X_{t}e^{\int_{t}^{T}\mu_{s}ds},\forall t\in\left[0,T\right)
\end{equation}
thus our objective at $t\in\left[0,T\right)$ with $X_{t}=x_{t}$
is to find a control $\pi$ under the above constraint so as to minimize
\begin{equation}
J\left(t,X_{t};\pi\right)=\mathbb{E}_{t}\left[h\left(X_{T}-X_{t}e^{\int_{t}^{T}\mu_{s}ds}\right)\right]
\end{equation}
which means we want to solve a family of the following problems for
any $t\in\left[0,T\right)$

\begin{align}
\underset{\pi}{\min} & \mbox{ \ensuremath{}\ensuremath{}\ensuremath{}\ensuremath{}\mbox{ \ensuremath{}\ensuremath{}\ensuremath{}\ensuremath{}}\mbox{ \ensuremath{}\ensuremath{}\ensuremath{}\ensuremath{}}}\mathbb{E}_{t}\left[h\left(X_{T}-X_{t}e^{\int_{t}^{T}\mu_{s}ds}\right)\right]\nonumber \\
s.t. & \mbox{ \ensuremath{}\ensuremath{}\ensuremath{}\ensuremath{}\mbox{ \ensuremath{}\ensuremath{}\ensuremath{}\ensuremath{}}}\begin{cases}
dX_{s}=\left(r_{s}X_{s}+\pi_{s}^{T}\sigma_{s}\theta_{s}\right)ds+\pi_{s}^{T}\sigma_{s}dW_{s},s\in\left[t,T\right]\\
X_{t}=x_{t}\\
\pi\in\overline{U_{ad}^{\pi}}=\left\{ \pi\left|\right.\pi\in L_{\mathcal{F}}^{2}\left(0,T;\mathbb{R}^{d}\right)\mbox{and }\mathbb{E}_{t}\left[X_{T}\right]=X_{t}e^{\int_{t}^{T}\mu_{s}ds},\forall t\in\left[0,T\right)\right\} 
\end{cases}\label{eq:2.102-2-1}
\end{align}
According to Bjork and Murgoci\cite{key-3}, if the terminal evaluation
function $h(x)$ in our above kind family of problems depends on $X_{t}=x_{t}$,
then this $X_{t}$ will cause time inconsistency if it can not be
factorized outside the given utility function $h\left(x\right)$.
Here our utility function $h\left(x\right)$ depends on the term $X_{t}e^{\int_{t}^{T}\mu_{s}ds}$
and thus our family of problems (\ref{eq:2.102-2-1}) is time inconsistent
in most cases.\\

\noindent When we use the utility function $h\left(x\right)=e^{-\alpha x}$
for some $\alpha>0$, the objective of our family of problems (\ref{eq:2.102-2-1})
is 
\begin{eqnarray*}
J\left(t,X_{t};\pi\right) & = & \mathbb{E}_{t}\left[h\left(X_{T}-X_{t}e^{\int_{t}^{T}\mu_{s}ds}\right)\right]\\
 & = & \mathbb{E}_{t}\left[e^{-\alpha\left(X_{T}-X_{t}e^{\int_{t}^{T}\mu_{s}ds}\right)}\right]\\
 & = & e^{\alpha X_{t}e^{\int_{t}^{T}\mu_{s}ds}}\mathbb{E}_{t}\left[e^{-\alpha X_{T}}\right]
\end{eqnarray*}
which implies that our family of problems is equivalent to a standard
control problem with objective 
\[
\hat{J}\left(t,X_{t};\pi\right)=\mathbb{E}_{t}\left[e^{-\alpha X_{T}}\right]
\]
which becomes a time consistent problem and can be solved by dynamic
programming when $X$ is a Markov process. But for power and $x^{-}$
utility functions time inconsistency would be caused and thus this
paper focuses on the following utility functions $h\left(x\right)=\frac{x^{2}}{2},-\frac{x^{3}}{3},\frac{x^{4}}{4},x^{-}$
by using which our family of problems (\ref{eq:2.102-2-1}) becomes
a time inconsistent problem.

\subsection{Definition of equilibrium for our power utility functions}

In terms of time inconsistency, the notion ``optimality'' needs
to be defined in an appropriate way. Here we adopt the concept of
the equilibrium control which is optimal only for spike variation
in an infinitesimal way for any $t\in\left[0,T\right)$.\\

\noindent Given a control $u^{*}$, for any $t\in\left[0,T\right)$,
$\varepsilon>0$ and $v\in L_{\mathcal{F}_{t}}^{2}\left(\Omega;\mathbb{R}^{d}\right)$,
define 

\begin{equation}
u_{s}^{t,\varepsilon,v}=u_{s}^{*}+v\mathbf{1}_{s\in\left[t,t+\varepsilon\right)},\mbox{ }s\in\left[t,T\right]\label{eq:1.1}
\end{equation}

\begin{defn}
\label{2.1}Let $u^{*}\in U_{ad}$ be a given control with $U_{ad}$
being the set of admissible controls. Let $X^{*}$ be the state process
corresponding to $u^{*}$$ $. The control $u^{*}$ is called an equilibrium
if 

\begin{equation}
\underset{\varepsilon\left\downarrow 0\right.}{\lim}\frac{J\left(t,X_{t}^{*};u^{t,\varepsilon,v}\right)-J\left(t,X_{t}^{*};u^{*}\right)}{\varepsilon}\geq0
\end{equation}
for any $t\in\left[0,T\right)$ and $v\in L_{\mathcal{F}_{t}}^{2}\left(\Omega;\mathbb{R}^{d}\right)$
s.t. $u^{t,\varepsilon,v}\in U_{ad}$, where $u^{t,\varepsilon,v}$
is defined by (\ref{eq:1.1}).\\

\end{defn}
\noindent Notice that here we use the definition of an equilibrium
control defined in Hu, Jin and Zhou\cite{key-1}, which is defined
in the class of open-loop controls and is different form the one in
Bjork and Murgoci\cite{key-3} where the definition is based on the
feedback controls. In this definition, the perturbation of the control
in $\left[t,t+\varepsilon\right)$ will not change the control process
in $\left[t+\varepsilon,T\right]$. We use this definition for our
power utility functions $h\left(x\right)=\frac{x^{2}}{2},-\frac{x^{3}}{3},\frac{x^{4}}{4}$
only. For $h(x)=x^{-}$ we will use another definition of equilibrium
which will be defined in that section.

\subsection{A sufficient equilibrium condition for $\mbox{C}^{2}$ utility functions}

In the following 3 sections when we use our power utility functions
in our problem, we will have our problem transformed into a family
of problems of the following form for any $t\in\left[0,T\right)$

\begin{align}
\underset{u}{\min} & \mbox{ \ensuremath{}\ensuremath{}\ensuremath{}\ensuremath{}\mbox{ \ensuremath{}\ensuremath{}\ensuremath{}\ensuremath{}}\mbox{ \ensuremath{}\ensuremath{}\ensuremath{}\ensuremath{}}}\mathbb{E}_{t}\left[h\left(Y_{T}-Y_{t}\right)\right]\nonumber \\
s.t. & \mbox{ \ensuremath{}\ensuremath{}\ensuremath{}\ensuremath{}\mbox{ \ensuremath{}\ensuremath{}\ensuremath{}\ensuremath{}}}\begin{cases}
dY_{s}=\left(\widehat{r_{s}}Y_{s}+u_{s}^{T}\theta_{s}\right)ds+u_{s}^{T}dW_{s}\\
Y_{t}=y_{t}\\
u\in U_{ad}^{u}=\left\{ u\left|\right.u\in L_{\mathcal{F}}^{2}\left(0,T;\mathbb{R}^{d}\right)\right\} 
\end{cases}\label{eq:1.3}
\end{align}
where $\theta$ is bounded, $\widehat{r}$ is a deterministic process,
and $h(x)$ is the given $C^{2}$ utility function. Here we have $J\left(t,Y_{t};u\right)=\mathbb{E}_{t}\left[h\left(Y_{T}-Y_{t}\right)\right]$.
For this family of problems, the definition \ref{2.1} for an equilibrium
control is translated into the following proposition. \\

\begin{prop}
\label{1.1}Let $u^{*}\in L_{\mathcal{F}}^{2}\left(0,T;\mathbb{R}^{d}\right)$\textup{
}be a given control and \textup{$Y^{*}$} be the state process corresponding
to \textup{$u^{*}$. }The control \textup{$u^{*}$ }is an equilibrium
for the family of\textup{ }problems\textup{ (\ref{eq:1.3}) }if

\begin{equation}
\underset{\varepsilon\left\downarrow 0\right.}{\lim}\frac{J\left(t,Y_{t}^{*};u^{t,\varepsilon,v}\right)-J\left(t,Y_{t}^{*};u^{*}\right)}{\varepsilon}\geq0
\end{equation}
 for any $t\in\left[0,T\right)$ and $v\in L_{\mathcal{F}_{t}}^{2}\left(\Omega;\mathbb{R}^{d}\right)$,
where\textup{ $u^{t,\varepsilon,v}$} is defined by (\ref{eq:1.1}).
\\

\end{prop}
\noindent Here we give a sufficient condition for equilibrium controls
when $Y$ is the process defined in (\ref{eq:1.3}) by using the the
second order Taylor expansion at any $t\in\left[0,T\right)$, which
is inspired by the idea used in Hu, Jin and Zhou\cite{key-1}. Let
$u^{*}$ be a fixed control and $Y^{*}$ be the corresponding state
process. For any $t\in\left[0,T\right)$, we define in the time interval
$\left[t,T\right]$ the processes $p^{t}\in L_{\mathcal{F}}^{2}\left(t,T;\mathbb{R}\right)$,
$q^{t}\in L_{\mathcal{F}}^{2}\left(t,T;\mathbb{R}^{d}\right)$, $P^{t}\in L_{\mathcal{F}}^{2}\left(t,T;\mathbb{R}\right)$
and $Q^{t}\in L_{\mathcal{F}}^{2}\left(t,T;\mathbb{R}^{d}\right)$
which satisfy the following system of BSDEs
\begin{equation}
\begin{cases}
dp_{s}^{t} & =-\widehat{r_{s}}p_{s}^{t}ds+\left(q_{s}^{t}\right)^{T}dW_{s},\mbox{ \ensuremath{}s\ensuremath{\in\left[t,T\right]}}\\
p_{T}^{t} & =\frac{dh\left(Y_{T}^{*}-Y_{t}^{*}\right)}{dx}
\end{cases}\label{eq:1.4}
\end{equation}

\begin{equation}
\begin{cases}
dP_{s}^{t} & =-2\widehat{r_{s}}P_{s}^{t}ds+\left(Q_{s}^{t}\right)^{T}dW_{s},\mbox{ \ensuremath{}s\ensuremath{\in\left[t,T\right]}}\\
P_{T}^{t} & =\frac{d^{2}h\left(Y_{T}^{*}-Y_{t}^{*}\right)}{dx^{2}}
\end{cases}\label{eq:1.5}
\end{equation}
\\

\begin{prop}
\label{prop:1.2}For any $\varepsilon>0$, $t\in\left[0,T\right)$,
$v\in L_{\mathcal{F}_{t}}^{2}\left(\Omega;\mathbb{R}^{d}\right)$
and $u^{t,\varepsilon,v}$ defined by (\ref{eq:1.1}). We have that

\begin{equation}
J\left(t,Y_{t}^{*};u^{t,\varepsilon,v}\right)-J\left(t,Y_{t}^{*};u^{*}\right)=\mathbb{E}_{t}\int_{t}^{t+\varepsilon}v^{T}\Lambda_{s}^{t}+\frac{1}{2}P_{s}^{t}\left|v\right|^{2}ds+o\left(\varepsilon\right)\label{eq:1.6}
\end{equation}
where \textup{$\Lambda_{s}^{t}=p_{s}^{t}\theta_{s}+q_{s}^{t}$ for
any $s\in\left[t,T\right]$.}\end{prop}
\begin{proof}
Here we use the standard perturbation approach in Yong and Zhou\cite{key-2}.
Let $Y^{t,\varepsilon,v}$ be the state process corresponding to $u^{t,\varepsilon,v}$,
we have that 

\[
Y_{s}^{t,\varepsilon,v}=Y_{s}^{*}+I_{s}^{t,\varepsilon,v}+Z_{s}^{t,\varepsilon,v},\mbox{\ensuremath{}s\ensuremath{\in\left[t,T\right]}}
\]
where $I\equiv I_{s}^{t,\varepsilon,v}$ and $Z\equiv Z_{s}^{t,\varepsilon,v}$
satisfy that 
\[
\begin{cases}
dI_{s} & =\widehat{r_{s}}I_{s}ds+v^{T}\mathbf{1}_{s\in\left[t,t+\varepsilon\right)}dW_{s}\\
I_{t} & =0
\end{cases}
\]
\[
\begin{cases}
dZ_{s} & =\left[\widehat{r_{s}}Z_{s}+v^{T}\theta_{s}\mathbf{1}_{s\in\left[t,t+\varepsilon\right)}\right]ds\\
Z_{t} & =0
\end{cases}
\]
and we have that
\[
\mathbb{E}_{t}\left[\underset{s\in\left[t,T\right]}{\sup}\left|I_{s}\right|^{2}\right]=O\left(\varepsilon\right),\mbox{ \ensuremath{\mathbb{E}_{t}\left[\underset{s\in\left[t,T\right]}{\sup}\left|Z_{s}\right|^{2}\right]}\ensuremath{=O\ensuremath{\left(\varepsilon^{2}\right)}}}
\]
which implies that 
\begin{eqnarray}
\mathbb{E}_{t}\left[\underset{s\in\left[t,T\right]}{\sup}\left|I_{s}+Z_{s}\right|^{2}\right] & \leq & \mathbb{E}_{t}\left[\underset{s\in\left[t,T\right]}{\sup}\left(\left|I_{s}\right|+\left|Z_{s}\right|\right)^{2}\right]\nonumber \\
 & \leq & 2\mathbb{E}_{t}\left[\underset{s\in\left[t,T\right]}{\sup}\left(\left|I_{s}\right|^{2}+\left|Z_{s}\right|^{2}\right)\right]\nonumber \\
 & \leq & 2\left(\mathbb{E}_{t}\left[\underset{s\in\left[t,T\right]}{\sup}\left|I_{s}\right|^{2}\right]+\mathbb{E}_{t}\left[\underset{s\in\left[t,T\right]}{\sup}\left|Z_{s}\right|^{2}\right]\right)\nonumber \\
 & = & O\left(\varepsilon\right)\label{eq:3.6}
\end{eqnarray}
then we have that
\begin{align}
 & J\left(t,Y_{t}^{*};u^{t,\varepsilon,v}\right)-J\left(t,Y_{t}^{*};u^{*}\right)\nonumber \\
= & \mathbb{E}_{t}\left[h\left(Y_{T}^{t,\varepsilon,v}-Y_{t}^{*}\right)-h\left(Y_{T}^{*}-Y_{t}^{*}\right)\right]\nonumber \\
= & \mathbb{E}_{t}\left[\frac{dh\left(Y_{T}^{*}-Y_{t}^{*}\right)}{dx}\left(Y_{T}^{t,\varepsilon,v}-Y_{T}^{*}\right)+\frac{1}{2}\frac{d^{2}h\left(Y_{T}^{*}-Y_{t}^{*}\right)}{dx^{2}}\left(Y_{T}^{t,\varepsilon,v}-Y_{T}^{*}\right)^{2}\right]+o\left(\mathbb{E}_{t}\left[\left(Y_{T}^{t,\varepsilon,v}-Y_{T}^{*}\right)^{2}\right]\right)\nonumber \\
= & \mathbb{E}_{t}\left[p_{T}^{t}\left(I_{T}+Z_{T}\right)\right]+\frac{1}{2}\mathbb{E}_{t}\left[P_{T}^{t}\left(I_{T}+Z_{T}\right)^{2}\right]+o\left(\mathbb{E}_{t}\left[\left(I_{T}+Z_{T}\right)^{2}\right]\right)\nonumber \\
= & \mathbb{E}_{t}\left[p_{T}^{t}\left(I_{T}+Z_{T}\right)\right]+\frac{1}{2}\mathbb{E}_{t}\left[P_{T}^{t}\left(I_{T}+Z_{T}\right)^{2}\right]+o\left(\varepsilon\right)\label{eq:1.7}
\end{align}
because by (\ref{eq:3.6}) we have 
\[
\mathbb{E}_{t}\left[\left(I_{T}+Z_{T}\right)^{2}\right]\leq\mathbb{E}_{t}\left[\underset{s\in\left[t,T\right]}{\sup}\left|I_{s}+Z_{s}\right|^{2}\right]\leq O\left(\varepsilon\right)
\]

\noindent Since we have that
\[
\begin{cases}
d\left(I_{s}+Z_{s}\right) & =\left[\widehat{r_{s}}\left(I_{s}+Z_{s}\right)+v^{T}\theta_{s}\mathbf{1}_{s\in\left[t,t+\varepsilon\right)}\right]ds+v^{T}\mathbf{1}_{s\in\left[t,t+\varepsilon\right)}dW_{s}\\
I_{t}+Z_{t} & =0
\end{cases}
\]
we could calculate that
\begin{align*}
 & d\left[p_{s}^{t}\left(I_{s}+Z_{s}\right)\right]\\
= & p_{s}^{t}d\left(I_{s}+Z_{s}\right)+\left(I_{s}+Z_{s}\right)dp_{s}^{t}+d\left\langle p^{t},I+Z\right\rangle _{s}\\
= & p_{s}^{t}\left[\widehat{r_{s}}\left(I_{s}+Z_{s}\right)+v^{T}\theta_{s}\mathbf{1}_{s\in\left[t,t+\varepsilon\right)}\right]ds-\widehat{r_{s}}p_{s}^{t}\left(I_{s}+Z_{s}\right)ds+v^{T}q_{s}^{t}\mathbf{1}_{s\in\left[t,t+\varepsilon\right)}ds\\
 & +\left[p_{s}^{t}v^{T}\mathbf{1}_{s\in\left[t,t+\varepsilon\right)}+\left(I_{s}+Z_{s}\right)\left(q_{s}^{t}\right)^{T}\right]dW_{s}\\
= & \left[p_{s}^{t}v^{T}\theta_{s}\mathbf{1}_{s\in\left[t,t+\varepsilon\right)}+v^{T}q_{s}^{t}\mathbf{1}_{s\in\left[t,t+\varepsilon\right)}\right]ds+\left[\cdots\right]dW_{s}
\end{align*}
thus
\begin{align}
\mathbb{E}_{t}\left[p_{T}^{t}\left(I_{T}+Z_{T}\right)\right] & =\mathbb{E}_{t}\int_{t}^{T}p_{s}^{t}v^{T}\theta_{s}\mathbf{1}_{s\in\left[t,t+\varepsilon\right)}+v^{T}q_{s}^{t}\mathbf{1}_{s\in\left[t,t+\varepsilon\right)}ds\nonumber \\
 & =\mathbb{E}_{t}\int_{t}^{t+\varepsilon}v^{T}\left[p_{s}^{t}\theta_{s}+q_{s}^{t}\right]ds\nonumber \\
 & =\mathbb{E}_{t}\int_{t}^{t+\varepsilon}v^{T}\Lambda_{s}^{t}ds\label{eq:1.8}
\end{align}
and that
\begin{align*}
 & d\left[P_{s}^{t}\left(I_{s}+Z_{s}\right)^{2}\right]\\
= & P_{s}^{t}d\left(I_{s}+Z_{s}\right)^{2}+\left(I_{s}+Z_{s}\right)^{2}dP_{s}^{t}+d\left\langle P^{t},\left(I+Z\right)^{2}\right\rangle _{s}\\
= & P_{s}^{t}\left[2\left(I_{s}+Z_{s}\right)d\left(I_{s}+Z_{s}\right)+d\left\langle I+Z\right\rangle _{s}\right]+\left(I_{s}+Z_{s}\right)^{2}\left[-2\widehat{r_{s}}P_{s}^{t}ds+\left(Q_{s}^{t}\right)^{T}dW_{s}\right]\\
 & +2\left(I_{s}+Z_{s}\right)d\left\langle P^{t},I+Z\right\rangle _{s}\\
= & 2P_{s}^{t}\left(I_{s}+Z_{s}\right)\left[\widehat{r_{s}}\left(I_{s}+Z_{s}\right)+v^{T}\theta_{s}\mathbf{1}_{s\in\left[t,t+\varepsilon\right)}\right]ds+P_{s}^{t}v^{T}v\mathbf{1}_{s\in\left[t,t+\varepsilon\right)}ds\\
 & -2\widehat{r_{s}}P_{s}^{t}\left(I_{s}+Z_{s}\right)^{2}ds+2\left(I_{s}+Z_{s}\right)v^{T}Q_{s}^{t}\mathbf{1}_{s\in\left[t,t+\varepsilon\right)}ds+\left[\cdots\right]dW_{s}\\
= & 2\left(I_{s}+Z_{s}\right)v^{T}\left[P_{s}^{t}\theta_{s}+Q_{s}^{t}\right]\mathbf{1}_{s\in\left[t,t+\varepsilon\right)}ds+P_{s}^{t}v^{T}v\mathbf{1}_{s\in\left[t,t+\varepsilon\right)}ds+\left[\cdots\right]dW_{s}
\end{align*}
thus
\begin{eqnarray}
\mathbb{E}_{t}\left[P_{T}^{t}\left(I_{T}+Z_{T}\right)^{2}\right] & = & \mathbb{E}_{t}\int_{t}^{T}P_{s}^{t}v^{T}v\mathbf{1}_{s\in\left[t,t+\varepsilon\right)}ds+\mathbb{E}_{t}\int_{t}^{T}2\left(I_{s}+Z_{s}\right)v^{T}\left[P_{s}^{t}\theta_{s}+Q_{s}^{t}\right]\mathbf{1}_{s\in\left[t,t+\varepsilon\right)}ds\nonumber \\
 & = & \mathbb{E}_{t}\int_{t}^{t+\varepsilon}P_{s}^{t}v^{T}vds+\mathbb{E}_{t}\int_{t}^{t+\varepsilon}2\left(I_{s}+Z_{s}\right)v^{T}\left[P_{s}^{t}\theta_{s}+Q_{s}^{t}\right]ds\nonumber \\
 & \overset{(\ref{eq:1.10})}{=} & \mathbb{E}_{t}\int_{t}^{t+\varepsilon}P_{s}^{t}\left|v\right|^{2}ds+o\left(\varepsilon\right)\label{eq:1.9}
\end{eqnarray}
because we have 
\begin{align}
\mathbb{E}_{t}\int_{t}^{t+\varepsilon}2\left(I_{s}+Z_{s}\right)v^{T}\left[P_{s}^{t}\theta_{s}+Q_{s}^{t}\right]ds & \leq2\mathbb{E}_{t}\int_{t}^{t+\varepsilon}\left(\underset{u\in\left[t,T\right]}{\sup}\left|I_{u}+Z_{u}\right|\right)\left|v^{T}\left(P_{s}^{t}\theta_{s}+Q_{s}^{t}\right)\right|ds\nonumber \\
 & =2\mathbb{E}_{t}\left[\left(\underset{u\in\left[t,T\right]}{\sup}\left|I_{u}+Z_{u}\right|\right)\int_{t}^{t+\varepsilon}\left|v^{T}\left(P_{s}^{t}\theta_{s}+Q_{s}^{t}\right)\right|ds\right]\nonumber \\
 & \leq2\sqrt{\mathbb{E}_{t}\left[\left(\underset{u\in\left[t,T\right]}{\sup}\left|I_{u}+Z_{u}\right|\right)^{2}\right]}\sqrt{\mathbb{E}_{t}\left[\left(\int_{t}^{t+\varepsilon}\left|v^{T}\left(P_{s}^{t}\theta_{s}+Q_{s}^{t}\right)\right|ds\right)^{2}\right]}\nonumber \\
 & \mbox{ \ensuremath{}\ensuremath{}\ensuremath{}\ensuremath{}\ensuremath{}}\mbox{ \ensuremath{}\ensuremath{}\ensuremath{}\ensuremath{}\ensuremath{}}\mbox{ \ensuremath{}\ensuremath{}\ensuremath{}\ensuremath{}\ensuremath{}}\mbox{ \ensuremath{}\ensuremath{}\ensuremath{}\ensuremath{}\ensuremath{}}\mbox{\mbox{by Cauchy-Schwarz}}\nonumber \\
 & =2\sqrt{\mathbb{E}_{t}\left[\underset{u\in\left[t,T\right]}{\sup}\left|I_{u}+Z_{u}\right|^{2}\right]}\sqrt{\mathbb{E}_{t}\left[\left(\int_{t}^{t+\varepsilon}\left|v^{T}\left(P_{s}^{t}\theta_{s}+Q_{s}^{t}\right)\right|ds\right)^{2}\right]}\nonumber \\
 & \leq\sqrt{O\left(\varepsilon\right)}\sqrt{\mathbb{E}_{t}\left[\left(\int_{t}^{t+\varepsilon}\left|v\right|\left|P_{s}^{t}\theta_{s}+Q_{s}^{t}\right|ds\right)^{2}\right]}\nonumber \\
 & \leq\sqrt{O\left(\varepsilon\right)}\sqrt{\varepsilon\mathbb{E}_{t}\left[\left|v\right|^{2}\int_{t}^{t+\varepsilon}\left|P_{s}^{t}\theta_{s}+Q_{s}^{t}\right|^{2}ds\right],}\mbox{by Cauchy-Schwarz}\nonumber \\
 & =O\left(\varepsilon\right)\sqrt{\mathbb{E}_{t}\left[\left|v\right|^{2}\int_{t}^{t+\varepsilon}\left|P_{s}^{t}\theta_{s}+Q_{s}^{t}\right|^{2}ds\right]}\nonumber \\
 & =O\left(\varepsilon\right)\sqrt{\left|v\right|^{2}\int_{t}^{t+\varepsilon}\mathbb{E}_{t}\left[\left|P_{s}^{t}\theta_{s}+Q_{s}^{t}\right|^{2}\right]ds}\nonumber \\
 & =o\left(\varepsilon\right)\label{eq:1.10}
\end{align}
as we have 
\[
\underset{\varepsilon\left\downarrow 0\right.}{\lim}\sqrt{\left|v\right|^{2}\int_{t}^{t+\varepsilon}\mathbb{E}_{t}\left[\left|P_{s}^{t}\theta_{s}+Q_{s}^{t}\right|^{2}\right]ds}=0
\]
which is implied by the above assumptions that $v\in L_{\mathcal{F}_{t}}^{2}\left(\Omega;\mathbb{R}^{d}\right)$,
$P^{t}\in L_{\mathcal{F}}^{2}\left(t,T;\mathbb{R}\right)$, $Q^{t}\in L_{\mathcal{F}}^{2}\left(t,T;\mathbb{R}^{d}\right)$
and $\theta$ is bounded. Then plug (\ref{eq:1.8}) and (\ref{eq:1.9})
into (\ref{eq:1.7}) we get (\ref{eq:1.6}).\\
\end{proof}
\begin{thm}
\label{thm:1.3}If the following system of equations
\begin{equation}
\begin{cases}
dY_{s} & =\left(\widehat{r_{s}}Y+\theta_{s}^{T}u_{s}\right)ds+u_{s}^{T}dW_{s},\mbox{ \ensuremath{s\in\left[0,T\right]}}\\
Y_{0} & =y_{0}\\
dp_{s}^{t} & =-\widehat{r_{s}}p_{s}^{t}ds+\left(q_{s}^{t}\right)^{T}dW_{s},\mbox{ \ensuremath{}s\ensuremath{\in\left[t,T\right]}}\\
p_{T}^{t} & =\frac{dh\left(Y_{T}^{*}-Y_{t}^{*}\right)}{dx}
\end{cases}
\end{equation}
admits a solution $\left(u^{*},Y^{*},p^{t},q^{t}\right)$ for any
$t\in\left[0,T\right)$, s.t.
\begin{equation}
\begin{cases}
u^{*}\in L_{\mathcal{F}}^{2}\left(0,T;\mathbb{R}^{d}\right)\\
\mathbb{E}_{t}\int_{t}^{T}\left|\Lambda_{s}^{t}\right|ds<\infty\mbox{ and \ensuremath{\underset{s\left\downarrow t\right.}{\lim}\mathbb{E}_{t}\left[\Lambda_{s}^{t}\right]=0}} & where\mbox{ }\Lambda_{s}^{t}=p_{s}^{t}\theta_{s}+q_{s}^{t}.\\
\mathbb{E}_{t}\left[\frac{d^{2}h\left(Y_{T}^{*}-Y_{t}^{*}\right)}{dx^{2}}\right]\geq0
\end{cases}\label{eq:1.11}
\end{equation}
then $u^{*}$ is an equilibrium control.\end{thm}
\begin{proof}
Suppose there exits solution $\left(u^{*},Y^{*},p^{t},q^{t}\right)$
which satisfies the above condition (\ref{eq:1.11}). Then it is given
by (\ref{eq:1.5}) that 
\[
\begin{cases}
dP_{s}^{t} & =-2\widehat{r_{s}}P_{s}^{t}ds+\left(Q_{s}^{t}\right)^{T}dW_{s},\mbox{ \ensuremath{}s\ensuremath{\in\left[t,T\right]}}\\
P_{T}^{t} & =\frac{d^{2}h\left(Y_{T}^{*}-Y_{t}^{*}\right)}{dx^{2}}
\end{cases}
\]
which means
\[
\begin{cases}
d\left(e^{\int_{0}^{s}2\widehat{r_{u}}du}P_{s}^{t}\right)=e^{\int_{0}^{s}2\widehat{r_{u}}du}\left(Q_{s}^{t}\right)^{T}dW_{s},\mbox{ \ensuremath{}s\ensuremath{\in\left[t,T\right]}}\\
e^{\int_{0}^{T}2\widehat{r_{u}}du}P_{T}^{t}=e^{\int_{0}^{T}2\widehat{r_{u}}du}\frac{d^{2}h\left(Y_{T}^{*}-Y_{t}^{*}\right)}{dx^{2}}
\end{cases}
\]
thus
\begin{align*}
 & e^{\int_{0}^{s}2\widehat{r_{u}}du}P_{s}^{t}=\mathbb{E}_{s}\left[e^{\int_{0}^{T}2\widehat{r_{u}}du}\frac{d^{2}h\left(Y_{T}^{*}-Y_{t}^{*}\right)}{dx^{2}}\right]\\
\Rightarrow & P_{s}^{t}\mathbb{=E}_{s}\left[e^{\int_{s}^{T}2\widehat{r_{u}}du}\frac{d^{2}h\left(Y_{T}^{*}-Y_{t}^{*}\right)}{dx^{2}}\right]
\end{align*}
so we get
\begin{eqnarray}
\mathbb{E}_{t}\left[P_{s}^{t}\right] & = & \mathbb{E}_{t}\left\{ \mathbb{E}_{s}\left[e^{\int_{s}^{T}2\widehat{r_{u}}du}\frac{d^{2}h\left(Y_{T}^{*}-Y_{t}^{*}\right)}{dx^{2}}\right]\right\} \nonumber \\
 & = & \mathbb{E}_{t}\left[e^{\int_{s}^{T}2\widehat{r_{u}}du}\frac{d^{2}h\left(Y_{T}^{*}-Y_{t}^{*}\right)}{dx^{2}}\right]\nonumber \\
 & = & e^{\int_{s}^{T}2\widehat{r_{u}}du}\mathbb{E}_{t}\left[\frac{d^{2}h\left(Y_{T}^{*}-Y_{t}^{*}\right)}{dx^{2}}\right]\nonumber \\
 & \overset{(\ref{eq:1.11})}{\geqslant} & 0\label{eq:1.12}
\end{eqnarray}
as $\widehat{r}$ is a deterministic and $\mathbb{E}_{t}\left[\frac{d^{2}h\left(Y_{T}^{*}-Y_{t}^{*}\right)}{dx^{2}}\right]\geq0$
by our assumption. Then for any $t\in\left[0,T\right)$ and $v\in L_{\mathcal{F}_{t}}^{2}\left(\Omega;\mathbb{R}^{d}\right)$
we have by proposition \ref{prop:1.2} that 
\begin{eqnarray*}
\underset{\varepsilon\left\downarrow 0\right.}{\lim}\frac{J\left(t,Y_{t}^{*};u^{t,\varepsilon,v}\right)-J\left(t,Y_{t}^{*};u^{*}\right)}{\varepsilon} & \overset{(\ref{eq:1.6})}{=} & \underset{\varepsilon\left\downarrow 0\right.}{\lim}\frac{\mathbb{E}_{t}\int_{t}^{t+\varepsilon}v^{T}\Lambda_{s}^{t}+\frac{1}{2}P_{s}^{t}\left|v\right|^{2}ds+o\left(\varepsilon\right)}{\varepsilon}\\
 & = & \underset{\varepsilon\left\downarrow 0\right.}{\lim}\frac{\int_{t}^{t+\varepsilon}v^{T}\mathbb{E}_{t}\left[\Lambda_{s}^{t}\right]+\frac{1}{2}\mathbb{E}_{t}\left[P_{s}^{t}\right]\left|v\right|^{2}ds}{\varepsilon}\\
 & \overset{(\ref{eq:1.12})}{\geq} & \underset{\varepsilon\left\downarrow 0\right.}{\lim}\frac{\int_{t}^{t+\varepsilon}v^{T}\mathbb{E}_{t}\left[\Lambda_{s}^{t}\right]ds}{\varepsilon}\\
 & \geq & v^{T}\left(\underset{\varepsilon\left\downarrow 0\right.}{\lim}\underset{s\in\left[t,t+\varepsilon\right]}{\inf}\mathbb{E}_{t}\left[\Lambda_{s}^{t}\right]\right)\\
\\
 & \overset{(\ref{eq:1.11})}{=} & v^{T}\left(\underset{\varepsilon\left\downarrow 0\right.}{\lim}\mathbb{E}_{t}\left[\Lambda_{t+\varepsilon}^{t}\right]\right)\\
 & = & 0
\end{eqnarray*}
as $\underset{s\left\downarrow t\right.}{\lim}\mathbb{E}_{t}\left[\Lambda_{s}^{t}\right]=0$
by our assumption, which proves that $u^{*}$ is an equilibrium control
by proposition \ref{1.1}.
\end{proof}
\newpage{}

\section{Utility function $h\left(x\right)=\frac{X^{2}}{2}$ for mean variance
portfolio selection}

\paragraph{\textmd{As being given in the section 2, a standard static mean variance
portfolio selection problem with a fixed mean target $l$ is}}

\begin{align}
\underset{\pi}{\min} & \mbox{ \ensuremath{}\ensuremath{}\ensuremath{}\ensuremath{}\mbox{ \ensuremath{}\ensuremath{}\ensuremath{}\ensuremath{}}\mbox{ \ensuremath{}\ensuremath{}\ensuremath{}\ensuremath{}}}Var(X_{T})\nonumber \\
s.t. & \mbox{ \ensuremath{}\ensuremath{}\ensuremath{}\ensuremath{}\mbox{ \ensuremath{}\ensuremath{}\ensuremath{}\ensuremath{}}}\begin{cases}
dX_{s}=\left(r_{s}X_{s}+\pi_{s}^{T}\sigma_{s}\theta_{s}\right)ds+\pi_{s}^{T}\sigma_{s}dW_{s}\\
X_{0}=x_{0}\\
\mathbb{E}\left[X_{T}\right]=l\\
\pi\in U_{ad}^{\pi}=\left\{ \pi\left|\right.\pi\in L_{\mathcal{F}}^{2}\left(0,T;\mathbb{R}^{d}\right)\right\} 
\end{cases}\label{eq:2.3}
\end{align}
for the constraint $\mathbb{E}\left[X_{T}\right]=l$, we could introduce
a Lagrangian multiplier $2\lambda$ and consider the following problem
\begin{align}
\underset{\pi}{\min} & \mbox{ \ensuremath{}\ensuremath{}\ensuremath{}\ensuremath{}\mbox{ \ensuremath{}\ensuremath{}\ensuremath{}\ensuremath{}}\mbox{ \ensuremath{}\ensuremath{}\ensuremath{}\ensuremath{}}}\mathbb{E}\left[\left(X_{T}-\lambda\right)^{2}\right]\nonumber \\
s.t. & \mbox{ \ensuremath{}\ensuremath{}\ensuremath{}\ensuremath{}\mbox{ \ensuremath{}\ensuremath{}\ensuremath{}\ensuremath{}}}\begin{cases}
dX_{s}=\left(r_{s}X_{s}+\pi_{s}^{T}\sigma_{s}\theta_{s}\right)ds+\pi_{s}^{T}\sigma_{s}dW_{s}\\
X_{0}=x_{0}\\
\pi\in U_{ad}^{\pi}=\left\{ \pi\left|\right.\pi\in L_{\mathcal{F}}^{2}\left(0,T;\mathbb{R}^{d}\right)\right\} 
\end{cases}\label{eq:2.4}
\end{align}
and if there exits $\lambda^{*}$ s.t. the optimal solution $\left(u^{*},X^{*}\right)$
to the above problem (\ref{eq:2.4}) satisfies $\mathbb{E}\left[X_{T}^{*}\right]=l$,
then $\left(u^{*},X^{*}\right)$ is also optimal for (\ref{eq:2.3}).\\

\noindent Here, as having been described in section 2, we consider
a moving target $\mathbb{E}_{t}\left[X_{T}\right]=X_{t}e^{\int_{t}^{T}\mu_{s}ds}$
for any $t\in\left[0,T\right)$, where $\mu$ is our required return
process which is assumed to be deterministic and bounded. Then the
mean variance portfolio selection with moving target that we want
to solve is a family of following problems

\begin{align}
\underset{\pi}{\min} & \mbox{ \ensuremath{}\ensuremath{}\ensuremath{}\ensuremath{}\mbox{ \ensuremath{}\ensuremath{}\ensuremath{}\ensuremath{}}\mbox{ \ensuremath{}\ensuremath{}\ensuremath{}\ensuremath{}}}\frac{1}{2}Var_{t}(X_{T})\nonumber \\
s.t. & \mbox{ \ensuremath{}\ensuremath{}\ensuremath{}\ensuremath{}\mbox{ \ensuremath{}\ensuremath{}\ensuremath{}\ensuremath{}}}\begin{cases}
dX_{s}=\left(r_{s}X_{s}+\pi_{s}^{T}\sigma_{s}\theta_{s}\right)ds+\pi_{s}^{T}\sigma_{s}dW_{s}\\
X_{t}=x_{t}\\
\pi\in\overline{U_{ad}^{\pi}}=\left\{ \pi\left|\right.\pi\in L_{\mathcal{F}}^{2}\left(0,T;\mathbb{R}^{d}\right)\mbox{and }\mathbb{E}_{t}\left[X_{T}\right]=X_{t}e^{\int_{t}^{T}\mu_{s}ds},\forall t\in\left[0,T\right)\right\} 
\end{cases}\label{eq:2.45}
\end{align}
which is the family of following problems

\begin{align}
\underset{\pi}{\min} & \mbox{ \ensuremath{}\ensuremath{}\ensuremath{}\ensuremath{}\mbox{ \ensuremath{}\ensuremath{}\ensuremath{}\ensuremath{}}\mbox{ \ensuremath{}\ensuremath{}\ensuremath{}\ensuremath{}}}\mathbb{E}_{t}\left[\frac{\left(X_{T}-X_{t}e^{\int_{t}^{T}\mu_{s}ds}\right)^{2}}{2}\right]\nonumber \\
s.t. & \mbox{ \ensuremath{}\ensuremath{}\ensuremath{}\ensuremath{}\mbox{ \ensuremath{}\ensuremath{}\ensuremath{}\ensuremath{}}}\begin{cases}
dX_{s}=\left(r_{s}X_{s}+\pi_{s}^{T}\sigma_{s}\theta_{s}\right)ds+\pi_{s}^{T}\sigma_{s}dW_{s}\\
X_{t}=x_{t}\\
\pi\in\overline{U_{ad}^{\pi}}=\left\{ \pi\left|\right.\pi\in L_{\mathcal{F}}^{2}\left(0,T;\mathbb{R}^{d}\right)\mbox{and }\mathbb{E}_{t}\left[X_{T}\right]=X_{t}e^{\int_{t}^{T}\mu_{s}ds},\forall t\in\left[0,T\right)\right\} 
\end{cases}\label{eq:2.5}
\end{align}
or we can write the family of problems as the following form 

\begin{align}
\underset{\pi}{\min} & \mbox{ \ensuremath{}\ensuremath{}\ensuremath{}\ensuremath{}\mbox{ \ensuremath{}\ensuremath{}\ensuremath{}\ensuremath{}}\mbox{ \ensuremath{}\ensuremath{}\ensuremath{}\ensuremath{}}}\mathbb{E}_{t}\left[h\left(X_{T}-X_{t}e^{\int_{t}^{T}\mu_{s}ds}\right)\right]\nonumber \\
s.t. & \mbox{ \ensuremath{}\ensuremath{}\ensuremath{}\ensuremath{}\mbox{ \ensuremath{}\ensuremath{}\ensuremath{}\ensuremath{}}}\begin{cases}
dX_{s}=\left(r_{s}X_{s}+\pi_{s}^{T}\sigma_{s}\theta_{s}\right)ds+\pi_{s}^{T}\sigma_{s}dW_{s}\\
X_{t}=x_{t}\\
\pi\in\overline{U_{ad}^{\pi}}=\left\{ \pi\left|\right.\pi\in L_{\mathcal{F}}^{2}\left(0,T;\mathbb{R}^{d}\right)\mbox{and }\mathbb{E}_{t}\left[X_{T}\right]=X_{t}e^{\int_{t}^{T}\mu_{s}ds},\forall t\in\left[0,T\right)\right\} 
\end{cases}\label{eq:2.102-2-1-1}
\end{align}
which is one of our problems being set in section 2 with $h(x)=\frac{X^{2}}{2}$

\subsection{Introducing a Lagrangian multiplier for our problem}

\noindent Inspired by the static Lagrangian multiplier method we introduce
a deterministic process $\lambda$ with $\int_{0}^{T}\left|\lambda_{s}\right|ds<\infty$
and consider a family of problems for any $t\in\left[0,T\right)$
as follows

\begin{align}
\underset{\pi}{\min} & \mbox{ \ensuremath{}\ensuremath{}\ensuremath{}\ensuremath{}\mbox{ \ensuremath{}\ensuremath{}\ensuremath{}\ensuremath{}}\mbox{ \ensuremath{}\ensuremath{}\ensuremath{}\ensuremath{}}}\mathbb{E}_{t}\left[\frac{\left(X_{T}-X_{t}e^{\int_{t}^{T}\lambda_{s}ds}\right)^{2}}{2}\right]\nonumber \\
s.t. & \mbox{ \ensuremath{}\ensuremath{}\ensuremath{}\ensuremath{}\mbox{ \ensuremath{}\ensuremath{}\ensuremath{}\ensuremath{}}}\begin{cases}
dX_{s}=\left(r_{s}X_{s}+\pi_{s}^{T}\sigma_{s}\theta_{s}\right)ds+\pi_{s}^{T}\sigma_{s}dW_{s}\\
X_{t}=x_{t}\\
\pi\in U_{ad}^{\pi}=\left\{ \pi\left|\right.\pi\in L_{\mathcal{F}}^{2}\left(0,T;\mathbb{R}^{d}\right)\right\} 
\end{cases}\label{eq:2.6}
\end{align}
\\

\begin{thm}
\label{thm:2.1}If there exits a deterministic process $\lambda$$^{*}$
with\textup{ $\int_{0}^{T}\left|\lambda_{s}^{*}\right|ds<\infty$}
s.t. an equilibrium solution $\left(\pi^{*},X^{*}\right)$ to the
above family of problems (\ref{eq:2.6}) having $\lambda^{*}$ as
parameter satisfies the condition that \textup{$\mathbb{E}_{t}\left[X_{T}^{*}\right]=X_{t}^{*}e^{\int_{t}^{T}\mu_{s}ds}$
}for any $t\in\left[0,T\right)$, then $\left(\pi^{*},X^{*}\right)$
is also equilibrium for the family of problems (\ref{eq:2.5}).\\
\end{thm}
\begin{proof}
Suppose $\left(\pi^{*},X^{*}\right)$ is an equilibrium solution to
the family of problems (\ref{eq:2.6}) having $\lambda^{*}$ as parameter
with $\int_{0}^{T}\left|\lambda_{s}^{*}\right|ds<\infty$ s.t. $\mathbb{E}_{t}\left[X_{T}^{*}\right]=X_{t}^{*}e^{\int_{t}^{T}\mu_{s}ds}$,
which says $\left(\pi^{*},X^{*}\right)$ is an equilibrium solution
to the family of following problems for any $t\in\left[0,T\right)$
\begin{align}
\underset{\pi}{\min} & \mbox{ \ensuremath{}\ensuremath{}\ensuremath{}\ensuremath{}\mbox{ \ensuremath{}\ensuremath{}\ensuremath{}\ensuremath{}}\mbox{ \ensuremath{}\ensuremath{}\ensuremath{}\ensuremath{}}}\mathbb{E}_{t}\left[\frac{\left(X_{T}-X_{t}e^{\int_{t}^{T}\lambda_{s}^{*}ds}\right)^{2}}{2}\right]\nonumber \\
s.t. & \mbox{ \ensuremath{}\ensuremath{}\ensuremath{}\ensuremath{}\mbox{ \ensuremath{}\ensuremath{}\ensuremath{}\ensuremath{}}}\begin{cases}
dX_{s}=\left(r_{s}X_{s}+\pi_{s}^{T}\sigma_{s}\theta_{s}\right)ds+\pi_{s}^{T}\sigma_{s}dW_{s}\\
X_{t}=x_{t}\\
\pi\in U_{ad}^{\pi}=\left\{ \pi\left|\right.\pi\in L_{\mathcal{F}}^{2}\left(0,T;\mathbb{R}^{d}\right)\right\} 
\end{cases}\label{eq:2.6-2}
\end{align}
Firstly, $\mathbb{E}_{t}\left[X_{T}^{*}\right]=X_{t}^{*}e^{\int_{t}^{T}\mu_{s}ds}$
implies that $ $$\pi^{*}\in\overline{U_{ad}^{\pi}}$ which can be
used with definition \ref{2.1} to prove that $\left(\pi^{*},X^{*}\right)$
is also equilibrium for the family of following problems for any $t\in\left[0,T\right)$
\begin{align}
\underset{\pi}{\min} & \mbox{ \ensuremath{}\ensuremath{}\ensuremath{}\ensuremath{}\mbox{ \ensuremath{}\ensuremath{}\ensuremath{}\ensuremath{}}\mbox{ \ensuremath{}\ensuremath{}\ensuremath{}\ensuremath{}}}\mathbb{E}_{t}\left[\frac{\left(X_{T}-X_{t}e^{\int_{t}^{T}\lambda_{s}^{*}ds}\right)^{2}}{2}\right]\nonumber \\
s.t. & \mbox{ \ensuremath{}\ensuremath{}\ensuremath{}\ensuremath{}\mbox{ \ensuremath{}\ensuremath{}\ensuremath{}\ensuremath{}}}\begin{cases}
dX_{s}=\left(r_{s}X_{s}+\pi_{s}^{T}\sigma_{s}\theta_{s}\right)ds+\pi_{s}^{T}\sigma_{s}dW_{s}\\
X_{t}=x_{t}\\
\pi\in\overline{U_{ad}^{\pi}}=\left\{ \pi\left|\right.\pi\in L_{\mathcal{F}}^{2}\left(0,T;\mathbb{R}^{d}\right)\mbox{and }\mathbb{E}_{t}\left[X_{T}\right]=X_{t}e^{\int_{t}^{T}\mu_{s}ds},\forall t\in\left[0,T\right)\right\} 
\end{cases}\label{eq:2.8}
\end{align}
Here we define $J\left(t,X_{t};\pi\right)$ for all problems in the
same way as before. Since the $J\left(t,X_{t}^{*};\pi^{t,\varepsilon,v}\right)-J\left(t,X_{t}^{*};\pi^{*}\right)$
term used in definition \ref{2.1} for problems (\ref{eq:2.8}) above
and the one used for problems (\ref{eq:2.9}) below are the same after
calculation, we deduce that the above family of problems is equivalent
to
\begin{align}
\underset{\pi}{\min} & \mbox{ \ensuremath{}\ensuremath{}\ensuremath{}\ensuremath{}\mbox{ \ensuremath{}\ensuremath{}\ensuremath{}\ensuremath{}}\mbox{ \ensuremath{}\ensuremath{}\ensuremath{}\ensuremath{}}}\frac{1}{2}\left\{ \mathbb{E}_{t}\left[X_{T}^{2}\right]-2X_{t}e^{\int_{t}^{T}\lambda_{s}^{*}ds}\mathbb{E}_{t}\left[X_{T}\right]\right\} +\frac{1}{2}X_{t}^{2}\left[2e^{\int_{t}^{T}\lambda_{s}^{*}ds}e^{\int_{t}^{T}\mu_{s}ds}-e^{2\int_{t}^{T}\mu_{s}ds}\right]\nonumber \\
s.t. & \mbox{ \ensuremath{}\ensuremath{}\ensuremath{}\ensuremath{}\mbox{ \ensuremath{}\ensuremath{}\ensuremath{}\ensuremath{}}}\begin{cases}
dX_{s}=\left(r_{s}X_{s}+\pi_{s}^{T}\sigma_{s}\theta_{s}\right)ds+\pi_{s}^{T}\sigma_{s}dW_{s}\\
X_{t}=x_{t}\\
\pi\in\overline{U_{ad}^{\pi}}=\left\{ \pi\left|\right.\pi\in L_{\mathcal{F}}^{2}\left(0,T;\mathbb{R}^{d}\right)\mbox{and }\mathbb{E}_{t}\left[X_{T}\right]=X_{t}e^{\int_{t}^{T}\mu_{s}ds},\forall t\in\left[0,T\right)\right\} 
\end{cases}\label{eq:2.9}
\end{align}
which after calculation is
\begin{align}
\underset{\pi}{\min} & \mbox{ \ensuremath{}\ensuremath{}\ensuremath{}\ensuremath{}\mbox{ \ensuremath{}\ensuremath{}\ensuremath{}\ensuremath{}}\mbox{ \ensuremath{}\ensuremath{}\ensuremath{}\ensuremath{}}}\frac{1}{2}\left\{ \mathbb{E}_{t}\left[X_{T}^{2}\right]-\left(X_{t}e^{\int_{t}^{T}\mu_{s}ds}\right)^{2}\right\} \nonumber \\
s.t. & \mbox{ \ensuremath{}\ensuremath{}\ensuremath{}\ensuremath{}\mbox{ \ensuremath{}\ensuremath{}\ensuremath{}\ensuremath{}}}\begin{cases}
dX_{s}=\left(r_{s}X_{s}+\pi_{s}^{T}\sigma_{s}\theta_{s}\right)ds+\pi_{s}^{T}\sigma_{s}dW_{s}\\
X_{t}=x_{t}\\
\pi\in\overline{U_{ad}^{\pi}}=\left\{ \pi\left|\right.\pi\in L_{\mathcal{F}}^{2}\left(0,T;\mathbb{R}^{d}\right)\mbox{and }\mathbb{E}_{t}\left[X_{T}\right]=X_{t}e^{\int_{t}^{T}\mu_{s}ds},\forall t\in\left[0,T\right)\right\} 
\end{cases}\label{eq:2.10}
\end{align}
which again by using the constraint $\mathbb{E}_{t}\left[X_{T}\right]=X_{t}e^{\int_{t}^{T}\mu_{s}ds}$
can be written as 
\begin{align}
\underset{\pi}{\min} & \mbox{ \ensuremath{}\ensuremath{}\ensuremath{}\ensuremath{}\mbox{ \ensuremath{}\ensuremath{}\ensuremath{}\ensuremath{}}\mbox{ \ensuremath{}\ensuremath{}\ensuremath{}\ensuremath{}}}\frac{1}{2}\mathbb{E}_{t}\left[\left(X_{T}-X_{t}e^{\int_{t}^{T}\mu_{s}ds}\right)^{2}\right]\nonumber \\
s.t. & \mbox{ \ensuremath{}\ensuremath{}\ensuremath{}\ensuremath{}\mbox{ \ensuremath{}\ensuremath{}\ensuremath{}\ensuremath{}}}\begin{cases}
dX_{s}=\left(r_{s}X_{s}+\pi_{s}^{T}\sigma_{s}\theta_{s}\right)ds+\pi_{s}^{T}\sigma_{s}dW_{s}\\
X_{t}=x_{t}\\
\pi\in\overline{U_{ad}^{\pi}}=\left\{ \pi\left|\right.\pi\in L_{\mathcal{F}}^{2}\left(0,T;\mathbb{R}^{d}\right)\mbox{and }\mathbb{E}_{t}\left[X_{T}\right]=X_{t}e^{\int_{t}^{T}\mu_{s}ds},\forall t\in\left[0,T\right)\right\} 
\end{cases}\label{eq:2.10-1}
\end{align}
thus the family of problems (\ref{eq:2.8}) is equivalent to the family
of following problems
\begin{align}
\underset{\pi}{\min} & \mbox{ \ensuremath{}\ensuremath{}\ensuremath{}\ensuremath{}\mbox{ \ensuremath{}\ensuremath{}\ensuremath{}\ensuremath{}}\mbox{ \ensuremath{}\ensuremath{}\ensuremath{}\ensuremath{}}}\mathbb{E}_{t}\left[\frac{\left(X_{T}-X_{t}e^{\int_{t}^{T}\mu_{s}ds}\right)^{2}}{2}\right]\nonumber \\
s.t. & \mbox{ \ensuremath{}\ensuremath{}\ensuremath{}\ensuremath{}\mbox{ \ensuremath{}\ensuremath{}\ensuremath{}\ensuremath{}}}\begin{cases}
dX_{s}=\left(r_{s}X_{s}+\pi_{s}^{T}\sigma_{s}\theta_{s}\right)ds+\pi_{s}^{T}\sigma_{s}dW_{s}\\
X_{t}=x_{t}\\
\pi\in\overline{U_{ad}^{\pi}}=\left\{ \pi\left|\right.\pi\in L_{\mathcal{F}}^{2}\left(0,T;\mathbb{R}^{d}\right)\mbox{and }\mathbb{E}_{t}\left[X_{T}\right]=X_{t}e^{\int_{t}^{T}\mu_{s}ds},\forall t\in\left[0,T\right)\right\} 
\end{cases}\label{eq:2.102}
\end{align}
which is exactly the family of problems (\ref{eq:2.5}), thus we deduce
that $\left(u^{*},X^{*}\right)$ is also equilibrium for (\ref{eq:2.5}).
\end{proof}

\subsection{Transformation of our problem}

By the above theorem, we try to solve the family of problems (\ref{eq:2.6})
instead, that is

\begin{align}
\underset{\pi}{\min} & \mbox{ \ensuremath{}\ensuremath{}\ensuremath{}\ensuremath{}\mbox{ \ensuremath{}\ensuremath{}\ensuremath{}\ensuremath{}}\mbox{ \ensuremath{}\ensuremath{}\ensuremath{}\ensuremath{}}}\mathbb{E}_{t}\left[\frac{\left(X_{T}-X_{t}e^{\int_{t}^{T}\lambda_{s}ds}\right)^{2}}{2}\right]\nonumber \\
s.t. & \mbox{ \ensuremath{}\ensuremath{}\ensuremath{}\ensuremath{}\mbox{ \ensuremath{}\ensuremath{}\ensuremath{}\ensuremath{}}}\begin{cases}
dX_{s}=\left(r_{s}X_{s}+\pi_{s}^{T}\sigma_{s}\theta_{s}\right)ds+\pi_{s}^{T}\sigma_{s}dW_{s}\\
X_{t}=x_{t}\\
\pi\in U_{ad}^{\pi}=\left\{ \pi\left|\right.\pi\in L_{\mathcal{F}}^{2}\left(0,T;\mathbb{R}^{d}\right)\right\} 
\end{cases}\label{eq:2.6-1}
\end{align}
by letting for any $s\in\left[0,T\right]$
\begin{equation}
Y_{s}=X_{s}e^{\int_{s}^{T}\lambda_{u}du}\label{eq:2.13}
\end{equation}
the above family of problems is equivalent to 
\begin{align}
\underset{u}{\min} & \mbox{ \ensuremath{}\ensuremath{}\ensuremath{}\ensuremath{}\mbox{ \ensuremath{}\ensuremath{}\ensuremath{}\ensuremath{}}\mbox{ \ensuremath{}\ensuremath{}\ensuremath{}\ensuremath{}}}\mathbb{E}_{t}\left[\frac{\left(Y_{T}-Y_{t}\right)^{2}}{2}\right]\nonumber \\
s.t. & \mbox{ \ensuremath{}\ensuremath{}\ensuremath{}\ensuremath{}\mbox{ \ensuremath{}\ensuremath{}\ensuremath{}\ensuremath{}}}\begin{cases}
dY_{s}=\left(\widehat{r_{s}}Y_{s}+u_{s}^{T}\theta_{s}\right)ds+u_{s}^{T}dW_{s}\\
Y_{t}=y_{t}\\
u\in U_{ad}^{u}=\left\{ u\left|\right.u\in L_{\mathcal{F}}^{2}\left(0,T;\mathbb{R}^{d}\right)\right\} 
\end{cases}\label{eq:2.14}
\end{align}
where $\widehat{r}_{s}=r_{s}-\lambda_{s}$, $u_{s}=e^{\int_{s}^{T}\lambda_{u}du}\sigma_{s}^{T}\pi_{s}$,
$y_{t}=x_{t}e^{\int_{t}^{T}\lambda_{u}du}$ and $ $$\pi\in L_{\mathcal{F}}^{2}\left(0,T;\mathbb{R}^{d}\right)$
if and only if $u\in L_{\mathcal{F}}^{2}\left(0,T;\mathbb{R}^{d}\right)$
which is implied by our previous assumption $\sigma$ is bounded and
that $e^{\int_{s}^{T}\lambda_{u}du}$ is bounded due to $\int_{0}^{T}\left|\lambda_{s}\right|ds<\infty$.\\

\noindent The above family of problems can be written as 

\noindent 
\begin{align}
\underset{u}{\min} & \mbox{ \ensuremath{}\ensuremath{}\ensuremath{}\ensuremath{}\mbox{ \ensuremath{}\ensuremath{}\ensuremath{}\ensuremath{}}\mbox{ \ensuremath{}\ensuremath{}\ensuremath{}\ensuremath{}}}\mathbb{E}_{t}\left[h\left(Y_{T}-Y_{t}\right)\right]\nonumber \\
s.t. & \mbox{ \ensuremath{}\ensuremath{}\ensuremath{}\ensuremath{}\mbox{ \ensuremath{}\ensuremath{}\ensuremath{}\ensuremath{}}}\begin{cases}
dY_{s}=\left(\widehat{r_{s}}Y_{s}+u_{s}^{T}\theta_{s}\right)ds+u_{s}^{T}dW_{s}\\
Y_{t}=y_{t}\\
u\in U_{ad}^{u}=\left\{ u\left|\right.u\in L_{\mathcal{F}}^{2}\left(0,T;\mathbb{R}^{d}\right)\right\} 
\end{cases}\label{eq:2.15}
\end{align}
where $h\left(x\right)=\frac{x^{2}}{2}$, and this is exactly a family
of problems of the form in (\ref{eq:1.3}), so we have the following
system of BSDEs by (\ref{eq:1.4}) and (\ref{eq:1.5}) 
\begin{equation}
\begin{cases}
dp_{s}^{t} & =-\widehat{r_{s}}p_{s}^{t}ds+\left(q_{s}^{t}\right)^{T}dW_{s},\mbox{ \ensuremath{}s\ensuremath{\in\left[t,T\right]}}\\
p_{T}^{t} & =\frac{dh\left(Y_{T}^{*}-Y_{t}^{*}\right)}{dx}=Y_{T}^{*}-Y_{t}^{*}
\end{cases}\label{eq:2.16}
\end{equation}
\begin{equation}
\begin{cases}
dP_{s}^{t} & =-2\widehat{r_{s}}P_{s}^{t}ds+\left(Q_{s}^{t}\right)^{T}dW_{s},\mbox{ \ensuremath{}s\ensuremath{\in\left[t,T\right]}}\\
P_{T}^{t} & =\frac{d^{2}h\left(Y_{T}^{*}-Y_{t}^{*}\right)}{dx^{2}}=1
\end{cases}\label{eq:1.5-1}
\end{equation}

\begin{prop}
\label{prop:2.2}If (\ref{eq:2.15}) and (\ref{eq:2.16}) admit a
solution $\left(u^{*},Y^{*},p^{t},q^{t}\right)$ for any $t\in\left[0,T\right)$
s.t.
\begin{equation}
\begin{cases}
u^{*}\in L_{\mathcal{F}}^{2}\left(0,T;\mathbb{R}^{d}\right)\\
\mathbb{E}_{t}\int_{t}^{T}\left|\Lambda_{s}^{t}\right|ds<\infty\mbox{ and \ensuremath{\underset{s\left\downarrow t\right.}{\lim}\mathbb{E}_{t}\left[\Lambda_{s}^{t}\right]=0}} & where\mbox{ }\Lambda_{s}^{t}=p_{s}^{t}\theta+q_{s}^{t}.
\end{cases}\label{eq:2.18}
\end{equation}
 then $u^{*}$ is an equilibrium control for the family of problems
in (\ref{eq:2.14}). \end{prop}
\begin{proof}
Since one of the sufficient conditions in (\ref{eq:1.11}) $\mathbb{E}_{t}\left[\frac{d^{2}h\left(Y_{T}^{*}-Y_{t}^{*}\right)}{dx^{2}}\right]=1\geq0$
has been met, thus we have the required result by theorem \ref{thm:1.3}.
\end{proof}

\subsection{Details of finding a potential equilibrium}

\noindent By the assumptions we have made so far we have that $\theta$
is bounded and $\widehat{r}$ is deterministic with$\int_{0}^{T}\left|\widehat{r}_{s}\right|ds<\infty$.
Firstly, we allow $\mu^{x}$ which is our drift rate vector process
of risky assets to be random, i.e. $\theta$ could be random.\\

\noindent Then for any $t\in\left[0,T\right)$ we make the following
Ansatz
\begin{equation}
p_{s}^{t}=M_{s}Y_{s}^{*}-\Gamma_{s}Y_{t}^{*}
\end{equation}
where $(M,K),(\Gamma,\phi)$ are the solutions of the following BSDEs
\begin{align}
 & \begin{cases}
dM_{s} & =-f_{M,K}\left(s,M_{s},K_{s}\right)ds+K_{s}^{T}dW_{s}\\
M_{T} & =1
\end{cases}\label{eq:2.20-1}\\
 & \begin{cases}
d\Gamma_{s} & =-f_{\Gamma,\phi}\left(s,\Gamma{}_{s},\phi_{s}\right)ds+\phi_{s}^{T}dW_{s}\\
\Gamma_{T} & =1
\end{cases}\label{eq:2.21}
\end{align}
by It$\hat{\mbox{o}}$ formula we get 
\begin{eqnarray}
d\left(M_{s}Y_{s}^{*}\right) & = & M_{s}dY_{s}^{*}+Y_{s}^{*}dM_{s}+d\left\langle Y^{*},M\right\rangle _{s}\nonumber \\
 & = & M_{s}\left[\left(\widehat{r_{s}}Y_{s}^{*}+\left(u_{s}^{*}\right)^{T}\theta_{s}\right)ds+\left(u_{s}^{*}\right)^{T}dW_{s}\right]\nonumber \\
 &  & +Y_{s}^{*}\left[-f_{M,K}\left(s,M_{s},K_{s}\right)ds+K_{s}^{T}dW_{s}\right]+K_{s}^{T}u_{s}^{*}ds\nonumber \\
 & = & \left[M_{s}\widehat{r_{s}}Y_{s}^{*}+M_{s}\left(u_{s}^{*}\right)^{T}\theta_{s}+K_{s}^{T}u_{s}^{*}-f_{M,K}\left(s,M_{s},K_{s}\right)Y_{s}^{*}\right]ds\nonumber \\
 &  & +\left[M_{s}\left(u_{s}^{*}\right)^{T}+K_{s}^{T}Y_{s}^{*}\right]dW_{s}
\end{eqnarray}
then
\begin{eqnarray}
dp_{s}^{t} & = & dM_{s}Y_{s}^{*}-Y_{t}^{*}d\Gamma_{s}\nonumber \\
 & = & \left[M_{s}\widehat{r_{s}}Y_{s}^{*}+M_{s}\left(u_{s}^{*}\right)^{T}\theta_{s}+K_{s}^{T}u_{s}^{*}-f_{M,K}\left(s,M_{s},K_{s}\right)Y_{s}^{*}+Y_{t}^{*}f_{\Gamma,\phi}\left(s,\Gamma_{s},\phi_{s}\right)\right]ds\nonumber \\
 &  & +\left[M_{s}\left(u_{s}^{*}\right)^{T}+K_{s}^{T}Y_{s}^{*}-Y_{t}^{*}\phi_{s}^{T}\right]dW_{s}\label{eq:2.23}
\end{eqnarray}
\\
by comparing the $dW$ term of (\ref{eq:2.23}) and (\ref{eq:2.16}),
we get
\begin{equation}
q_{s}^{t}=M_{s}u_{s}^{*}+K_{s}Y_{s}^{*}-Y_{t}^{*}\phi_{s}\label{eq:2.24}
\end{equation}
\\
Suppose $\Lambda^{t}$ is continuous and bounded, then $\underset{s\left\downarrow t\right.}{\lim}\mathbb{E}_{t}\left[\Lambda_{s}^{t}\right]=0,\forall t\in\left[0,T\right)$
is ensured by $\Lambda_{t}^{t}=0,\forall t\in\left[0,T\right)$. To
get a possible linear feedback $u^{*}$, we try by setting 
\[
0=\Lambda_{s}^{s}=p_{s}^{s}\theta_{s}+q_{s}^{s},\mbox{ \ensuremath{}s\ensuremath{\in\left[0,T\right]}}
\]
that is
\[
0=\left(M_{s}-\Gamma_{s}\right)Y_{s}^{*}\theta_{s}+M_{s}u_{s}^{*}+K_{s}Y_{s}^{*}-Y_{s}^{*}\phi_{s}
\]
from which we get
\begin{eqnarray}
u_{s}^{*} & = & \left[\left(\frac{\Gamma_{s}}{M_{s}}-1\right)\theta_{s}-\frac{K_{s}}{M_{s}}+\frac{\phi_{s}}{M_{s}}\right]Y_{s}^{*}\nonumber \\
 & = & \alpha_{s}Y_{s}^{*}\label{eq:2.25-1}
\end{eqnarray}
where $\alpha_{s}=\left(\frac{\Gamma_{s}}{M_{s}}-1\right)\theta_{s}-\frac{K_{s}}{M_{s}}+\frac{\phi_{s}}{M_{s}}$
based on the assumption that $M_{s}\neq0,\forall s\in\left[0,T\right)$\\

\noindent also by comparing the $ds$ term of (\ref{eq:2.23}) and
(\ref{eq:2.16}), we get
\begin{eqnarray*}
 & -\widehat{r_{s}}\left(M_{s}Y_{s}^{*}-\Gamma_{s}Y_{t}^{*}\right)= & M_{s}\widehat{r_{s}}Y_{s}^{*}+M_{s}\left(u_{s}^{*}\right)^{T}\theta_{s}+K_{s}^{T}u_{s}^{*}\\
 &  & -f_{M,K}\left(s,M_{s},K_{s}\right)Y_{s}^{*}+Y_{t}^{*}f_{\Gamma,\phi}\left(s,\Gamma_{s},\phi_{s}\right)\\
\Leftrightarrow & -\widehat{r_{s}}M_{s}Y_{s}^{*}+\widehat{r_{s}}\Gamma_{s}Y_{t}^{*}= & \left[M_{s}\widehat{r_{s}}-f_{M,K}\left(s,M_{s},K_{s}\right)\right]Y_{s}^{*}+M_{s}\alpha_{s}^{T}\theta_{s}Y_{s}^{*}\\
 &  & +K_{s}^{T}\alpha_{s}Y_{s}^{*}+Y_{t}^{*}f_{\Gamma,\phi}\left(s,\Gamma_{s},\phi_{s}\right)\\
\Leftrightarrow & \left[\widehat{r_{s}}\Gamma_{s}-f_{\Gamma,\phi}\left(s,\Gamma_{s},\phi_{s}\right)\right]Y_{t}^{*}= & \left[M_{s}\widehat{r_{s}}-f_{M,K}\left(s,M_{s},K_{s}\right)+M_{s}\alpha_{s}^{T}\theta_{s}+K_{s}^{T}\alpha_{s}+\widehat{r_{s}}M_{s}\right]Y_{s}^{*}
\end{eqnarray*}
which leads to the following equations
\begin{eqnarray}
f_{M,K}\left(s,M_{s},K_{s}\right) & = & 2\widehat{r_{s}}M_{s}+\left(M_{s}\theta_{s}^{T}+K_{s}^{T}\right)\alpha_{s}\label{eq:2.25}\\
f_{\Gamma,\phi}\left(s,\Gamma_{s},\phi_{s}\right) & = & \widehat{r_{s}}\Gamma_{s}\label{eq:2.26}
\end{eqnarray}
\\
plug (\ref{eq:2.26}) back into (\ref{eq:2.21}) we get
\begin{equation}
\begin{cases}
d\Gamma_{s} & =-\widehat{r_{s}}\Gamma_{s}ds+\phi_{s}^{T}dW_{s}\\
\Gamma_{T} & =1
\end{cases}
\end{equation}
that is
\begin{equation}
\begin{cases}
d\Gamma_{s}e^{\int_{0}^{s}\widehat{r_{u}}du} & =e^{\int_{0}^{s}\widehat{r_{u}}du}\phi_{s}^{T}dW_{s}\\
\Gamma_{T}e^{\int_{0}^{T}\widehat{r_{u}}du} & =e^{\int_{0}^{T}\widehat{r_{u}}du}
\end{cases}
\end{equation}
which could be solved and we get

\noindent 
\begin{eqnarray}
\Gamma_{s} & = & \mathbb{E}_{s}\left[e^{\int_{s}^{T}\widehat{r_{u}}du}\right]\nonumber \\
 & = & e^{\int_{s}^{T}\widehat{r_{u}}du}
\end{eqnarray}
which implies
\begin{equation}
\phi=0\label{eq:2.32}
\end{equation}
and plug (\ref{eq:2.32}) back into (\ref{eq:2.24}) and (\ref{eq:2.25-1})
we get
\begin{equation}
q_{s}^{t}=M_{s}u_{s}^{*}+K_{s}Y_{s}^{*}
\end{equation}
and

\begin{eqnarray}
u_{s}^{*} & = & \left[\left(\frac{\Gamma_{s}}{M_{s}}-1\right)\theta_{s}-\frac{K_{s}}{M_{s}}\right]Y_{s}^{*}\nonumber \\
 & = & \alpha_{s}Y_{s}^{*}\label{eq:2.20}
\end{eqnarray}
where $\alpha_{s}=\left(\frac{\Gamma_{s}}{M_{s}}-1\right)\theta_{s}-\frac{K_{s}}{M_{s}}$\\

\noindent plug (\ref{eq:2.25}) into (\ref{eq:2.20-1}) we get that
\begin{equation}
\begin{cases}
dM_{s} & =-\left[2\widehat{r_{s}}M_{s}+\left(M_{s}\theta_{s}^{T}+K_{s}^{T}\right)\alpha_{s}\right]ds+K_{s}^{T}dW_{s}\\
M_{T} & =1
\end{cases}\label{eq:2.34}
\end{equation}
and plug $\alpha$ into (\ref{eq:2.34}) we have
\begin{equation}
\begin{cases}
dM_{s} & =-\left\{ 2\widehat{r_{s}}M_{s}+\left(M_{s}\theta_{s}^{T}+K_{s}^{T}\right)\left[\left(\frac{\Gamma_{s}}{M_{s}}-1\right)\theta_{s}-\frac{K_{s}}{M_{s}}\right]\right\} ds+K_{s}^{T}dW_{s}\\
M_{T} & =1
\end{cases}
\end{equation}
that is
\begin{equation}
\begin{cases}
dM_{s} & =-\left[\left(2\widehat{r_{s}}-\left|\theta_{s}\right|^{2}\right)M_{s}+\Gamma_{s}\left|\theta_{s}\right|^{2}-2K_{s}^{T}\theta_{s}+\frac{\Gamma_{s}}{M_{s}}K_{s}^{T}\theta_{s}-\frac{\left|K_{s}\right|^{2}}{M_{s}}\right]ds+K_{s}^{T}dW_{s}\\
M_{T} & =1
\end{cases}\label{eq:2.36}
\end{equation}
\\
Here we list the fact about the BMO martingale in Kazamaki\cite{key-4}
which will be used to prove the existence of a solution to the above
BSDE (\ref{eq:2.36}). 
\begin{fact}
\label{2.3}A process $\int_{0}^{\cdot}Z_{s}^{T}dW_{s}$ is a BMO
martingale if and only if there exists a constant $C>0$ s.t.
\begin{equation}
\mathbb{E}\left[\left.\int_{\tau}^{T}\left|Z_{s}\right|^{2}ds\right|\mathcal{F}_{\tau}\right]\leq C
\end{equation}
for any stopping time $\tau\leq T$. A BMO martingale $\int_{0}^{\cdot}Z_{s}^{T}dW_{s}$
has the property that the stochastic exponential $\mathcal{E}\left(\int_{0}^{\cdot}Z_{s}^{T}dW_{s}\right)$
is a martingale. Thus by defining $\left.\frac{d\mathbb{Q}}{d\mathbb{P}}\right|_{\mathcal{F}_{t}}=\mathcal{E}\left(\int_{0}^{t}Z_{s}^{T}dW_{s}\right)$,
we have that $W_{t}^{\mathbb{Q}}=W_{t}-$\textup{$\int_{0}^{t}Z_{s}ds$
is a $\mathbb{Q}$-Brownian motion}\\
\end{fact}
\begin{prop}
\label{prop;2.4}The BSDE (\ref{eq:2.36}) admits a solution ($M,K)\in L_{\mathcal{F}}^{\infty}\left(0,T;\mathbb{R}\right)\times L_{\mathcal{F}}^{2}\left(0,T;\mathbb{R}^{d}\right)$
s.t. $M\geq\eta$ for some constant $\eta>0$\end{prop}
\begin{proof}
Here we use the truncation method. Choose any constant $c>0$ and
bound the $\frac{1}{M}$ in the BSDE (\ref{eq:2.36}) by the chosen
$c$, which leads to the following BSDE 
\begin{equation}
\begin{cases}
dM_{s} & =-\left[\left(2\widehat{r_{s}}-\left|\theta_{s}\right|^{2}\right)M_{s}+\Gamma_{s}\left|\theta_{s}\right|^{2}-2K_{s}^{T}\theta_{s}+\frac{\Gamma_{s}}{M_{s}\vee c}K_{s}^{T}\theta_{s}-\frac{\left|K_{s}\right|^{2}}{M_{s}\vee c}\right]ds+K_{s}^{T}dW_{s}\\
M_{T} & =1
\end{cases}\label{eq:2.38}
\end{equation}
which can be written as 
\begin{equation}
\begin{cases}
dM_{s} & =-\left[\left(2\widehat{r_{s}}-\left|\theta_{s}\right|^{2}\right)M_{s}+\Gamma_{s}\left|\theta_{s}\right|^{2}\right]ds+K_{s}^{T}\left[dW_{s}-\left(-2\theta_{s}+\frac{\Gamma_{s}}{M_{s}\vee c}\theta_{s}-\frac{K_{s}}{M_{s}\vee c}\right)ds\right]\\
M_{T} & =1
\end{cases}\label{eq:2.39}
\end{equation}
(\ref{eq:2.38}) is a standard quadratic BSDE. Hence there exists
a solution $\left(M^{c},K^{c}\right)\in L_{\mathcal{F}}^{\infty}\left(0,T;\mathbb{R}\right)\times L_{\mathcal{F}}^{2}\left(0,T;\mathbb{R}^{d}\right)$
depending on the chosen constant $c$ and $\int_{0}^{\cdot}\left(K_{s}^{c}\right)^{T}dW_{s}$
is a BMO martingale according to the results in Kobylanski\cite{key-5}
and Morlais\cite{key-6}.\\

\noindent since $\int_{0}^{\cdot}\left(K_{s}^{c}\right)^{T}dW_{s}$
is a BMO martingale and $\frac{1}{M^{c}\vee c},\theta,\Gamma$ are
all bounded, we have that 
\begin{equation}
\int_{0}^{\cdot}\left(-2\theta_{s}+\frac{\Gamma_{s}}{M_{s}^{c}\vee c}\theta_{s}-\frac{K_{s}^{c}}{M_{s}^{c}\vee c}\right)^{T}dW_{s}
\end{equation}
is also a BMO martingale by definition according to fact \ref{2.3}.
Thus by defining 
\[
\left.\frac{d\mathbb{Q}}{d\mathbb{P}}\right|_{\mathcal{F}_{t}}=\mathcal{E}\left[\int_{0}^{t}\left(-2\theta_{s}+\frac{\Gamma_{s}}{M_{s}^{c}\vee c}\theta_{s}-\frac{K_{s}^{c}}{M_{s}^{c}\vee c}\right){}^{T}dW_{s}\right]
\]
we have that 
\[
W_{t}^{\mathbb{Q}}=W_{t}-\int_{0}^{t}\left(-2\theta_{s}+\frac{\Gamma_{s}}{M_{s}^{c}\vee c}\theta_{s}-\frac{K_{s}^{c}}{M_{s}^{c}\vee c}\right)ds
\]
is a $\mathbb{Q}$-Brownian motion, and thus$ $by using (\ref{eq:2.39})
we have
\begin{equation}
\begin{cases}
dM_{s}^{c}e^{\int_{0}^{s}2\widehat{r_{u}}-\left|\theta_{u}\right|^{2}du} & =-e^{\int_{0}^{s}2\widehat{r_{u}}-\left|\theta_{u}\right|^{2}du}\Gamma_{s}\left|\theta_{s}\right|^{2}ds+e^{\int_{0}^{s}2\widehat{r_{u}}-\left|\theta_{u}\right|^{2}du}\left(K_{s}^{c}\right)^{T}dW_{s}^{\mathbb{Q}}\\
M_{T}^{c}e^{\int_{0}^{T}2\widehat{r_{u}}-\left|\theta_{u}\right|^{2}du} & =e^{\int_{0}^{T}2\widehat{r_{u}}-\left|\theta_{u}\right|^{2}du}
\end{cases}
\end{equation}
which is
\begin{equation}
\begin{cases}
d\left(M_{s}^{c}e^{\int_{0}^{s}2\widehat{r_{u}}-\left|\theta_{u}\right|^{2}du}+\int_{0}^{s}e^{\int_{0}^{u}2\widehat{r_{v}}-\left|\theta_{v}\right|^{2}dv}\Gamma_{u}\left|\theta_{u}\right|^{2}du\right) & =e^{\int_{0}^{s}2\widehat{r_{u}}-\left|\theta_{u}\right|^{2}du}\left(K_{s}^{c}\right)^{T}dW_{s}^{\mathbb{Q}}\\
M_{T}^{c}e^{\int_{0}^{T}2\widehat{r_{u}}-\left|\theta_{u}\right|^{2}du}+\int_{0}^{T}e^{\int_{0}^{u}2\widehat{r_{v}}-\left|\theta_{v}\right|^{2}dv}\Gamma_{u}\left|\theta_{u}\right|^{2}du & =e^{\int_{0}^{T}2\widehat{r_{u}}-\left|\theta_{u}\right|^{2}du}+\int_{0}^{T}e^{\int_{0}^{u}2\widehat{r_{v}}-\left|\theta_{v}\right|^{2}dv}\Gamma_{u}\left|\theta_{u}\right|^{2}du
\end{cases}
\end{equation}
thus
\begin{equation}
M_{s}^{c}=\mathbb{E}_{s}^{\mathbb{Q}}\left[e^{\int_{s}^{T}2\widehat{r_{u}}-\left|\theta_{u}\right|^{2}du}+\int_{s}^{T}e^{\int_{s}^{u}2\widehat{r_{v}}-\left|\theta_{v}\right|^{2}dv}\Gamma_{u}\left|\theta_{u}\right|^{2}du\right]
\end{equation}
it implies that there exists a constant $\eta>0$ s.t. we have $M^{c}\geq\eta$
for any chosen constant $c>0$, and we could choose $c=\eta$ in the
BSDE (\ref{eq:2.38}) and deduce that the BSDE (\ref{eq:2.36}) admits
a solution ($M,K)\in L_{\mathcal{F}}^{\infty}\left(0,T;\mathbb{R}\right)\times L_{\mathcal{F}}^{2}\left(0,T;\mathbb{R}^{d}\right)$
s.t. $M\geq\eta$.
\end{proof}
\noindent Since we have proved proposition (\ref{prop;2.4}), we could
have our linear feedback as
\begin{eqnarray}
u_{s}^{*} & = & \left[\left(\frac{\Gamma_{s}}{M_{s}}-1\right)\theta_{s}-\frac{K_{s}}{M_{s}}\right]Y_{s}^{*}\nonumber \\
 & = & \alpha_{s}Y_{s}^{*}\label{eq:2.44}
\end{eqnarray}
where $\alpha_{s}=\left(\frac{\Gamma_{s}}{M_{s}}-1\right)\theta_{s}-\frac{K_{s}}{M_{s}}$
has the process $\lambda$ as its parameter, i.e. $\alpha_{s}=f_{s}(\lambda)$
with $f_{s}=\left(\frac{\Gamma_{s}}{M_{s}}-1\right)\theta_{s}-\frac{K_{s}}{M_{s}}$
for any $s\in\left[0,T\right]$\\

\begin{thm}
The control $u^{*}$ in (\ref{eq:2.44}) is an equilibrium control
for the family of problems in (\ref{eq:2.14})\end{thm}
\begin{proof}
Let $Y^{*}$ be the corresponding state process with respect to the
above control $u^{*}$ in (\ref{eq:2.44}). We have for $s\in\left[0,T\right]$
\[
u_{s}^{*}=\alpha_{s}Y_{s}^{*}
\]
Plug the $u^{*}$ into the dynamics of $Y$ in (\ref{eq:2.15}) we
get
\begin{eqnarray}
dY_{s}^{*} & = & \left(\widehat{r_{s}}Y_{s}^{*}+\left(u_{s}^{*}\right)^{T}\theta_{s}\right)ds+\left(u_{s}^{*}\right)^{T}dW_{s}\nonumber \\
 & = & Y_{s}^{*}\left[\left(\widehat{r_{s}}+\alpha_{s}^{T}\theta_{s}\right)ds+\alpha_{s}^{T}dW_{s}\right]\label{eq:4.47}
\end{eqnarray}
thus we have that
\begin{align}
Y_{t}^{*} & =y_{0}e^{\int_{0}^{t}\widehat{r_{s}}+\alpha_{s}^{T}\theta_{s}-\frac{\left|\alpha_{s}\right|^{2}}{2}ds+\int_{0}^{t}\alpha_{s}^{T}dW_{s}}\nonumber \\
 & =y_{0}e^{\int_{0}^{t}\widehat{r_{s}}ds}e^{\int_{0}^{t}-\frac{\left|\alpha_{s}\right|^{2}}{2}ds+\int_{0}^{t}\alpha_{s}^{T}\left(dW_{s}+\theta_{s}ds\right)}\nonumber \\
 & =y_{0}e^{\int_{0}^{t}\widehat{r_{s}}ds}e^{\int_{0}^{t}-\frac{\left|\alpha_{s}\right|^{2}}{2}ds+\int_{0}^{t}\left(\frac{\Gamma_{s}}{M_{s}}-1\right)\theta_{s}^{T}\left(dW_{s}+\theta_{s}ds\right)-\int_{0}^{t}\frac{K_{s}^{T}}{M_{s}}\left(dW_{s}+\theta_{s}ds\right)}\nonumber \\
 & =y_{0}e^{\int_{0}^{t}\widehat{r_{s}}ds}e^{\int_{0}^{t}-\frac{\left|\alpha_{s}\right|^{2}}{2}+\left(\frac{\Gamma_{s}}{M_{s}}-1\right)\left|\theta_{s}\right|^{2}ds-\int_{0}^{t}\frac{K_{s}^{T}}{M_{s}}\left(dW_{s}+\theta_{s}ds\right)}e^{\int_{0}^{t}\left(\frac{\Gamma_{s}}{M_{s}}-1\right)\theta_{s}^{T}dW_{s}}\nonumber \\
 & =y_{0}e^{\int_{0}^{t}\widehat{r_{s}}ds}e^{-\int_{0}^{t}\left[\frac{\left|\alpha_{s}\right|^{2}}{2}-\left(\frac{\Gamma_{s}}{M_{s}}-1\right)\left|\theta_{s}\right|^{2}-\frac{1}{2}\left(\frac{\Gamma_{s}}{M_{s}}-1\right)^{2}\left|\theta_{s}\right|^{2}\right]ds-\int_{0}^{t}\frac{K_{s}^{T}}{M_{s}}\left(dW_{s}+\theta_{s}ds\right)}\mathcal{E}\left(\int_{0}^{t}\left(\frac{\Gamma_{s}}{M_{s}}-1\right)\theta_{s}^{T}dW_{s}\right)\nonumber \\
 & =y_{0}e^{\int_{0}^{t}\widehat{r_{s}}ds}e^{-\int_{0}^{t}\left\{ \left[\frac{\left|\alpha_{s}\right|^{2}}{2}-\left(\frac{\Gamma_{s}}{M_{s}}-1\right)\left|\theta_{s}\right|^{2}-\frac{1}{2}\left(\frac{\Gamma_{s}}{M_{s}}-1\right)^{2}\left|\theta_{s}\right|^{2}+\frac{K_{s}^{T}}{M_{s}}\theta_{s}\right]ds+\frac{K_{s}^{T}}{M_{s}}dW_{s}\right\} }\mathcal{E}\left(\int_{0}^{t}\left(\frac{\Gamma_{s}}{M_{s}}-1\right)\theta_{s}^{T}dW_{s}\right)\label{eq:2.46}
\end{align}
by applying It$\hat{\mbox{o}}$ formula to $\log\left(M\right)$ we
get 
\begin{align}
d\log\left(M_{s}\right) & =\frac{1}{M_{s}}dM_{s}-\frac{1}{2M_{s}^{2}}d\left\langle M\right\rangle _{s}\nonumber \\
 & =-\left[\left(2\widehat{r_{s}}-\left|\theta_{s}\right|^{2}\right)+\frac{\Gamma_{s}\left|\theta_{s}\right|^{2}}{M_{s}}-\frac{2K_{s}^{T}\theta_{s}}{M_{s}}+\frac{\Gamma_{s}K_{s}^{T}\theta_{s}}{M_{s}^{2}}-\frac{\left|K_{s}\right|^{2}}{M_{s}^{2}}\right]ds+\frac{K_{s}^{T}}{M_{s}}dW_{s}-\frac{\left|K_{s}\right|^{2}}{2M_{s}^{2}}ds\nonumber \\
 & =-\left[2\widehat{r_{s}}+\left(\frac{\Gamma_{s}}{M_{s}}-1\right)\left|\theta_{s}\right|^{2}+\left(\frac{\Gamma_{s}}{M_{s}}-1\right)\frac{K_{s}^{T}\theta_{s}}{M_{s}}-\frac{\left|K_{s}\right|^{2}}{2M_{s}^{2}}\right]ds+\frac{K_{s}^{T}}{M_{s}}\left(dW_{s}+\theta_{s}ds\right)\nonumber \\
 & =\left[\frac{\left|\alpha_{s}\right|^{2}}{2}-\frac{1}{2}\left(\frac{\Gamma_{s}}{M_{s}}-1\right)^{2}\left|\theta_{s}\right|^{2}-2\widehat{r_{s}}-\left(\frac{\Gamma_{s}}{M_{s}}-1\right)\left|\theta_{s}\right|^{2}\right]ds+\frac{K_{s}^{T}}{M_{s}}\left(dW_{s}+\theta_{s}ds\right)\nonumber \\
 & =-2\widehat{r_{s}}+\left\{ \left[\frac{\left|\alpha_{s}\right|^{2}}{2}-\left(\frac{\Gamma_{s}}{M_{s}}-1\right)\left|\theta_{s}\right|^{2}-\frac{1}{2}\left(\frac{\Gamma_{s}}{M_{s}}-1\right)^{2}\left|\theta_{s}\right|^{2}+\frac{K_{s}^{T}}{M_{s}}\theta_{s}\right]ds+\frac{K_{s}^{T}}{M_{s}}dW_{s}\right\} 
\end{align}
from which we get
\begin{align}
 & \log\left(\frac{M_{t}}{M_{0}}\right)=\int_{0}^{t}-2\widehat{r_{s}}ds+\int_{0}^{t}\left\{ \left[\frac{\left|\alpha_{s}\right|^{2}}{2}-\left(\frac{\Gamma_{s}}{M_{s}}-1\right)\left|\theta_{s}\right|^{2}-\frac{1}{2}\left(\frac{\Gamma_{s}}{M_{s}}-1\right)^{2}\left|\theta_{s}\right|^{2}+\frac{K_{s}^{T}}{M_{s}}\theta_{s}\right]ds+\frac{K_{s}^{T}}{M_{s}}dW_{s}\right\} \nonumber \\
\Rightarrow & \frac{M_{t}}{M_{0}}=e^{\int_{0}^{t}-2\widehat{r_{s}}ds}e^{\int_{0}^{t}\left\{ \left[\frac{\left|\alpha_{s}\right|^{2}}{2}-\left(\frac{\Gamma_{s}}{M_{s}}-1\right)\left|\theta_{s}\right|^{2}-\frac{1}{2}\left(\frac{\Gamma_{s}}{M_{s}}-1\right)^{2}\left|\theta_{s}\right|^{2}+\frac{K_{s}^{T}}{M_{s}}\theta_{s}\right]ds+\frac{K_{s}^{T}}{M_{s}}dW_{s}\right\} }\label{eq:2.48}
\end{align}
by combining (\ref{eq:2.46}) and (\ref{eq:2.48}) we get 
\begin{equation}
Y_{t}^{*}=y_{0}e^{\int_{0}^{t}-\widehat{r_{s}}ds}\frac{M_{0}}{M_{t}}\mathcal{E}\left(\int_{0}^{t}\left(\frac{\Gamma_{s}}{M_{s}}-1\right)\theta_{s}^{T}dW_{s}\right)
\end{equation}
\\
So we have as $\int_{0}^{T}\left|\widehat{r}_{s}\right|ds<\infty$
and $\Gamma,M,\frac{1}{M},\theta$ are all bounded that
\begin{eqnarray}
\mathbb{E}\left[\underset{t\in\left[0,T\right]}{\sup}\left(Y_{t}^{*}\right)^{2}\right] & = & \mathbb{E}\left\{ \underset{t\in\left[0,T\right]}{\sup}\left[\left(y_{0}e^{\int_{0}^{t}-\widehat{r_{s}}ds}\frac{M_{0}}{M_{t}}\right)^{2}\mathcal{E}\left(\int_{0}^{t}\left(\frac{\Gamma_{s}}{M_{s}}-1\right)\theta_{s}^{T}dW_{s}\right)^{2}\right]\right\} \nonumber \\
 & \leq & \mathbb{E}\left\{ \underset{t\in\left[0,T\right]}{\sup}\left(y_{0}e^{\int_{0}^{t}-\widehat{r_{s}}ds}\frac{M_{0}}{M_{t}}\right)^{2}\underset{t\in\left[0,T\right]}{\sup}\mathcal{E}\left(\int_{0}^{t}\left(\frac{\Gamma_{s}}{M_{s}}-1\right)\theta_{s}^{T}dW_{s}\right)^{2}\right\} \nonumber \\
 & \leq & y_{0}^{2}\underset{t\in\left[0,T\right]}{\sup}e^{\int_{0}^{t}-2\widehat{r_{s}}ds}\mathbb{E}\left[\underset{t\in\left[0,T\right]}{\sup}\left(\frac{M_{0}}{M_{t}}\right)^{2}\underset{t\in\left[0,T\right]}{\sup}\mathcal{E}\left(\int_{0}^{t}\left(\frac{\Gamma_{s}}{M_{s}}-1\right)\theta_{s}^{T}dW_{s}\right)^{2}\right]\nonumber \\
 & \leq & C\mathbb{E}\left[\underset{t\in\left[0,T\right]}{\sup}\mathcal{E}\left(\int_{0}^{t}\left(\frac{\Gamma_{s}}{M_{s}}-1\right)\theta_{s}^{T}dW_{s}\right)^{2}\right]\mbox{, for some constant \ensuremath{C}}\nonumber \\
 &  & \mbox{and }\mathcal{E}\left(\int_{0}^{t}\left(\frac{\Gamma_{s}}{M_{s}}-1\right)\theta_{s}^{T}dW_{s}\right)\mbox{is a martingale by Novikov}\nonumber \\
 & \leq & 4C\mathbb{E}\left[\mathcal{E}\left(\int_{0}^{T}\left(\frac{\Gamma_{s}}{M_{s}}-1\right)\theta_{s}^{T}dW_{s}\right)^{2}\right],\mbox{by \ensuremath{L^{p}}Maximal Inequality}\nonumber \\
 & = & 4C\mathbb{E}\left[e^{\int_{0}^{T}\left(\frac{\Gamma_{s}}{M_{s}}-1\right)^{2}\left|\theta_{s}\right|^{2}ds}\mathcal{E}\left(\int_{0}^{T}2\left(\frac{\Gamma_{s}}{M_{s}}-1\right)\theta_{s}^{T}dW_{s}\right)\right]\nonumber \\
 & < & \infty\label{eq:4.52}
\end{eqnarray}
as we also have $\mathbb{E}\left[\mathcal{E}\left(\int_{0}^{T}2\left(\frac{\Gamma_{s}}{M_{s}}-1\right)\theta_{s}^{T}dW_{s}\right)\right]=1$
, thus we deduce that 

\begin{equation}
Y^{*}\in L_{\mathcal{F}}^{2}\left(\Omega;C\left(0,T;\mathbb{R}\right)\right)\label{eq:2.52}
\end{equation}
\\
Since $\left(u^{*},Y^{*}\right)$ can be regraded as the solution
to the following BSDE
\begin{equation}
\begin{cases}
dY_{s} & =\left(\widehat{r_{s}}Y_{s}+\left(u_{s}\right)^{T}\theta_{s}\right)ds+\left(u_{s}\right)^{T}dW_{s}\\
Y_{T} & =Y_{T}^{*}
\end{cases}
\end{equation}
and we have 
\begin{equation}
\begin{cases}
d\left(e^{-\int_{0}^{s}\widehat{r_{u}}du}Y_{s}^{*}\right) & =\left[\left(e^{-\int_{0}^{s}\widehat{r_{u}}du}u_{s}^{*}\right)^{T}\right]\theta_{s}ds+\left[\left(e^{-\int_{0}^{s}\widehat{r_{u}}du}u_{s}^{*}\right)^{T}\right]dW_{s}\\
e^{-\int_{0}^{T}\widehat{r_{u}}du}Y_{T}^{*} & =e^{-\int_{0}^{T}\widehat{r_{u}}du}Y_{T}^{*}
\end{cases}\label{eq:4.55-1-1}
\end{equation}
which implies $\left(e^{-\int_{0}^{\cdot}\widehat{r_{u}}du}u^{*},e^{-\int_{0}^{\cdot}\widehat{r_{u}}du}Y^{*}\right)$
can be regraded as the solution to the following BSDE$ $
\begin{equation}
\begin{cases}
d\tilde{Y_{s}} & =\left(\tilde{u_{s}}\right)^{T}\theta_{s}ds+\left(\tilde{u_{s}}\right)^{T}dW_{s}\\
\tilde{Y_{T}} & =e^{-\int_{0}^{T}\widehat{r_{u}}du}Y_{T}^{*}
\end{cases}\label{eq:4.55-1}
\end{equation}
since the above BSDE (\ref{eq:4.55-1}) is Lipschitz as $\theta$
is bounded and $e^{-\int_{0}^{T}\widehat{r_{u}}du}Y_{T}^{*}\in L_{\mathcal{F}_{T}}^{2}\left(\Omega;\mathbb{R}\right)$,
thus we deduce that $e^{-\int_{0}^{\cdot}\widehat{r_{u}}du}u^{*}\in L_{\mathcal{F}}^{2}\left(0,T;\mathbb{R}^{d}\right)$
which implies that 
\begin{equation}
u^{*}\in L_{\mathcal{F}}^{2}\left(0,T;\mathbb{R}^{d}\right)\label{eq:2.53}
\end{equation}
\\
Also we have that for any $t\in\left[0,T\right)$ 
\begin{eqnarray}
\mbox{ }\Lambda_{s}^{t} & = & p_{s}^{t}\theta_{s}+q_{s}^{t}\nonumber \\
 & = & \left[M_{s}Y_{s}^{*}-\Gamma_{s}Y_{t}^{*}\right]\theta_{s}+M_{s}u_{s}^{*}+K_{s}Y_{s}^{*}\nonumber \\
 & = & \left\{ M_{s}\theta_{s}+M_{s}\left[\left(\frac{\Gamma_{s}}{M_{s}}-1\right)\theta_{s}-\frac{K_{s}}{M_{s}}\right]+K_{s}\right\} Y_{s}^{*}-\Gamma_{s}Y_{t}^{*}\theta_{s}\nonumber \\
 & = & \Gamma_{s}\theta_{s}Y_{s}^{*}-\Gamma_{s}Y_{t}^{*}\theta_{s}\nonumber \\
 & = & \Gamma_{s}\theta_{s}\left(Y_{s}^{*}-Y_{t}^{*}\right)
\end{eqnarray}
Since $\Gamma,\theta$ are bounded, it is clearly by (\ref{eq:2.52})
that 
\begin{equation}
\mathbb{E}_{t}\int_{t}^{T}\left|\Lambda_{s}^{t}\right|ds<\infty\label{eq:2.58}
\end{equation}
and we have 
\begin{eqnarray}
\underset{s\left\downarrow t\right.}{\lim}\mathbb{E}_{t}\left[\Lambda_{s}^{t}\right] & = & \underset{s\left\downarrow t\right.}{\lim}\mathbb{E}_{t}\left[\Gamma_{s}\theta_{s}\left(Y_{s}^{*}-Y_{t}^{*}\right)\right]\nonumber \\
 & \underset{=}{(\ref{eq:2.52})} & \mathbb{E}_{t}\left[\underset{s\left\downarrow t\right.}{\lim}\Gamma_{s}\theta_{s}\left(Y_{s}^{*}-Y_{t}^{*}\right)\right],\mbox{by Dominated Convergence}\nonumber \\
 & = & 0\label{eq:2.59}
\end{eqnarray}
Then (\ref{eq:2.53}), (\ref{eq:2.58}), (\ref{eq:2.59}) are exactly
the required conditions in (\ref{eq:2.18}) and we deduce that $u^{*}$
is an equilibrium control for the family of problems in (\ref{eq:2.14})
by proposition \ref{prop:2.2}.\\

\end{proof}
\noindent Since $u_{s}^{*}=e^{\int_{s}^{T}\lambda_{u}du}\sigma_{s}^{T}\pi_{s}^{*}$
by (\ref{eq:2.14}) and $u_{s}^{*}=\alpha_{s}Y_{s}^{*}=\alpha_{s}e^{\int_{s}^{T}\lambda_{u}du}X_{s}^{*}$
by (\ref{eq:2.20}) and (\ref{eq:2.13}). Thus we have that
\begin{equation}
\pi_{s}^{*}=\left(\sigma_{s}^{T}\right)^{-1}\alpha_{s}X_{s}^{*}\label{eq:2.60}
\end{equation}
which says $\pi_{s}^{*}$ is equilibrium to the family of problems
in (\ref{eq:2.6})\\

\subsection{Conditions for obtaining an equilibrium for our problem}
\begin{thm}
\label{thm:2.6}If there exists a deterministic process \textup{$\lambda^{*}$
with} \textup{$\int_{0}^{T}\left|\lambda_{s}^{*}\right|ds<\infty$}
s.t.\textup{ $\mathbb{E}_{t}^{\mathbb{Q}}\left[e^{\int_{t}^{T}\alpha_{s}^{T}\theta_{s}ds}\right]=e^{\int_{t}^{T}\mu_{s}-r_{s}ds}$}
for any $t\in\left[0,T\right)$ with \textup{$\left.\frac{d\mathbb{Q}}{d\mathbb{P}}\right|_{\mathcal{F}_{t}}=\mathcal{E}\left(\int_{0}^{t}\alpha_{s}^{T}dW_{s}\right)$,
}then the above $\pi^{*}$ in (\ref{eq:2.60}) is an equilibrium control
for our original family of problems (\ref{eq:2.5})\end{thm}
\begin{proof}
Suppose there exists a deterministic process $\lambda^{*}$ with $\int_{0}^{T}\left|\lambda_{s}^{*}\right|ds<\infty$
s.t. $\mathbb{E}_{t}^{\mathbb{Q}}\left[e^{\int_{t}^{T}r_{s}+\alpha_{s}^{T}\theta_{s}-\mu_{s}ds}\right]=1$
for any $t\in\left[0,T\right)$ with $\left.\frac{d\mathbb{Q}}{d\mathbb{P}}\right|_{\mathcal{F}_{t}}=\mathcal{E}\left(\int_{0}^{t}\alpha_{s}^{T}dW_{s}\right)$.
Firstly, we verify that $\left.\frac{d\mathbb{Q}}{d\mathbb{P}}\right|_{\mathcal{F}_{t}}=\mathcal{E}\left(\int_{0}^{t}\alpha_{s}^{T}dW_{s}\right)$
is well defined. Since $\int_{0}^{\cdot}\left(K_{s}\right)^{T}dW_{s}$
is a BMO martingale and $\frac{1}{M},\theta,\Gamma$ are all bounded,
we deduce that $\int_{0}^{\cdot}\left(\alpha_{s}\right)^{T}dW_{s}$
is a also BMO martingale by definition according to fact \ref{2.3}.
Thus we could define
\[
\left.\frac{d\mathbb{Q}}{d\mathbb{P}}\right|_{\mathcal{F}_{t}}=\mathcal{E}\left(\int_{0}^{t}\alpha_{s}^{T}dW_{s}\right)
\]
By combining (\ref{eq:2.13}) and (\ref{eq:4.47}) we could get
\begin{align}
 & X_{t}^{*}e^{\int_{t}^{T}\lambda_{s}^{*}ds}=X_{0}^{*}e^{\int_{0}^{T}\lambda_{s}^{*}ds}e^{\int_{0}^{t}r_{s}-\lambda_{s}^{*}+\alpha_{s}^{T}\theta_{s}-\frac{1}{2}\left|\alpha_{s}\right|^{2}ds+\int_{0}^{t}\alpha_{s}^{T}dW_{s}}\nonumber \\
\Rightarrow & X_{t}^{*}=x_{0}e^{\int_{0}^{t}r_{s}+\alpha_{s}^{T}\theta_{s}-\frac{1}{2}\left|\alpha_{s}\right|^{2}ds+\int_{0}^{t}\alpha_{s}^{T}dW_{s}}
\end{align}
where $\alpha_{s}=\left(\frac{\Gamma_{s}}{M_{s}}-1\right)\theta_{s}-\frac{K_{s}}{M_{s}}$,
and we have that
\begin{eqnarray}
\mathbb{E}_{t}\left[X_{T}^{*}\right] & = & X_{t}^{*}e^{\int_{t}^{T}\mu_{s}ds}\mathbb{E}_{t}\left[e^{\int_{t}^{T}r_{s}+\alpha_{s}^{T}\theta_{s}-\mu_{s}ds}e^{\int_{t}^{T}-\frac{1}{2}\left|\alpha_{s}\right|^{2}ds+\int_{t}^{T}\alpha_{s}^{T}dW_{s}}\right]\nonumber \\
 & = & X_{t}^{*}e^{\int_{t}^{T}\mu_{s}ds}\mathbb{E}_{t}^{\mathbb{Q}}\left[e^{\int_{t}^{T}r_{s}+\alpha_{s}^{T}\theta_{s}-\mu_{s}ds}\right]\nonumber \\
 & = & X_{t}^{*}e^{\int_{t}^{T}\mu_{s}ds}
\end{eqnarray}
for any $t\in\left[0,T\right)$. Since $\pi_{s}^{*}$ is equilibrium
to the family of problems in (\ref{eq:2.6}), thus we could deduce
that having $\lambda^{*}$ as parameter $\pi_{s}^{*}$$ $ is an equilibrium
control for our original family of problems (\ref{eq:2.5}) by theorem
\ref{thm:2.1}. 
\end{proof}
\noindent If we also assume that $\mu^{x}$ is deterministic, i.e.
$\theta$ is deterministic, then we have as $K=0$ that
\begin{eqnarray}
u_{s}^{*} & = & \left(\frac{\Gamma_{s}}{M_{s}}-1\right)\theta_{s}Y_{s}^{*}\nonumber \\
 & = & \alpha_{s}Y_{s}^{*}\label{eq:2.20-2}
\end{eqnarray}
where $\alpha_{s}=\left(\frac{\Gamma_{s}}{M_{s}}-1\right)\theta_{s}$
for any $s\in\left[0,T\right]$ and we could solve (\ref{eq:2.36})
and get

\noindent 
\begin{equation}
\begin{cases}
\Gamma_{s}=e^{\int_{s}^{T}\widehat{r_{u}}du}\\
M_{s}=e^{\int_{s}^{T}2\widehat{r_{u}}-\left|\theta_{u}\right|^{2}du}+\int_{s}^{T}\Gamma_{u}\left|\theta_{u}\right|^{2}e^{\int_{s}^{u}2\widehat{r_{x}}-\left|\theta_{x}\right|^{2}dx}du
\end{cases}
\end{equation}
In this case, theorem \ref{thm:2.6} becomes that if there exists
a deterministic process $\lambda^{*}$ with $\int_{0}^{T}\left|\lambda_{s}^{*}\right|ds<\infty$
s.t. $\int_{t}^{T}r_{s}+\alpha_{s}^{T}\theta_{s}-\mu_{s}ds=0$ for
any $t\in\left[0,T\right)$, then the above $\pi^{*}$ in (\ref{eq:2.60})
is an equilibrium control for our original family of problems in (\ref{eq:2.5}).\newpage{}

\section{Utility function\textmd{ $h\left(x\right)=-\frac{x^{3}}{3}$ }for
mean-cubic portfolio selection}

Some investors are risk seekers and they may choose a utility function
that looks quite risky as the one we will use here. $\frac{}{}$In
this section, we want to solve our moving target portfolio selection
problem when we choose to use $h(x)=-\frac{x^{3}}{3}$ as the utility
function. That means we want to solve the family of following problems
for any $t\in\left[0,T\right)$
\begin{align}
\underset{\pi}{\min} & \mbox{ \ensuremath{}\ensuremath{}\ensuremath{}\ensuremath{}\mbox{ \ensuremath{}\ensuremath{}\ensuremath{}\ensuremath{}}\mbox{ \ensuremath{}\ensuremath{}\ensuremath{}\ensuremath{}}}\mathbb{E}_{t}\left[-\frac{\left(X_{T}-X_{t}e^{\int_{t}^{T}\mu_{s}ds}\right)^{3}}{3}\right]\nonumber \\
s.t. & \mbox{ \ensuremath{}\ensuremath{}\ensuremath{}\ensuremath{}\mbox{ \ensuremath{}\ensuremath{}\ensuremath{}\ensuremath{}}}\begin{cases}
dX_{s}=\left(r_{s}X_{s}+\pi_{s}^{T}\sigma_{s}\theta_{s}\right)ds+\pi_{s}^{T}\sigma_{s}dW_{s}\\
X_{t}=x_{t}\\
\pi\in\overline{U_{ad}^{\pi}}=\left\{ \pi\left|\right.\pi\in L_{\mathcal{F}}^{2}\left(0,T;\mathbb{R}^{d}\right)\mbox{and }\mathbb{E}_{t}\left[X_{T}\right]=X_{t}e^{\int_{t}^{T}\mu_{s}ds},\forall t\in\left[0,T\right)\right\} 
\end{cases}\label{eq:3.1}
\end{align}
where $\mu$ is our required return process which is bounded and deterministic.

\subsection{Transformation of our problem}

Here we use another approach to solve our problem rather than the
Lagrangian multiplier method used in the previous section for $h(x)=\frac{x^{2}}{2}$.
By letting for any $s\in\left[0,T\right]$
\begin{equation}
Y_{s}=X_{s}e^{\int_{s}^{T}\mu_{u}du}
\end{equation}
we have the family of following problems for any $t\in\left[0,T\right)$,
which is equivalent to the above family of problems (\ref{eq:3.1})
\begin{align}
\underset{u}{\min} & \mbox{ \ensuremath{}\ensuremath{}\ensuremath{}\ensuremath{}\mbox{ \ensuremath{}\ensuremath{}\ensuremath{}\ensuremath{}}\mbox{ \ensuremath{}\ensuremath{}\ensuremath{}\ensuremath{}}}\mathbb{E}_{t}\left[-\frac{\left(Y_{T}-Y_{t}\right)^{3}}{3}\right]\nonumber \\
s.t. & \mbox{ \ensuremath{}\ensuremath{}\ensuremath{}\ensuremath{}\mbox{ \ensuremath{}\ensuremath{}\ensuremath{}\ensuremath{}}}\begin{cases}
dY_{s}=\left(\widehat{r_{s}}Y_{s}+u_{s}^{T}\theta_{s}\right)ds+u_{s}^{T}dW_{s}\\
Y_{t}=y_{t}\\
u\in\overline{U_{ad}^{u}}=\left\{ u\left|\right.u\in L_{\mathcal{F}}^{2}\left(0,T;\mathbb{R}^{d}\right)\mbox{and }\mathbb{E}_{t}\left[Y_{T}\right]=Y_{t},\forall t\in\left[0,T\right)\right\} 
\end{cases}\label{eq:3.2}
\end{align}
where $\widehat{r_{s}}=r_{s}-\mu_{s}$, $u_{s}=e^{\int_{s}^{T}\mu_{u}du}\sigma_{s}^{T}\pi_{s}$,
$y_{t}=x_{t}e^{\int_{t}^{T}\mu_{u}du}$. \\

\noindent $\mathbb{E}_{t}\left[Y_{T}\right]=Y_{t},\forall t\in\left[0,T\right)$
implies that $Y$must be a martingale as 
\begin{equation}
\forall\hat{t}\in\left(t,T\right),\mathbb{E}_{t}\left[Y_{\hat{t}}\right]=\mathbb{E}_{t}\left[\mathbb{E}_{\hat{t}}\left(Y_{T}\right)\right]=\mathbb{E}_{t}\left[Y_{T}\right]=Y_{t}
\end{equation}
which is an admissible constraint on $u$ and thus we could firstly
consider the family of following problems

\noindent 
\begin{align}
\underset{u}{\min} & \mbox{ \ensuremath{}\ensuremath{}\ensuremath{}\ensuremath{}\mbox{ \ensuremath{}\ensuremath{}\ensuremath{}\ensuremath{}}\mbox{ \ensuremath{}\ensuremath{}\ensuremath{}\ensuremath{}}}\mathbb{E}_{t}\left[-\frac{\left(Y_{T}-Y_{t}\right)^{3}}{3}\right]\nonumber \\
s.t. & \mbox{ \ensuremath{}\ensuremath{}\ensuremath{}\ensuremath{}\mbox{ \ensuremath{}\ensuremath{}\ensuremath{}\ensuremath{}}}\begin{cases}
dY_{s}=\left(\widehat{r_{s}}Y_{s}+u_{s}^{T}\theta_{s}\right)ds+u_{s}^{T}dW_{s}\\
Y_{t}=y_{t}\\
u\in U_{ad}^{u}=\left\{ u\left|\right.u\in L_{\mathcal{F}}^{2}\left(0,T;\mathbb{R}^{d}\right)\right\} 
\end{cases}\label{eq:3.3}
\end{align}

\begin{prop}
\label{thm:3.1}If an equilibrium solution $\left(u^{*},Y^{*}\right)$
to the above family of problems (\ref{eq:3.3}) satisfies\textup{
$\widehat{r_{s}}Y_{s}^{*}+\left(u_{s}^{*}\right)^{T}\theta_{s}=0$
}for any $s\in\left[0,T\right)$, then $\left(u^{*},Y^{*}\right)$
is also equilibrium for the family of problems in (\ref{eq:3.2})\end{prop}
\begin{proof}
Suppose $\left(u^{*},Y^{*}\right)$ is an equilibrium solution to
(\ref{eq:3.3}) which satisfies $\widehat{r_{s}}Y_{s}^{*}+\left(u_{s}^{*}\right)^{T}\theta_{s}=0$
for any $s\in\left[0,T\right)$, then we have that $dY_{s}^{*}=u_{s}^{*T}dW_{s},\forall s\in\left[0,T\right)$$ $
which implies that $Y^{*}$ is a martingale. So we have $\mathbb{E}_{t}\left[Y_{T}^{*}\right]=Y_{t}^{*},\forall t\in\left[0,T\right)$
and thus $u^{*}\in\overline{U_{ad}^{u}}$, which can be used together
with definition \ref{2.1} to deduce that $\left(u^{*},Y^{*}\right)$
is also equilibrium for the family of problems (\ref{eq:3.2})
\end{proof}
\noindent By the above proposition, we try to solve the family of
problems (\ref{eq:3.3}) instead, that is
\begin{align}
\underset{u}{\min} & \mbox{ \ensuremath{}\ensuremath{}\ensuremath{}\ensuremath{}\mbox{ \ensuremath{}\ensuremath{}\ensuremath{}\ensuremath{}}\mbox{ \ensuremath{}\ensuremath{}\ensuremath{}\ensuremath{}}}\mathbb{E}_{t}\left[h\left(Y_{T}-Y_{t}\right)\right]\nonumber \\
s.t. & \mbox{ \ensuremath{}\ensuremath{}\ensuremath{}\ensuremath{}\mbox{ \ensuremath{}\ensuremath{}\ensuremath{}\ensuremath{}}}\begin{cases}
dY_{s}=\left(\widehat{r_{s}}Y_{s}+u_{s}^{T}\theta_{s}\right)ds+u_{s}^{T}dW_{s}\\
Y_{t}=y_{t}\\
u\in U_{ad}^{u}=\left\{ u\left|\right.u\in L_{\mathcal{F}}^{2}\left(0,T;\mathbb{R}^{d}\right)\right\} 
\end{cases}\label{eq:3.4}
\end{align}

\noindent where $h\left(x\right)=-\frac{x^{3}}{3}$, and this is again
exactly a family of problems of the form in (\ref{eq:1.3}) , so we
have the following system of BSDEs by (\ref{eq:1.4}) and (\ref{eq:1.5})
\begin{equation}
\begin{cases}
dp_{s}^{t} & =-\widehat{r_{s}}p_{s}^{t}ds+\left(q_{s}^{t}\right)^{T}dW_{s},\mbox{ \ensuremath{}s\ensuremath{\in\left[t,T\right]}}\\
p_{T}^{t} & =\frac{dh\left(Y_{T}^{*}-Y_{t}^{*}\right)}{dx}=-\left(Y_{T}^{*}-Y_{t}^{*}\right)^{2}
\end{cases}\label{eq:3.5}
\end{equation}
\begin{equation}
\begin{cases}
dP_{s}^{t} & =-2\widehat{r_{s}}P_{s}^{t}ds+\left(Q_{s}^{t}\right)^{T}dW_{s},\mbox{ \ensuremath{}s\ensuremath{\in\left[t,T\right]}}\\
P_{T}^{t} & =\frac{d^{2}h\left(Y_{T}^{*}-Y_{t}^{*}\right)}{dx^{2}}=-2\left(Y_{T}^{*}-Y_{t}^{*}\right)
\end{cases}\label{eq:3.5-2}
\end{equation}

\begin{prop}
\label{prop:3.2}If (\textup{\ref{eq:3.4}}) and (\ref{eq:3.5}) admit
a solution $\left(u^{*},Y^{*},p^{t},q^{t}\right)$ for any $t\in\left[0,T\right)$
s.t.
\begin{equation}
\begin{cases}
u^{*}\in L_{\mathcal{F}}^{2}\left(0,T;\mathbb{R}^{d}\right)\\
\mathbb{E}_{t}\int_{t}^{T}\left|\Lambda_{s}^{t}\right|ds<\infty\mbox{ and \ensuremath{\underset{s\left\downarrow t\right.}{\lim}\mathbb{E}_{t}\left[\Lambda_{s}^{t}\right]=0}} & where\mbox{ }\Lambda_{s}^{t}=p_{s}^{t}\theta_{s}+q_{s}^{t}.\\
\widehat{r_{t}}Y_{t}^{*}+\left(u_{t}^{*}\right)^{T}\theta_{t}=0
\end{cases}\label{eq:3.7}
\end{equation}
then $u^{*}$ is an equilibrium control for the family of problems
(\ref{eq:3.2}) \end{prop}
\begin{proof}
Since the additional admissible condition $\widehat{r_{s}}Y_{s}^{*}+\left(u_{s}^{*}\right)^{T}\theta_{s}=0,\forall s\in\left[0,T\right)$
implies that $\mathbb{E}_{t}\left[Y_{T}^{*}\right]=Y_{t}^{*},\forall t\in\left[0,T\right)$,
so we have that $\mathbb{E}_{t}$$\left[\frac{d^{2}h\left(Y_{T}^{*}-Y_{t}^{*}\right)}{dx^{2}}\right]=$$-2\left(\mathbb{E}_{t}\left[Y_{T}^{*}\right]-Y_{t}^{*}\right)=0,\forall t\in\left[0,T\right)$.
Thus one of the sufficient condition for equilibrium, i.e. $\mathbb{E}_{t}$$\left[\frac{d^{2}h\left(Y_{T}^{*}-Y_{t}^{*}\right)}{dx^{2}}\right]\geq0$,
in (\ref{eq:1.11}) under theorem \ref{thm:1.3} is covered by the
condition $\widehat{r_{s}}Y_{s}^{*}+\left(u_{s}^{*}\right)^{T}\theta_{s}=0,\forall s\in\left[0,T\right)$,
and here we could make a replacement. Then by combining theorem \ref{thm:1.3}
and proposition \ref{thm:3.1}, we deduce that $u^{*}$ is an equilibrium
for the family of problems (\ref{eq:3.2}).
\end{proof}

\subsection{Details of finding a potential equilibrium}

\noindent By the assumptions we have made, $\widehat{r}$ and $\sigma$
are deterministic and bounded, $\mu^{x}$ is bounded. Here we also
assume that $\mu^{x}$ is deterministic, i.e. $\theta$ is deterministic
and bounded. Then for any $t\in\left[0,T\right)$, we make the following
Ansatz
\begin{equation}
p_{s}^{t}=-M_{s}\left(Y_{s}^{*}\right)^{2}+N_{s}Y_{t}^{*}Y_{s}^{*}-\Gamma_{s}\left(Y_{t}^{*}\right)^{2},s\in\left[t,T\right]\label{eq:3.8}
\end{equation}
where $M,N,\Gamma$ are deterministic functions which are differentiable
with $M_{T}=1,N_{T}=2,\Gamma_{T}=1$ 

\noindent by It$\hat{\mbox{o}}$ formula we have

\begin{eqnarray}
d\left(Y_{s}^{*}\right)^{2} & = & 2Y_{s}^{*}dY_{s}^{*}+d\left\langle Y_{s}^{*}\right\rangle \nonumber \\
 & = & \left[2Y_{s}^{*}\left(\widehat{r_{s}}Y_{s}^{*}+\left(u_{s}^{*}\right)^{T}\theta_{s}\right)+\left(u_{s}^{*}\right)^{T}u_{s}^{*}\right]ds+2Y_{s}^{*}\left(u_{s}^{*}\right)^{T}dW_{s}
\end{eqnarray}
\begin{eqnarray}
d\left[M_{s}\left(Y_{s}^{*}\right)^{2}\right] & = & M_{s}d\left(Y_{s}^{*}\right)^{2}+\left(Y_{s}^{*}\right)^{2}dM_{s}\nonumber \\
 & = & \left\{ M_{s}\left[2Y_{s}^{*}\left(\widehat{r_{s}}Y_{s}^{*}+\left(u_{s}^{*}\right)^{T}\theta_{s}\right)+\left(u_{s}^{*}\right)^{T}u_{s}^{*}\right]+\left(Y_{s}^{*}\right)^{2}M_{s}^{\prime}\right\} ds\nonumber \\
 &  & +2M_{s}Y_{s}^{*}\left(u_{s}^{*}\right)^{T}dW_{s}
\end{eqnarray}
\begin{eqnarray}
d\left(N_{s}Y_{t}^{*}Y_{s}^{*}\right) & = & Y_{t}^{*}\left[N_{s}dY_{s}^{*}+Y_{s}^{*}dN_{s}\right]\nonumber \\
 & = & Y_{t}^{*}\left[N_{s}\left(\widehat{r_{s}}Y_{s}^{*}+\left(u_{s}^{*}\right)^{T}\theta_{s}\right)+Y_{s}^{*}N_{s}^{\prime}\right]ds+Y_{t}^{*}N_{s}\left(u_{s}^{*}\right)^{T}dW_{s}
\end{eqnarray}

\noindent So by applying It$\hat{\mbox{o}}$ formula to (\ref{eq:3.8})
with respect to $s$ we could get
\begin{eqnarray}
dp_{s}^{t} & = & -d\left[M_{s}\left(Y_{s}^{*}\right)^{2}\right]+d\left(N_{s}Y_{t}^{*}Y_{s}^{*}\right)-d\left[\Gamma_{s}\left(Y_{t}^{*}\right)^{2}\right],\nonumber \\
 & = & \left[-\left(Y_{s}^{*}\right)^{2}M_{s}^{\mathbf{\prime}}-M_{s}\left(u_{s}^{*}\right)^{T}u_{s}^{*}-2M_{s}Y_{s}^{*}\left(\widehat{r_{s}}Y_{s}^{*}+\left(u_{s}^{*}\right)^{T}\theta_{s}\right)\right]ds\nonumber \\
 &  & +\left[Y_{t}^{*}Y_{s}^{*}N_{s}^{\prime}+Y_{t}^{*}N_{s}\left(\widehat{r_{s}}Y_{s}^{*}+\left(u_{s}^{*}\right)^{T}\theta_{s}\right)\right]ds-\Gamma_{s}^{\prime}\left(Y_{t}^{*}\right)^{2}ds\nonumber \\
 &  & +\left[-2M_{s}Y_{s}^{*}+Y_{t}^{*}N_{s}\right]\left(u_{s}^{*}\right)^{T}dW_{s}\label{eq:3.9}
\end{eqnarray}
by comparing the $dW$ terms of $dp^{t}$ in (\ref{eq:3.5}) and (\ref{eq:3.9}),
we get that
\begin{equation}
q_{s}^{t}=\left[-2M_{s}Y_{s}^{*}+Y_{t}^{*}N_{s}\right]u_{s}^{*},\mbox{ \ensuremath{}s\ensuremath{\in\left[t,T\right]}}
\end{equation}
we again hope to find a possible linear feedback $u^{*}$ and as before
try by setting 
\begin{equation}
0=\Lambda_{s}^{s}=p_{s}^{s}\theta_{s}+q_{s}^{s},\mbox{ \ensuremath{}s\ensuremath{\in\left[0,T\right]}}
\end{equation}
which leads to the equation
\begin{equation}
\left[\left(-M_{s}+N_{s}-\Gamma_{s}\right)Y_{s}^{*}\theta_{s}+\left(N_{s}-2M_{s}\right)u_{s}^{*}\right]Y_{s}^{*}=0
\end{equation}
from which we get
\begin{equation}
u_{s}^{*}=\alpha_{s}\theta_{s}Y_{s}^{*}\label{eq:5.13}
\end{equation}
where 
\begin{equation}
\alpha_{s}=\begin{cases}
\frac{-M_{s}+N_{s}-\Gamma_{s}}{2M_{s}-N_{s}}, & \forall s\in\left[0,T\right)\\
\underset{s\left\uparrow T\right.}{\lim}\frac{-M_{s}+N_{s}-\Gamma_{s}}{2M_{s}-N_{s}}, & s=T
\end{cases}\label{eq:3.13}
\end{equation}
based on the assumption that $2M_{s}-N_{s}\neq0,\forall s\in\left[0,T\right)$
and $\underset{s\left\uparrow T\right.}{\lim}\frac{-M_{s}+N_{s}-\Gamma_{s}}{2M_{s}-N_{s}}$
exists. \\

\noindent By comparing the $ds$ terms of $dp^{t}$ in (\ref{eq:3.5})
and (\ref{eq:3.9}), we get that
\begin{eqnarray*}
 &  & -\widehat{r_{s}}\left[-M_{s}\left(Y_{s}^{*}\right)^{2}+N_{s}Y_{t}^{*}Y_{s}^{*}-\Gamma_{s}\left(Y_{t}^{*}\right)^{2}\right]\\
 & = & \left[-\left(Y_{s}^{*}\right)^{2}M_{s}^{\mathbf{\prime}}-M_{s}\left(u_{s}^{*}\right)^{T}u_{s}^{*}-2M_{s}Y_{s}^{*}\left(\widehat{r_{s}}Y_{s}^{*}+\left(u_{s}^{*}\right)^{T}\theta_{s}\right)\right]\\
 &  & +\left[Y_{t}^{*}Y_{s}^{*}N_{s}^{\prime}+Y_{t}^{*}N_{s}\left(\widehat{r_{s}}Y_{s}^{*}+\left(u_{s}^{*}\right)^{T}\theta_{s}\right)\right]-\Gamma_{s}^{\prime}\left(Y_{t}^{*}\right)^{2}\\
\\
\Leftrightarrow &  & -\widehat{r_{s}}\left[-M_{s}\left(Y_{s}^{*}\right)^{2}+N_{s}Y_{t}^{*}Y_{s}^{*}-\Gamma_{s}\left(Y_{t}^{*}\right)^{2}\right]\\
 & = & \left[-\left(Y_{s}^{*}\right)^{2}M_{s}^{\mathbf{\prime}}-M_{s}\alpha_{s}^{2}\left|\theta_{s}\right|^{2}\left(Y_{s}^{*}\right)^{2}-2M_{s}Y_{s}^{*}\left(\widehat{r_{s}}Y_{s}^{*}+\alpha_{s}\left|\theta_{s}\right|^{2}Y_{s}^{*}\right)\right]\\
 &  & +\left[Y_{t}^{*}Y_{s}^{*}N_{s}^{\prime}+Y_{t}^{*}N_{s}\left(\widehat{r_{s}}Y_{s}^{*}+\alpha_{s}\left|\theta_{s}\right|^{2}Y_{s}^{*}\right)\right]-\Gamma_{s}^{\prime}\left(Y_{t}^{*}\right)^{2}\\
\\
\Leftrightarrow &  & \widehat{r_{s}}M_{s}\left(Y_{s}^{*}\right)^{2}-\widehat{r_{s}}N_{s}Y_{t}^{*}Y_{s}^{*}+\widehat{r_{s}}\Gamma_{s}\left(Y_{t}^{*}\right)^{2}\\
 & = & \left[-M_{s}^{\mathbf{\prime}}-M_{s}\alpha_{s}^{2}\left|\theta_{s}\right|^{2}-2M_{s}\left(\widehat{r_{s}}+\alpha_{s}\left|\theta_{s}\right|^{2}\right)\right]\left(Y_{s}^{*}\right)^{2}\\
 &  & +\left[N_{s}^{\prime}+N_{s}\left(\widehat{r_{s}}+\alpha_{s}\left|\theta_{s}\right|^{2}\right)\right]Y_{t}^{*}Y_{s}^{*}-\Gamma_{s}^{\prime}\left(Y_{t}^{*}\right)^{2}
\end{eqnarray*}
after rearrangement we get
\begin{eqnarray}
0 & = & \left(\Gamma_{s}^{\prime}+\widehat{r_{s}}\Gamma_{s}\right)\left(Y_{t}^{*}\right)^{2}-\left(N_{s}^{\prime}+2\widehat{r_{s}}N_{s}+\alpha_{s}\left|\theta_{s}\right|^{2}N_{s}\right)Y_{t}^{*}Y_{s}^{*}\nonumber \\
 &  & +\left[M_{s}^{\mathbf{\prime}}+3\widehat{r_{s}}M_{s}+2\alpha_{s}\left|\theta_{s}\right|^{2}M_{s}+\alpha_{s}^{2}\left|\theta_{s}\right|^{2}M_{s}\right]\left(Y_{s}^{*}\right)^{2}
\end{eqnarray}
which leads to the following system of ODEs 
\begin{align}
 & \begin{cases}
M_{s}^{\mathbf{\prime}}+\left(3\widehat{r_{s}}+2\alpha_{s}\left|\theta_{s}\right|^{2}+\alpha_{s}^{2}\left|\theta_{s}\right|^{2}\right)M_{s}=0, & s\in\left[0,T\right]\\
M_{T}=1
\end{cases}\label{eq:3.15}\\
 & \begin{cases}
N_{s}^{\prime}+\left(2\widehat{r_{s}}+\alpha_{s}\left|\theta_{s}\right|^{2}\right)N_{s}=0, & s\in\left[0,T\right]\\
N_{T}=2
\end{cases}\label{eq:3.16}\\
 & \begin{cases}
\Gamma_{s}^{\prime}+\widehat{r_{s}}\Gamma_{s}=0, & s\in\left[0,T\right]\\
\Gamma_{T}=1
\end{cases}\label{eq:3.17}
\end{align}
the solution to equation (\ref{eq:3.17}) is $\Gamma_{s}=e^{\int_{s}^{T}\widehat{r_{u}}du}$,
which makes the unsettled system contains only (\ref{eq:3.15}) and
(\ref{eq:3.16}) as follows 
\begin{align}
 & \begin{cases}
M_{s}^{\mathbf{\prime}}+\left(3\widehat{r_{s}}+2\alpha_{s}\left|\theta_{s}\right|^{2}+\alpha_{s}^{2}\left|\theta_{s}\right|^{2}\right)M_{s}=0, & s\in\left[0,T\right]\\
M_{T}=1\\
N_{s}^{\prime}+\left(2\widehat{r_{s}}+\alpha_{s}\left|\theta_{s}\right|^{2}\right)N_{s}=0, & s\in\left[0,T\right]\\
N_{T}=2
\end{cases}\label{eq:3.18}
\end{align}

\begin{rem}
\noindent If we need $\underset{s\left\uparrow T\right.}{\lim}\frac{-M_{s}+N_{s}-\Gamma_{s}}{2M_{s}-N_{s}}\neq0$,
we also need to assume that $\underset{s\left\uparrow T\right.}{\lim}\frac{\widehat{r_{s}}}{\left|\theta_{s}\right|^{2}}\leq\frac{9}{16}$
which is a necessary condition for the existence of none zero $\underset{s\left\uparrow T\right.}{\lim}\frac{-M_{s}+N_{s}-\Gamma_{s}}{2M_{s}-N_{s}}$
where $\alpha,M,N,\Gamma$ are those defined in (\ref{eq:3.13}),
(\ref{eq:3.15}), (\ref{eq:3.16}) and (\ref{eq:3.17}).\end{rem}
\begin{proof}
Suppose $\underset{s\left\uparrow T\right.}{\lim}\frac{-M_{s}+N_{s}-\Gamma_{s}}{2M_{s}-N_{s}}$
exists and does not equal to zero. Let $\bar{\theta}=\underset{s\left\uparrow T\right.}{\lim}\theta_{s}$,
$\bar{r}=\underset{s\left\uparrow T\right.}{\lim}$$\widehat{r_{s}}$
and $\bar{\alpha}=\underset{s\left\uparrow T\right.}{\lim}\frac{-M_{s}+N_{s}-\Gamma_{s}}{2M_{s}-N_{s}}$
which also implies $\bar{\alpha}=\underset{s\left\uparrow T\right.}{\lim}\alpha_{s}$
by the definition of $\alpha$ in (\ref{eq:3.13}). Since we have
\begin{eqnarray}
\underset{s\left\uparrow T\right.}{\lim}\frac{-M_{s}^{\prime}+N_{s}^{\prime}-\Gamma_{s}^{\prime}}{2M_{s}^{\prime}-N_{s}^{\prime}} & = & \underset{s\left\uparrow T\right.}{\lim}\frac{\left(3\widehat{r_{s}}+2\alpha_{s}\left|\theta_{s}\right|^{2}+\alpha_{s}^{2}\left|\theta_{s}\right|^{2}\right)M_{s}-\left(2\widehat{r_{s}}+\alpha_{s}\left|\theta_{s}\right|^{2}\right)N_{s}+\widehat{r_{s}}\Gamma_{s}}{-2\left(3\widehat{r_{s}}+2\alpha_{s}\left|\theta_{s}\right|^{2}+\alpha_{s}^{2}\left|\theta_{s}\right|^{2}\right)M_{s}+\left(2\widehat{r_{s}}+\alpha_{s}\left|\theta_{s}\right|^{2}\right)N_{s}}\nonumber \\
 & = & \underset{s\left\uparrow T\right.}{\lim}\frac{\left(3\widehat{r_{s}}+2\alpha_{s}\left|\theta_{s}\right|^{2}+\alpha_{s}^{2}\left|\theta_{s}\right|^{2}\right)-2\left(2\widehat{r_{s}}+\alpha_{s}\left|\theta_{s}\right|^{2}\right)+\widehat{r_{s}}}{-2\left(3\widehat{r_{s}}+2\alpha_{s}\left|\theta_{s}\right|^{2}+\alpha_{s}^{2}\left|\theta_{s}\right|^{2}\right)+2\left(2\widehat{r_{s}}+\alpha_{s}\left|\theta_{s}\right|^{2}\right)}\nonumber \\
 & = & \underset{s\left\uparrow T\right.}{\lim}\frac{\alpha_{s}^{2}\left|\theta_{s}\right|^{2}}{-2\alpha_{s}^{2}\left|\theta_{s}\right|^{2}-2\alpha_{s}\left|\theta_{s}\right|^{2}-2\widehat{r_{s}}}\nonumber \\
 & = & \frac{\bar{\alpha}^{2}\left|\bar{\theta}\right|^{2}}{-2\bar{\alpha}^{2}\left|\bar{\theta}\right|^{2}-2\bar{\alpha}\left|\bar{\theta}\right|^{2}-2\bar{r}}\nonumber \\
 & = & \frac{\bar{\alpha}^{2}}{-2\bar{\alpha}^{2}-2\bar{\alpha}-2\eta}
\end{eqnarray}
where $\eta=\frac{\bar{r}}{\left|\bar{\theta}\right|^{2}}$, thus
we must have
\[
\bar{\alpha}=\underset{s\left\uparrow T\right.}{\lim}\frac{-M_{s}+N_{s}-\Gamma_{s}}{2M_{s}-N_{s}}=\underset{s\left\uparrow T\right.}{\lim}\frac{-M_{s}^{\prime}+N_{s}^{\prime}-\Gamma_{s}^{\prime}}{2M_{s}^{\prime}-N_{s}^{\prime}}=\frac{\bar{\alpha}^{2}}{-2\bar{\alpha}^{2}-2\bar{\alpha}-2\eta}
\]
which after the rearrangement is 
\begin{equation}
\left(2\bar{\alpha}^{2}+3\bar{\alpha}+2\eta\right)\bar{\alpha}=0
\end{equation}
which implies that $\eta\leq\frac{9}{16}$ is a necessary condition
for the existence of none zero $\bar{\alpha}$
\end{proof}

\subsection{Conditions for obtaining an equilibrium for our problem}
\begin{thm}
\label{thm:5.4}If the system of ODEs (\ref{eq:3.18}) admits a solution
$\left(M,N\right)$ s.t. the corresponding \textup{$\alpha$} defined
in (\ref{eq:3.13}) satisfies that $\widehat{r_{t}}+\alpha_{t}\left|\theta_{t}\right|^{2}=0,\forall t\in\left[0,T\right)$,
then $u^{*}$ is an equilibrium control for the family of problems
in (\ref{eq:3.2})\end{thm}
\begin{proof}
Suppose $\left(M,N\right)$ is a solution to (\ref{eq:3.18}) s.t.
$\widehat{r_{t}}+\alpha_{t}\left|\theta_{t}\right|^{2}=0,\forall t\in\left[0,T\right)$.
Since the deterministic $M,N,\Gamma$ are continuous and thus are
bounded on $\left[0,T\right]$, we have that $\alpha$ is also deterministic
and bounded on $\left[0,T\right]$ according to (\ref{eq:3.13}).
We have in (\ref{eq:5.13}) for $s\in\left[0,T\right]$ that 
\begin{equation}
u_{s}^{*}=\alpha_{s}\theta_{s}Y_{s}^{*}
\end{equation}
and thus we have
\begin{align}
dY_{s}^{*} & =\left(\widehat{r_{s}}Y_{s}^{*}+\left(u_{s}^{*}\right)^{T}\theta_{s}\right)ds+\left(u_{s}^{*}\right)^{T}dW_{s}\nonumber \\
 & =Y_{s}^{*}\left[\left(\widehat{r_{s}}+\alpha_{s}\left|\theta_{s}\right|^{2}\right)ds+\alpha_{s}\left(\theta_{s}\right)^{T}dW_{s}\right]
\end{align}
which leads to
\begin{equation}
Y_{t}^{*}=y_{0}e^{\int_{0}^{t}\widehat{r_{s}}+\alpha_{s}\left|\theta_{s}\right|^{2}-\frac{1}{2}\alpha_{s}^{2}\left|\theta_{s}\right|^{2}ds+\int_{0}^{t}\alpha_{s}\left(\theta_{s}\right)^{T}dW_{s}}
\end{equation}
thus

\begin{eqnarray}
\mathbb{E}\left[\underset{t\in\left[0,T\right]}{\sup}\left(Y_{t}^{*}\right)^{2}\right] & = & \mathbb{E}\left[\underset{t\in\left[0,T\right]}{\sup}\left[\left(y_{0}e^{\int_{0}^{t}\widehat{r_{s}}+\alpha_{s}\left|\theta_{s}\right|^{2}ds}\right)^{2}\left(e^{\int_{0}^{t}-\frac{1}{2}\alpha_{s}^{2}\left|\theta_{s}\right|^{2}ds+\int_{0}^{t}\alpha_{s}\left(\theta_{s}\right)^{T}dW_{s}}\right)^{2}\right]\right]\nonumber \\
 & \leq & \left[\underset{t\in\left[0,T\right]}{\sup}\left(y_{0}e^{\int_{0}^{t}\widehat{r_{s}}+\alpha_{s}\left|\theta_{s}\right|^{2}ds}\right)^{2}\right]\mathbb{E}\left[\underset{t\in\left[0,T\right]}{\sup}\left(e^{\int_{0}^{t}-\frac{1}{2}\alpha_{s}^{2}\left|\theta_{s}\right|^{2}ds+\int_{0}^{t}\alpha_{s}\left(\theta_{s}\right)^{T}dW_{s}}\right)^{2}\right]\nonumber \\
 & \leq & \left[\underset{t\in\left[0,T\right]}{\sup}\left(y_{0}e^{\int_{0}^{t}\widehat{r_{s}}+\alpha_{s}\left|\theta_{s}\right|^{2}ds}\right)^{2}\right]4\mathbb{E}\left[\left(e^{\int_{0}^{T}-\frac{1}{2}\alpha_{s}^{2}\left|\theta_{s}\right|^{2}ds+\int_{0}^{T}\alpha_{s}\left(\theta_{s}\right)^{T}dW_{s}}\right)^{2}\right]\nonumber \\
 &  & \mbox{by \ensuremath{L^{p}}Maximal Inequality}\nonumber \\
 & = & 4e^{\int_{0}^{T}\alpha_{s}^{2}\left|\theta_{s}\right|^{2}ds}\left[\underset{t\in\left[0,T\right]}{\sup}\left(y_{0}e^{\int_{0}^{t}\widehat{r_{s}}+\alpha_{s}\left|\theta_{s}\right|^{2}ds}\right)^{2}\right]\mathbb{E}\left[e^{\int_{0}^{T}-2\alpha_{s}^{2}\left|\theta_{s}\right|^{2}ds+\int_{0}^{T}2\alpha_{s}\left(\theta_{s}\right)^{T}dW_{s}}\right]\nonumber \\
 & = & 4e^{\int_{0}^{T}\alpha_{s}^{2}\left|\theta_{s}\right|^{2}ds}\underset{t\in\left[0,T\right]}{\sup}\left(y_{0}e^{\int_{0}^{t}\widehat{r_{s}}+\alpha_{s}\left|\theta_{s}\right|^{2}ds}\right)^{2}\nonumber \\
 & < & \infty
\end{eqnarray}
as $\alpha,\theta,\widehat{r}$ are bounded, which implies that 
\begin{equation}
Y^{*}\in L_{\mathcal{F}}^{2}\left(\Omega;C\left(0,T;\mathbb{R}\right)\right)\label{eq:3.23}
\end{equation}
and thus we have 
\begin{eqnarray}
\mathbb{E}\left[\int_{0}^{T}\left|u_{s}^{*}\right|^{2}ds\right] & = & \mathbb{E}\left[\int_{0}^{T}\alpha_{s}^{2}\left|\theta_{s}\right|^{2}\left(Y_{s}^{*}\right)^{2}ds\right]\nonumber \\
 & \leq & \underset{s\in\left[0,T\right]}{\sup}\left(\alpha_{s}^{2}\left|\theta_{s}\right|^{2}\right)\int_{0}^{T}\mathbb{E}\left[\underset{u\in\left[0,T\right]}{\sup}\left(Y_{u}^{*}\right)^{2}\right]ds\nonumber \\
 & = & \underset{s\in\left[0,T\right]}{\sup}\left(\alpha_{s}^{2}\left|\theta_{s}\right|^{2}\right)\mathbb{E}\left[\underset{u\in\left[0,T\right]}{\sup}\left(Y_{u}^{*}\right)^{2}\right]T\nonumber \\
 & < & \infty
\end{eqnarray}
which implies that
\begin{equation}
u^{*}\in L_{\mathcal{F}}^{2}\left(0,T;\mathbb{R}^{d}\right)\label{eq:3.25}
\end{equation}
also we have that for any $t\in\left[0,T\right)$ 
\begin{eqnarray}
\mbox{ }\Lambda_{s}^{t} & = & p_{s}^{t}\theta_{s}+q_{s}^{t}\nonumber \\
 & = & \left[-M_{s}\left(Y_{s}^{*}\right)^{2}+N_{s}Y_{t}^{*}Y_{s}^{*}-\Gamma_{s}\left(Y_{t}^{*}\right)^{2}\right]\theta_{s}+\left(Y_{t}^{*}N_{s}-2M_{s}Y_{s}^{*}\right)u_{s}^{*}\nonumber \\
 & = & \left[-M_{s}\left(Y_{s}^{*}\right)^{2}+N_{s}Y_{t}^{*}Y_{s}^{*}-\Gamma_{s}\left(Y_{t}^{*}\right)^{2}+\left(Y_{t}^{*}N_{s}-2M_{s}Y_{s}^{*}\right)\alpha_{s}Y_{s}^{*}\right]\theta_{s}\nonumber \\
 & = & \left[-\left(1+2\alpha_{s}\right)M_{s}\left(Y_{s}^{*}\right)^{2}+\left(1+\alpha_{s}\right)N_{s}Y_{t}^{*}Y_{s}^{*}-\Gamma_{s}\left(Y_{t}^{*}\right)^{2}\right]\theta_{s}
\end{eqnarray}
Since $\alpha,M,N,\Gamma,\theta$ are bounded, it is clearly by (\ref{eq:3.23})
that 
\begin{equation}
\mathbb{E}_{t}\int_{t}^{T}\left|\Lambda_{s}^{t}\right|ds<\infty\label{eq:3.27}
\end{equation}
and
\begin{equation}
\underset{s\left\downarrow t\right.}{\lim}\mathbb{E}_{t}\left[\Lambda_{s}^{t}\right]=\underset{s\left\downarrow t\right.}{\lim}\theta_{s}\mathbb{E}_{t}\left[-\left(1+2\alpha_{s}\right)M_{s}\left(Y_{s}^{*}\right)^{2}+\left(1+\alpha_{s}\right)N_{s}Y_{t}^{*}Y_{s}^{*}-\Gamma_{s}\left(Y_{t}^{*}\right)^{2}\right]
\end{equation}
since we have
\begin{align}
 & \underset{s\left\downarrow t\right.}{\lim}\mathbb{E}_{t}\left[-\left(1+2\alpha_{s}\right)M_{s}\left(Y_{s}^{*}\right)^{2}+\left(1+\alpha_{s}\right)N_{s}Y_{t}^{*}Y_{s}^{*}-\Gamma_{s}\left(Y_{t}^{*}\right)^{2}\right]\nonumber \\
= & \underset{s\left\downarrow t\right.}{\lim}\left\{ -\left(1+2\alpha_{s}\right)M_{s}\mathbb{E}_{t}\left[\left(Y_{s}^{*}\right)^{2}\right]+\left(1+\alpha_{s}\right)N_{s}Y_{t}^{*}\mathbb{E}_{t}\left[Y_{s}^{*}\right]-\Gamma_{s}\left(Y_{t}^{*}\right)^{2}\right\} \nonumber \\
\overset{(\ref{eq:3.23})}{=} & -\left(1+2\alpha_{t}\right)M_{t}\mathbb{E}_{t}\left[\underset{s\left\downarrow t\right.}{\lim}\left(Y_{s}^{*}\right)^{2}\right]+\left(1+\alpha_{t}\right)N_{t}Y_{t}^{*}\mathbb{E}_{t}\left[\underset{s\left\downarrow t\right.}{\lim}Y_{s}^{*}\right]-\Gamma_{t}\left(Y_{t}^{*}\right)^{2}\nonumber \\
 & \mbox{by Dominated Convergence}\nonumber \\
= & -\left(1+2\alpha_{t}\right)M_{t}\left(Y_{t}^{*}\right)^{2}+\left(1+\alpha_{t}\right)N_{t}\left(Y_{t}^{*}\right)^{2}-\Gamma_{t}\left(Y_{t}^{*}\right)^{2}\nonumber \\
= & \left[-M_{t}+N_{t}-\Gamma_{t}-\left(2M_{t}-N_{t}\right)\alpha_{t}\right]\left(Y_{t}^{*}\right)^{2}\nonumber \\
= & \left[\frac{-M_{t}+N_{t}-\Gamma_{t}}{2M_{t}-N_{t}}-\alpha_{t}\right]\left(2M_{t}-N_{t}\right)\left(Y_{t}^{*}\right)^{2}\nonumber \\
= & 0
\end{align}
and also $\theta$ is bounded, thus we deduce that 
\begin{equation}
\underset{s\left\downarrow t\right.}{\lim}\mathbb{E}_{t}\left[\Lambda_{s}^{t}\right]=0\label{eq:3.30}
\end{equation}
we also have $\widehat{r_{t}}+\alpha_{t}\left|\theta_{t}\right|^{2}=0,\forall t\in\left[0,T\right)$
which implies that
\begin{eqnarray}
 & \left(\widehat{r_{t}}+\alpha_{t}\left|\theta_{t}\right|^{2}\right)Y_{t}^{*} & =0,\forall t\in\left[0,T\right)\nonumber \\
\Rightarrow & \widehat{r_{t}}Y_{t}^{*}+\alpha_{t}\theta_{t}^{T}\theta_{t}Y_{t}^{*} & =0,\forall t\in\left[0,T\right)\nonumber \\
\Rightarrow & \widehat{r_{t}}Y_{t}^{*}+\left(u_{t}^{*}\right)^{T}\theta_{t} & =0,\forall t\in\left[0,T\right)
\end{eqnarray}
which means $\left(u^{*},Y^{*}\right)$ satisfies 
\begin{equation}
\widehat{r_{t}}Y_{t}^{*}+\left(u_{t}^{*}\right)^{T}\theta_{t}=0,\forall t\in\left[0,T\right)\label{eq:3.31}
\end{equation}
Then (\ref{eq:3.25}), (\ref{eq:3.27}), (\ref{eq:3.30}) and (\ref{eq:3.31})
are exactly the required conditions in (\ref{eq:3.7}) and we deduce
that $u^{*}$ is an equilibrium control for the family of problems
in (\ref{eq:3.2}) by proposition \ref{prop:3.2}.
\end{proof}

\subsection{A particular solution to our problem}

Although we had not managed to proved the general conditions for the
existence of the solutions for (\ref{eq:3.18}), we found two particular
solutions to (\ref{eq:3.18}) as follows.\\

\noindent We have got by (\ref{eq:3.17}) that $\Gamma_{s}=e^{\int_{s}^{T}\widehat{r_{u}}du}$.
Here we rearrange (\ref{eq:3.13}) to get
\begin{equation}
\left(2M_{s}-N_{s}\right)\alpha_{s}=-M_{s}+N_{s}-\Gamma_{s}
\end{equation}
which is equivalent to
\begin{equation}
-\left(2\alpha_{s}+1\right)M_{s}+\left(\alpha_{s}+1\right)N_{s}-\Gamma_{s}=0\label{eq:5.34}
\end{equation}
We could set two constant solutions for $\alpha$, i.e. $ $$\alpha_{s}=-\frac{1}{2}$
or $-1$ for $s\in\left[0,T\right]$, to solve the system of ODEs
(\ref{eq:3.18}) .\\

\noindent Firstly we set $\alpha_{s}=-\frac{1}{2}$ for $s\in\left[0,T\right]$,
then by (\ref{eq:5.34})we get 
\begin{eqnarray}
N_{s} & = & 2\Gamma_{s}\nonumber \\
 & = & 2e^{\int_{s}^{T}\widehat{r_{u}}du}\label{eq:5.35}
\end{eqnarray}
if we plug (\ref{eq:5.35}) and $\alpha_{s}=-\frac{1}{2}$ back into
(\ref{eq:3.16}), we get
\begin{equation}
\begin{cases}
\left(\widehat{r_{s}}-\frac{1}{2}\left|\theta_{s}\right|^{2}\right)N_{s}=0, & s\in\left[0,T\right]\\
N_{T}=2
\end{cases}
\end{equation}
then if we set $\widehat{r}=\frac{1}{2}\left|\theta\right|^{2}$,
i.e. we set our required return $ $$\mu=r-\frac{1}{2}\left|\theta\right|^{2}$,
we have that $N_{s}=2e^{\int_{s}^{T}\widehat{r_{u}}du}$ is a solution
to (\ref{eq:3.16}). Now we plug $\widehat{r}=\frac{1}{2}\left|\theta\right|^{2}$
and $\alpha_{s}=-\frac{1}{2}$ into (\ref{eq:3.15}) and get 
\begin{equation}
\begin{cases}
M_{s}^{\mathbf{\prime}}+\left(\widehat{r_{s}}+\frac{1}{4}\left|\theta_{s}\right|^{2}\right)M_{s}=0, & s\in\left[0,T\right]\\
M_{T}=1
\end{cases}
\end{equation}
by solving which we get $M_{s}=e^{\int_{s}^{T}\widehat{r_{u}}+\frac{1}{4}\left|\theta_{u}\right|^{2}du}$,
thus we get a solution for the system of ODEs (\ref{eq:3.18}) as
follows 
\begin{equation}
\begin{cases}
M_{s}=e^{\int_{s}^{T}\frac{3}{2}\widehat{r_{u}}du}, & s\in\left[0,T\right]\\
N_{s}=2e^{\int_{s}^{T}\widehat{r_{u}}du}, & s\in\left[0,T\right]
\end{cases}\label{eq:5.39}
\end{equation}
in this case, by (\ref{eq:5.13}) we get 
\begin{equation}
u_{s}^{*}=-\frac{1}{2}\theta_{s}Y_{s}^{*}\label{eq:5.41}
\end{equation}
and we verify that 
\begin{eqnarray}
 & \widehat{r_{t}}Y_{t}^{*}+\left(u_{t}^{*}\right)^{T}\theta_{t} & =\left(\widehat{r_{t}}-\frac{1}{2}\left|\theta_{t}\right|^{2}\right)Y_{t}^{*}\nonumber \\
 &  & =0,\forall t\in\left[0,T\right)\label{eq:5.40}
\end{eqnarray}
(\ref{eq:5.39}) and (\ref{eq:5.40}) deduce that $u^{*}$ in (\ref{eq:5.41})
is an equilibrium control for the family of problems in (\ref{eq:3.2})
by theorem \ref{thm:5.4}.\\

\noindent If we plug $u_{s}^{*}=e^{\int_{s}^{T}\mu_{u}du}\sigma_{s}^{T}\pi_{s}^{*}$
and $Y_{s}^{*}=X_{s}^{*}e^{\int_{s}^{T}\mu_{u}du}$$ $ into (\ref{eq:5.41}),
then we get
\begin{equation}
e^{\int_{s}^{T}\mu_{u}du}\sigma_{s}^{T}\pi_{s}^{*}=-\frac{1}{2}\theta_{s}X_{s}^{*}e^{\int_{s}^{T}\mu_{u}du}
\end{equation}
and thus get 
\begin{equation}
\pi_{s}^{*}=-\frac{1}{2}\left(\sigma_{s}^{-1}\right)^{T}\theta_{s}X_{s}^{*}
\end{equation}
This says that when we have our required return $\mu=r-\frac{1}{2}\left|\theta\right|^{2}$,
we could find an equilibrium control $\pi^{*}=-\frac{1}{2}\left(\sigma^{-1}\right)^{T}\theta X^{*}$
for the family of problems (\ref{eq:3.1}), although it sounds a bit
unusual as our required return $\mu$ is below the risk free rate
$r$. \\
\\
If we set $\alpha_{s}=-1$ then we could not deduce that the resulting
$u^{*}$ is equilibrium by our theorem, which is explained as follows.
We set $\alpha_{s}=-1$ for $s\in\left[0,T\right]$, then by (\ref{eq:5.34})we
get 
\begin{eqnarray}
M_{s} & = & \Gamma_{s}\nonumber \\
 & = & e^{\int_{s}^{T}\widehat{r_{u}}du}\label{eq:5.35-1}
\end{eqnarray}
if we plug (\ref{eq:5.35-1}) and $\alpha_{s}=-1$ back into (\ref{eq:3.15}),
we get
\begin{equation}
\begin{cases}
\left(2\widehat{r_{s}}-\left|\theta_{s}\right|^{2}\right)M_{s}=0, & s\in\left[0,T\right]\\
M_{T}=1
\end{cases}
\end{equation}
Again we set $\widehat{r}=\frac{1}{2}\left|\theta\right|^{2}$, i.e.
we set our required return $ $$\mu=r-\frac{1}{2}\left|\theta\right|^{2}$
which happens to be the same as that in the case $\alpha_{s}=-\frac{1}{2}$,
we have that $M_{s}=e^{\int_{s}^{T}\widehat{r_{u}}du}$ is a solution
to (\ref{eq:3.15}). Now we plug $\widehat{r}=\frac{1}{2}\left|\theta\right|^{2}$
and $\alpha_{s}=-1$ into (\ref{eq:3.16}) and get 
\begin{equation}
\begin{cases}
N_{s}^{\prime}=0, & s\in\left[0,T\right]\\
N_{T}=2
\end{cases}
\end{equation}
by solving which we get $N_{s}=2$, thus we get a solution for the
system of ODEs (\ref{eq:3.18}) as follows 
\begin{equation}
\begin{cases}
M=e^{\int_{s}^{T}\widehat{r_{u}}du}, & s\in\left[0,T\right]\\
N=2, & s\in\left[0,T\right]
\end{cases}\label{eq:5.39-1}
\end{equation}
in this case, by (\ref{eq:5.13}) we get 
\begin{equation}
u_{s}^{*}=-\theta_{s}Y_{s}^{*}\label{eq:5.41-1}
\end{equation}
and we verify that 
\begin{eqnarray}
 & \widehat{r_{t}}Y_{t}^{*}+\left(u_{t}^{*}\right)^{T}\theta_{t} & =\left(\widehat{r_{t}}-\left|\theta_{t}\right|^{2}\right)Y_{t}^{*}\nonumber \\
 &  & =-\frac{1}{2}\left|\theta_{t}\right|^{2}Y_{t}^{*}\nonumber \\
 &  & \neq0,\exists t\in\left[0,T\right)\label{eq:5.40-1}
\end{eqnarray}
 (\ref{eq:5.40-1}) implies that we cannot deduce $u^{*}$ in (\ref{eq:5.41-1})
is an equilibrium control for the family of problems in (\ref{eq:3.2})
by our theorem \ref{thm:5.4}. Since for our required return $\mu=r-\frac{1}{2}\left|\theta\right|^{2}$
we have already found an equilibrium control $\pi^{*}=-\frac{1}{2}\left(\sigma^{-1}\right)^{T}\theta X^{*}$
for the family of problems (\ref{eq:3.1}), thus the result here dose
not matter.

\newpage{}

\section{Utility function\textmd{ }$h\left(x\right)=\frac{x^{4}}{4}$ for
strong risk aversion }

Some investors have strong risk aversion and they would like to use
a kind of utility function that we use here. In this section, we want
to solve the portfolio selection problem when we choose to use $h(x)=\frac{x^{4}}{4}$
as the utility function. That means we want to solve the family of
following problems for any $t\in\left[0,T\right)$
\begin{align}
\underset{\pi}{\min} & \mbox{ \ensuremath{}\ensuremath{}\ensuremath{}\ensuremath{}\mbox{ \ensuremath{}\ensuremath{}\ensuremath{}\ensuremath{}}\mbox{ \ensuremath{}\ensuremath{}\ensuremath{}\ensuremath{}}}\mathbb{E}_{t}\left[\frac{\left(X_{T}-X_{t}e^{\int_{t}^{T}\mu_{s}ds}\right)^{4}}{4}\right]\nonumber \\
s.t. & \mbox{ \ensuremath{}\ensuremath{}\ensuremath{}\ensuremath{}\mbox{ \ensuremath{}\ensuremath{}\ensuremath{}\ensuremath{}}}\begin{cases}
dX_{s}=\left(r_{s}X_{s}+\pi_{s}^{T}\sigma_{s}\theta_{s}\right)ds+\pi_{s}^{T}\sigma_{s}dW_{s}\\
X_{t}=x_{t}\\
\pi\in\overline{U_{ad}^{\pi}}=\left\{ \pi\left|\right.\pi\in L_{\mathcal{F}}^{2}\left(0,T;\mathbb{R}^{d}\right)\mbox{and }\mathbb{E}_{t}\left[X_{T}\right]=X_{t}e^{\int_{t}^{T}\mu_{s}ds},\forall t\in\left[0,T\right)\right\} 
\end{cases}\label{eq:3.1-1}
\end{align}
where $\mu$ as usual is our required return process which is bounded
and deterministic.

\subsection{Transformation of our problem}

Here we again use same approach which is used above for $h(x)=-\frac{x^{3}}{3}$
to solve our problem. Again by letting for any $s\in\left[0,T\right]$
\begin{equation}
Y_{s}=X_{s}e^{\int_{s}^{T}\mu_{u}du}
\end{equation}
we have the family of following problems for any $t\in\left[0,T\right)$,
which is equivalent to the above family of problems (\ref{eq:3.1-1})
\begin{align}
\underset{u}{\min} & \mbox{ \ensuremath{}\ensuremath{}\ensuremath{}\ensuremath{}\mbox{ \ensuremath{}\ensuremath{}\ensuremath{}\ensuremath{}}\mbox{ \ensuremath{}\ensuremath{}\ensuremath{}\ensuremath{}}}\mathbb{E}_{t}\left[\frac{\left(Y_{T}-Y_{t}\right)^{4}}{4}\right]\nonumber \\
s.t. & \mbox{ \ensuremath{}\ensuremath{}\ensuremath{}\ensuremath{}\mbox{ \ensuremath{}\ensuremath{}\ensuremath{}\ensuremath{}}}\begin{cases}
dY_{s}=\left(\widehat{r_{s}}Y_{s}+u_{s}^{T}\theta_{s}\right)ds+u_{s}^{T}dW_{s}\\
Y_{t}=y_{t}\\
u\in\overline{U_{ad}^{u}}=\left\{ u\left|\right.u\in L_{\mathcal{F}}^{2}\left(0,T;\mathbb{R}^{d}\right)\mbox{and }\mathbb{E}_{t}\left[Y_{T}\right]=Y_{t},\forall t\in\left[0,T\right)\right\} 
\end{cases}\label{eq:3.2-1}
\end{align}
where $\widehat{r_{s}}=r_{s}-\mu_{s}$, $u_{s}=e^{\int_{s}^{T}\mu_{u}du}\sigma_{s}^{T}\pi_{s}$,
$y_{t}=x_{t}e^{\int_{t}^{T}\mu_{u}du}$. 

$\mathbb{E}_{t}\left[Y_{T}\right]=Y_{t},\forall t\in\left[0,T\right)$
again implies that $Y$must be a martingale which is an admissible
constraint on $u$ and thus as usual we could firstly consider the
family of following problems

\noindent 
\begin{align}
\underset{u}{\min} & \mbox{ \ensuremath{}\ensuremath{}\ensuremath{}\ensuremath{}\mbox{ \ensuremath{}\ensuremath{}\ensuremath{}\ensuremath{}}\mbox{ \ensuremath{}\ensuremath{}\ensuremath{}\ensuremath{}}}\mathbb{E}_{t}\left[\frac{\left(Y_{T}-Y_{t}\right)^{4}}{4}\right]\nonumber \\
s.t. & \mbox{ \ensuremath{}\ensuremath{}\ensuremath{}\ensuremath{}\mbox{ \ensuremath{}\ensuremath{}\ensuremath{}\ensuremath{}}}\begin{cases}
dY_{s}=\left(\widehat{r_{s}}Y_{s}+u_{s}^{T}\theta_{s}\right)ds+u_{s}^{T}dW_{s}\\
Y_{t}=y_{t}\\
u\in U_{ad}^{u}=\left\{ u\left|\right.u\in L_{\mathcal{F}}^{2}\left(0,T;\mathbb{R}^{d}\right)\right\} 
\end{cases}\label{eq:3.3-1}
\end{align}

\begin{prop}
\label{thm:3.1-1}If an equilibrium solution $\left(u^{*},Y^{*}\right)$
to the above family of problems (\ref{eq:3.3-1}) satisfies\textup{
$\widehat{r_{s}}Y_{s}^{*}+\left(u_{s}^{*}\right)^{T}\theta_{s}=0$
}for any $s\in\left[0,T\right)$, then $\left(u^{*},Y^{*}\right)$
is also equilibrium for the family of problems in (\ref{eq:3.2-1})\end{prop}
\begin{proof}
Suppose $\left(u^{*},Y^{*}\right)$ is an equilibrium solution to
(\ref{eq:3.3-1}) which satisfies $\widehat{r_{s}}Y_{s}^{*}+\left(u_{s}^{*}\right)^{T}\theta_{s}=0$
for any $s\in\left[0,T\right)$, then we have that $dY_{s}^{*}=u_{s}^{*T}dW_{s},\forall s\in\left[0,T\right)$$ $
which implies that $Y^{*}$ is a martingale. So we have $\mathbb{E}_{t}\left[Y_{T}^{*}\right]=Y_{t}^{*},\forall t\in\left[0,T\right)$,
which as before is used together with definition \ref{2.1} to deduce
that $\left(u^{*},Y^{*}\right)$ is also equilibrium for the family
of problems (\ref{eq:3.2-1})
\end{proof}
\noindent By the above proposition, we try to solve the family of
problems (\ref{eq:3.3-1}) instead, that is
\begin{align}
\underset{u}{\min} & \mbox{ \ensuremath{}\ensuremath{}\ensuremath{}\ensuremath{}\mbox{ \ensuremath{}\ensuremath{}\ensuremath{}\ensuremath{}}\mbox{ \ensuremath{}\ensuremath{}\ensuremath{}\ensuremath{}}}\mathbb{E}_{t}\left[h\left(Y_{T}-Y_{t}\right)\right]\nonumber \\
s.t. & \mbox{ \ensuremath{}\ensuremath{}\ensuremath{}\ensuremath{}\mbox{ \ensuremath{}\ensuremath{}\ensuremath{}\ensuremath{}}}\begin{cases}
dY_{s}=\left(\widehat{r_{s}}Y_{s}+u_{s}^{T}\theta_{s}\right)ds+u_{s}^{T}dW_{s}\\
Y_{t}=y_{t}\\
u\in U_{ad}^{u}=\left\{ u\left|\right.u\in L_{\mathcal{F}}^{2}\left(0,T;\mathbb{R}^{d}\right)\right\} 
\end{cases}\label{eq:3.4-1}
\end{align}

\noindent where $h\left(x\right)=\frac{x^{4}}{4}$, and this is once
again a family of problems of the form in (\ref{eq:1.3}) , so we
have the following system of BSDEs by (\ref{eq:1.4}) and (\ref{eq:1.5})
\begin{equation}
\begin{cases}
dp_{s}^{t} & =-\widehat{r_{s}}p_{s}^{t}ds+\left(q_{s}^{t}\right)^{T}dW_{s},\mbox{ \ensuremath{}s\ensuremath{\in\left[t,T\right]}}\\
p_{T}^{t} & =\frac{dh\left(Y_{T}^{*}-Y_{t}^{*}\right)}{dx}=\left(Y_{T}^{*}-Y_{t}^{*}\right)^{3}
\end{cases}\label{eq:3.5-1}
\end{equation}
\begin{equation}
\begin{cases}
dP_{s}^{t} & =-2\widehat{r_{s}}P_{s}^{t}ds+\left(Q_{s}^{t}\right)^{T}dW_{s},\mbox{ \ensuremath{}s\ensuremath{\in\left[t,T\right]}}\\
P_{T}^{t} & =\frac{d^{2}h\left(Y_{T}^{*}-Y_{t}^{*}\right)}{dx^{2}}=3\left(Y_{T}^{*}-Y_{t}^{*}\right)^{2}
\end{cases}\label{eq:3.5-2-1}
\end{equation}

\begin{prop}
\label{prop:3.2-1}If (\textup{\ref{eq:3.4-1}}) and (\ref{eq:3.5-1})
admit a solution $\left(u^{*},Y^{*},p^{t},q^{t}\right)$ for any $t\in\left[0,T\right)$
s.t.
\begin{equation}
\begin{cases}
u^{*}\in L_{\mathcal{F}}^{2}\left(0,T;\mathbb{R}^{d}\right)\\
\mathbb{E}_{t}\int_{t}^{T}\left|\Lambda_{s}^{t}\right|ds<\infty\mbox{ and \ensuremath{\underset{s\left\downarrow t\right.}{\lim}\mathbb{E}_{t}\left[\Lambda_{s}^{t}\right]=0}} & where\mbox{ }\Lambda_{s}^{t}=p_{s}^{t}\theta_{s}+q_{s}^{t}.\\
\widehat{r_{t}}Y_{t}^{*}+\left(u_{t}^{*}\right)^{T}\theta_{t}=0
\end{cases}\label{eq:3.7-1}
\end{equation}
then $u^{*}$ is an equilibrium control for the family of problems
(\ref{eq:3.2-1}) \end{prop}
\begin{proof}
Since one of the sufficient condition for equilibrium $\mathbb{E}_{t}$$\left[\frac{d^{2}h\left(Y_{T}^{*}-Y_{t}^{*}\right)}{dx^{2}}\right]=3\mathbb{E}_{t}\left[\left(Y_{T}^{*}-Y_{t}^{*}\right)^{2}\right]\geq0$
in (\ref{eq:1.11}) under theorem \ref{thm:1.3} has already been
satisfied. Then by combining theorem \ref{thm:1.3} and proposition
\ref{thm:3.1-1}, we deduce that $u^{*}$ is an equilibrium for the
family of problems (\ref{eq:3.2-1}).
\end{proof}

\subsection{Details of finding a potential equilibrium}

\noindent By the assumptions we have made, $\widehat{r}$ and $\sigma$
are deterministic and bounded, $\mu^{x}$ is bounded. Here we also
assume that $\mu^{x}$ is deterministic, i.e. $\theta$ is deterministic
and bounded. Then for any $t\in\left[0,T\right)$, we make the following
Ansatz
\begin{equation}
p_{s}^{t}=M_{s}\left(Y_{s}^{*}\right)^{3}-N_{s}\left(Y_{s}^{*}\right)^{2}Y_{t}^{*}+\Gamma_{s}Y_{s}^{*}\left(Y_{t}^{*}\right)^{2}-\Phi_{s}\left(Y_{t}^{*}\right)^{3},s\in\left[t,T\right]\label{eq:3.8-1}
\end{equation}
where $M,N,\Gamma,\Phi$ are deterministic functions which are differentiable
with $M_{T}=1,N_{T}=3,\Gamma_{T}=3,\Phi_{T}=1$ 

\noindent by applying It$\hat{\mbox{o}}$ formula to (\ref{eq:3.8-1})
with respect to $s$ we could get

\noindent 
\begin{eqnarray}
dp_{s}^{t} & = & d\left[M_{s}\left(Y_{s}^{*}\right)^{3}\right]-d\left[N_{s}\left(Y_{s}^{*}\right)^{2}Y_{t}^{*}\right]+d\left[\Gamma_{s}Y_{s}^{*}\left(Y_{t}^{*}\right)^{2}\right]-d\left[\Phi_{s}\left(Y_{t}^{*}\right)^{3}\right]\nonumber \\
 & = & \left[3M_{s}\left(Y_{s}^{*}\right)^{2}-2N_{s}Y_{t}^{*}Y_{s}^{*}+\Gamma_{s}\left(Y_{t}^{*}\right)^{2}\right]\left(\widehat{r_{s}}Y_{s}^{*}+\left(u_{s}^{*}\right)^{T}\theta_{s}\right)ds\nonumber \\
 &  & +\left[M_{s}^{\mathbf{\prime}}\left(Y_{s}^{*}\right)^{3}-N_{s}^{\prime}Y_{t}^{*}\left(Y_{s}^{*}\right)^{2}+Y_{s}^{*}\left(Y_{t}^{*}\right)^{2}\Gamma_{s}^{\prime}-\Phi_{s}^{\prime}\left(Y_{t}^{*}\right)^{3}\right]ds\nonumber \\
 &  & +\left[3M_{s}Y_{s}^{*}\left(u_{s}^{*}\right)^{T}u_{s}^{*}-N_{s}Y_{t}^{*}\left(u_{s}^{*}\right)^{T}u_{s}^{*}\right]ds\label{eq:3.9-1}\\
 &  & +\left[3M_{s}\left(Y_{s}^{*}\right)^{2}-2N_{s}Y_{t}^{*}Y_{s}^{*}+\left(Y_{t}^{*}\right)^{2}\Gamma_{s}\right]\left(u_{s}^{*}\right)^{T}dW_{s}
\end{eqnarray}
by comparing the $dW$ terms of $dp^{t}$ in (\ref{eq:3.5-1}) and
(\ref{eq:3.9-1}), we get that
\begin{equation}
q_{s}^{t}=\left[3M_{s}\left(Y_{s}^{*}\right)^{2}-2N_{s}Y_{t}^{*}Y_{s}^{*}+\left(Y_{t}^{*}\right)^{2}\Gamma_{s}\right]u_{s}^{*},\mbox{ \ensuremath{}s\ensuremath{\in\left[t,T\right]}}
\end{equation}
we again hope to find a possible linear feedback $u^{*}$ and as before
try by setting 
\begin{equation}
0=\Lambda_{s}^{s}=p_{s}^{s}\theta_{s}+q_{s}^{s},\mbox{ \ensuremath{}s\ensuremath{\in\left[0,T\right]}}
\end{equation}
which leads to the equation
\begin{equation}
\left[\left(M_{s}-N_{s}+\Gamma_{s}-\Phi_{s}\right)Y_{s}^{*}\theta_{s}+\left(3M_{s}-2N_{s}+\Gamma_{s}\right)u_{s}^{*}\right]\left(Y_{s}^{*}\right)^{2}=0
\end{equation}
from which we get
\begin{equation}
u_{s}^{*}=\alpha_{s}\theta_{s}Y_{s}^{*}\label{eq:5.13-1}
\end{equation}
where 
\begin{equation}
\alpha_{s}=\begin{cases}
\frac{M_{s}-N_{s}+\Gamma_{s}-\Phi_{s}}{-3M_{s}+2N_{s}-\Gamma_{s}}, & \forall s\in\left[0,T\right)\\
\underset{s\left\uparrow T\right.}{\lim}\frac{M_{s}-N_{s}+\Gamma_{s}-\Phi_{s}}{-3M_{s}+2N_{s}-\Gamma_{s}}, & s=T
\end{cases}\label{eq:3.13-1}
\end{equation}
based on the assumption that $-3M_{s}+2N_{s}-\Gamma_{s}\neq0,\forall s\in\left[0,T\right)$
and $\underset{s\left\uparrow T\right.}{\lim}\frac{M_{s}-N_{s}+\Gamma_{s}-\Phi_{s}}{-3M_{s}+2N_{s}-\Gamma_{s}}$
exists.\\

\noindent By comparing the $ds$ terms of $dp^{t}$ in (\ref{eq:3.5-1})
and (\ref{eq:3.9-1}), we get that
\begin{eqnarray*}
 &  & -\widehat{r_{s}}\left[M_{s}\left(Y_{s}^{*}\right)^{3}-N_{s}\left(Y_{s}^{*}\right)^{2}Y_{t}^{*}+\Gamma_{s}Y_{s}^{*}\left(Y_{t}^{*}\right)^{2}-\Phi_{s}\left(Y_{t}^{*}\right)^{3}\right]\\
 & = & \left[3M_{s}\left(Y_{s}^{*}\right)^{2}-2N_{s}Y_{t}^{*}Y_{s}^{*}+\Gamma_{s}\left(Y_{t}^{*}\right)^{2}\right]\left(\widehat{r_{s}}Y_{s}^{*}+\left(u_{s}^{*}\right)^{T}\theta_{s}\right)\\
 &  & +\left[M_{s}^{\mathbf{\prime}}\left(Y_{s}^{*}\right)^{3}-N_{s}^{\prime}Y_{t}^{*}\left(Y_{s}^{*}\right)^{2}+Y_{s}^{*}\left(Y_{t}^{*}\right)^{2}\Gamma_{s}^{\prime}-\Phi_{s}^{\prime}\left(Y_{t}^{*}\right)^{3}\right]\\
 &  & +\left[3M_{s}Y_{s}^{*}\left(u_{s}^{*}\right)^{T}u_{s}^{*}-N_{s}Y_{t}^{*}\left(u_{s}^{*}\right)^{T}u_{s}^{*}\right]\\
\\
\Leftrightarrow &  & -\widehat{r_{s}}\left[M_{s}\left(Y_{s}^{*}\right)^{3}-N_{s}\left(Y_{s}^{*}\right)^{2}Y_{t}^{*}+\Gamma_{s}Y_{s}^{*}\left(Y_{t}^{*}\right)^{2}-\Phi_{s}\left(Y_{t}^{*}\right)^{3}\right]\\
 & = & \left[3M_{s}\left(Y_{s}^{*}\right)^{2}-2N_{s}Y_{t}^{*}Y_{s}^{*}+\Gamma_{s}\left(Y_{t}^{*}\right)^{2}\right]\left(\widehat{r_{s}}Y_{s}^{*}+\alpha_{s}\left|\theta_{s}\right|^{2}Y_{s}^{*}\right)\\
 &  & +\left[M_{s}^{\mathbf{\prime}}\left(Y_{s}^{*}\right)^{3}-N_{s}^{\prime}Y_{t}^{*}\left(Y_{s}^{*}\right)^{2}+Y_{s}^{*}\left(Y_{t}^{*}\right)^{2}\Gamma_{s}^{\prime}-\Phi_{s}^{\prime}\left(Y_{t}^{*}\right)^{3}\right]\\
 &  & +\left[3M_{s}Y_{s}^{*}-N_{s}Y_{t}^{*}\right]\alpha_{s}^{2}\left|\theta_{s}\right|^{2}\left(Y_{s}^{*}\right)^{2}\\
\\
\Leftrightarrow &  & -\widehat{r_{s}}\left[M_{s}\left(Y_{s}^{*}\right)^{3}-N_{s}\left(Y_{s}^{*}\right)^{2}Y_{t}^{*}+\Gamma_{s}Y_{s}^{*}\left(Y_{t}^{*}\right)^{2}-\Phi_{s}\left(Y_{t}^{*}\right)^{3}\right]\\
 & = & \left[3M_{s}\widehat{r_{s}}\left(Y_{s}^{*}\right)^{3}-2N_{s}\widehat{r_{s}}Y_{t}^{*}\left(Y_{s}^{*}\right)^{2}+\Gamma_{s}\widehat{r_{s}}\left(Y_{t}^{*}\right)^{2}Y_{s}^{*}\right]\\
 &  & +\left[3\alpha_{s}\left|\theta_{s}\right|^{2}M_{s}\left(Y_{s}^{*}\right)^{3}-2\alpha_{s}\left|\theta_{s}\right|^{2}N_{s}Y_{t}^{*}\left(Y_{s}^{*}\right)^{2}+\alpha_{s}\left|\theta_{s}\right|^{2}\Gamma_{s}\left(Y_{t}^{*}\right)^{2}Y_{s}^{*}\right]\\
 &  & +\left[M_{s}^{\mathbf{\prime}}\left(Y_{s}^{*}\right)^{3}-N_{s}^{\prime}Y_{t}^{*}\left(Y_{s}^{*}\right)^{2}+Y_{s}^{*}\left(Y_{t}^{*}\right)^{2}\Gamma_{s}^{\prime}-\Phi_{s}^{\prime}\left(Y_{t}^{*}\right)^{3}\right]\\
 &  & +\left[3M_{s}Y_{s}^{*}-N_{s}Y_{t}^{*}\right]\alpha_{s}^{2}\left|\theta_{s}\right|^{2}\left(Y_{s}^{*}\right)^{2}\\
\\
\Leftrightarrow &  & -\widehat{r_{s}}M_{s}\left(Y_{s}^{*}\right)^{3}+\widehat{r_{s}}N_{s}\left(Y_{s}^{*}\right)^{2}Y_{t}^{*}-\widehat{r_{s}}\Gamma_{s}Y_{s}^{*}\left(Y_{t}^{*}\right)^{2}+\widehat{r_{s}}\Phi_{s}\left(Y_{t}^{*}\right)^{3}\\
 & = & \left[M_{s}^{\mathbf{\prime}}+3M_{s}\widehat{r_{s}}+3\alpha_{s}\left|\theta_{s}\right|^{2}M_{s}+3\alpha_{s}^{2}\left|\theta_{s}\right|^{2}M_{s}\right]\left(Y_{s}^{*}\right)^{3}\\
 &  & -\left[N_{s}^{\prime}+2N_{s}\widehat{r_{s}}+2\alpha_{s}\left|\theta_{s}\right|^{2}N_{s}+\alpha_{s}^{2}\left|\theta_{s}\right|^{2}N_{s}\right]Y_{t}^{*}\left(Y_{s}^{*}\right)^{2}\\
 &  & +\left[\Gamma_{s}^{\prime}+\Gamma_{s}\widehat{r_{s}}+\alpha_{s}\left|\theta_{s}\right|^{2}\Gamma_{s}\right]Y_{s}^{*}\left(Y_{t}^{*}\right)^{2}-\Phi_{s}^{\prime}\left(Y_{t}^{*}\right)^{3}
\end{eqnarray*}
after rearrangement we get
\begin{eqnarray}
0 & = & \left(\Phi_{s}^{\prime}+\widehat{r_{s}}\Phi_{s}\right)\left(Y_{t}^{*}\right)^{3}-\left[\Gamma_{s}^{\prime}+2\widehat{r_{s}}\Gamma_{s}+\alpha_{s}\left|\theta_{s}\right|^{2}\Gamma_{s}\right]Y_{s}^{*}\left(Y_{t}^{*}\right)^{2}\nonumber \\
 &  & +\left[N_{s}^{\prime}+3\widehat{r_{s}}N_{s}+2\alpha_{s}\left|\theta_{s}\right|^{2}N_{s}+\alpha_{s}^{2}\left|\theta_{s}\right|^{2}N_{s}\right]\left(Y_{s}^{*}\right)^{2}Y_{t}^{*}\nonumber \\
 &  & -\left[M_{s}^{\mathbf{\prime}}+4\widehat{r_{s}}M_{s}+3\alpha_{s}\left|\theta_{s}\right|^{2}M_{s}+3\alpha_{s}^{2}\left|\theta_{s}\right|^{2}M_{s}\right]\left(Y_{s}^{*}\right)^{3}
\end{eqnarray}
which leads to the following system of ODEs 
\begin{align}
 & \begin{cases}
M_{s}^{\mathbf{\prime}}+\left(4\widehat{r_{s}}+3\alpha_{s}\left|\theta_{s}\right|^{2}+3\alpha_{s}^{2}\left|\theta_{s}\right|^{2}\right)M_{s}=0, & s\in\left[0,T\right]\\
M_{T}=1
\end{cases}\label{eq:3.15-1}\\
 & \begin{cases}
N_{s}^{\prime}+\left(3\widehat{r_{s}}+2\alpha_{s}\left|\theta_{s}\right|^{2}+\alpha_{s}^{2}\left|\theta_{s}\right|^{2}\right)N_{s}=0, & s\in\left[0,T\right]\\
N_{T}=3
\end{cases}\label{eq:3.16-1}\\
 & \begin{cases}
\Gamma_{s}^{\prime}+\left(2\widehat{r_{s}}+\alpha_{s}\left|\theta_{s}\right|^{2}\right)\Gamma_{s}=0, & s\in\left[0,T\right]\\
\Gamma_{T}=3
\end{cases}\label{eq:3.16-2}\\
 & \begin{cases}
\Phi_{s}^{\prime}+\widehat{r_{s}}\Phi_{s}=0, & s\in\left[0,T\right]\\
\Phi_{T}=1
\end{cases}\label{eq:3.17-1}
\end{align}
the solution to equation (\ref{eq:3.17-1}) is $\Phi_{s}=e^{\int_{s}^{T}\widehat{r_{u}}du}$,
which makes the unsettled system contains (\ref{eq:3.15-1}), (\ref{eq:3.16-1})
and (\ref{eq:3.16-2}) as follows 
\begin{align}
 & \begin{cases}
M_{s}^{\mathbf{\prime}}+\left(4\widehat{r_{s}}+3\alpha_{s}\left|\theta_{s}\right|^{2}+3\alpha_{s}^{2}\left|\theta_{s}\right|^{2}\right)M_{s}=0, & s\in\left[0,T\right]\\
M_{T}=1\\
N_{s}^{\prime}+\left(3\widehat{r_{s}}+2\alpha_{s}\left|\theta_{s}\right|^{2}+\alpha_{s}^{2}\left|\theta_{s}\right|^{2}\right)N_{s}=0, & s\in\left[0,T\right]\\
N_{T}=3\\
\Gamma_{s}^{\prime}+\left(2\widehat{r_{s}}+\alpha_{s}\left|\theta_{s}\right|^{2}\right)\Gamma_{s}=0, & s\in\left[0,T\right]\\
\Gamma_{T}=3
\end{cases}\label{eq:3.18-1}
\end{align}

\subsection{Conditions for obtaining an equilibrium for our problem}
\begin{thm}
\label{thm:5.4-1}If the system of ODEs (\ref{eq:3.18-1}) admits
a solution $\left(M,N\right)$ s.t. the corresponding \textup{$\alpha$}
defined in (\ref{eq:3.13-1}) satisfies that $\widehat{r_{t}}+\alpha_{t}\left|\theta_{t}\right|^{2}=0,\forall t\in\left[0,T\right)$,
then $u^{*}$ is an equilibrium control for the family of problems
in (\ref{eq:3.2-1})\end{thm}
\begin{proof}
Suppose $\left(M,N\right)$ is a solution to (\ref{eq:3.18-1}) s.t.
$\widehat{r_{t}}+\alpha_{t}\left|\theta_{t}\right|^{2}=0,\forall t\in\left[0,T\right)$.
Since the deterministic $M,N,\Gamma$ are continuous and thus are
bounded on $\left[0,T\right]$, we have that $\alpha$ is also deterministic
and bounded on $\left[0,T\right]$ according to (\ref{eq:3.13-1}).
We have in (\ref{eq:5.13-1}) for $s\in\left[0,T\right]$ that
\begin{equation}
u_{s}^{*}=\alpha_{s}\theta_{s}Y_{s}^{*}
\end{equation}
and thus we have 
\begin{align}
dY_{s}^{*} & =\left(\widehat{r_{s}}Y_{s}^{*}+\left(u_{s}^{*}\right)^{T}\theta_{s}\right)ds+\left(u_{s}^{*}\right)^{T}dW_{s}\nonumber \\
 & =Y_{s}^{*}\left[\left(\widehat{r_{s}}+\alpha_{s}\left|\theta_{s}\right|^{2}\right)ds+\alpha_{s}\left(\theta_{s}\right)^{T}dW_{s}\right]
\end{align}
which leads to
\begin{equation}
Y_{t}^{*}=y_{0}e^{\int_{0}^{t}\widehat{r_{s}}+\alpha_{s}\left|\theta_{s}\right|^{2}-\frac{1}{2}\alpha_{s}^{2}\left|\theta_{s}\right|^{2}ds+\int_{0}^{t}\alpha_{s}\left(\theta_{s}\right)^{T}dW_{s}}
\end{equation}
thus
\begin{eqnarray}
\mathbb{E}\left[\underset{t\in\left[0,T\right]}{\sup}\left|Y_{t}^{*}\right|^{3}\right] & = & \mathbb{E}\left[\underset{t\in\left[0,T\right]}{\sup}\left|y_{0}e^{\int_{0}^{t}\widehat{r_{s}}+\alpha_{s}\left|\theta_{s}\right|^{2}ds}\right|^{3}\left|e^{\int_{0}^{t}-\frac{1}{2}\alpha_{s}^{2}\left|\theta_{s}\right|^{2}ds+\int_{0}^{t}\alpha_{s}\left(\theta_{s}\right)^{T}dW_{s}}\right|^{3}\right]\nonumber \\
 & \leq & \left[\underset{t\in\left[0,T\right]}{\sup}\left|y_{0}e^{\int_{0}^{t}\widehat{r_{s}}+\alpha_{s}\left|\theta_{s}\right|^{2}ds}\right|^{3}\right]\mathbb{E}\left[\underset{t\in\left[0,T\right]}{\sup}\left(e^{\int_{0}^{t}-\frac{1}{2}\alpha_{s}^{2}\left|\theta_{s}\right|^{2}ds+\int_{0}^{t}\alpha_{s}\left(\theta_{s}\right)^{T}dW_{s}}\right)^{3}\right]\nonumber \\
 & \leq & \left[\underset{t\in\left[0,T\right]}{\sup}\left|y_{0}e^{\int_{0}^{t}\widehat{r_{s}}+\alpha_{s}\left|\theta_{s}\right|^{2}ds}\right|^{3}\right]\frac{27}{8}\mathbb{E}\left[\left(e^{\int_{0}^{T}-\frac{1}{2}\alpha_{s}^{2}\left|\theta_{s}\right|^{2}ds+\int_{0}^{T}\alpha_{s}\left(\theta_{s}\right)^{T}dW_{s}}\right)^{3}\right]\nonumber \\
 &  & \mbox{by \ensuremath{L^{p}}Maximal Inequality}\nonumber \\
 & = & \frac{27}{8}e^{\int_{0}^{T}3\alpha_{s}^{2}\left|\theta_{s}\right|^{2}ds}\left[\underset{t\in\left[0,T\right]}{\sup}\left|y_{0}e^{\int_{0}^{t}\widehat{r_{s}}+\alpha_{s}\left|\theta_{s}\right|^{2}ds}\right|^{3}\right]\mathbb{E}\left[e^{\int_{0}^{T}-\frac{9}{2}\alpha_{s}^{2}\left|\theta_{s}\right|^{2}ds+\int_{0}^{T}3\alpha_{s}\left(\theta_{s}\right)^{T}dW_{s}}\right]\nonumber \\
 & = & \frac{27}{8}e^{\int_{0}^{T}3\alpha_{s}^{2}\left|\theta_{s}\right|^{2}ds}\underset{t\in\left[0,T\right]}{\sup}\left|y_{0}e^{\int_{0}^{t}\widehat{r_{s}}+\alpha_{s}\left|\theta_{s}\right|^{2}ds}\right|^{3}\nonumber \\
 & < & \infty
\end{eqnarray}
as $\alpha,\theta,\widehat{r}$ are bounded, which implies that 
\begin{equation}
Y^{*}\in L_{\mathcal{F}}^{3}\left(\Omega;C\left(0,T;\mathbb{R}\right)\right)\label{eq:3.23-1}
\end{equation}
and thus we have 
\begin{eqnarray}
\mathbb{E}\left[\int_{0}^{T}\left|u_{s}^{*}\right|^{2}ds\right] & = & \mathbb{E}\left[\int_{0}^{T}\alpha_{s}^{2}\left|\theta_{s}\right|^{2}\left(Y_{s}^{*}\right)^{2}ds\right]\nonumber \\
 & \leq & \underset{s\in\left[0,T\right]}{\sup}\left(\alpha_{s}^{2}\left|\theta_{s}\right|^{2}\right)\int_{0}^{T}\mathbb{E}\left[\left(Y_{s}^{*}\right)^{2}\right]ds\\
 & \leq & \underset{s\in\left[0,T\right]}{\sup}\left(\alpha_{s}^{2}\left|\theta_{s}\right|^{2}\right)\int_{0}^{T}\left(\mathbb{E}\left[\left|Y_{s}^{*}\right|^{3}\right]\right)^{\frac{2}{3}}ds,\mbox{by Holder's }\\
 & \leq & \underset{s\in\left[0,T\right]}{\sup}\left(\alpha_{s}^{2}\left|\theta_{s}\right|^{2}\right)\int_{0}^{T}\left(\mathbb{E}\left[\underset{u\in\left[0,T\right]}{\sup}\left|Y_{u}^{*}\right|^{3}\right]\right)^{\frac{2}{3}}ds\nonumber \\
 & = & \underset{s\in\left[0,T\right]}{\sup}\left(\alpha_{s}^{2}\left|\theta_{s}\right|^{2}\right)\left(\mathbb{E}\left[\underset{u\in\left[0,T\right]}{\sup}\left|Y_{u}^{*}\right|^{3}\right]\right)^{\frac{2}{3}}T\\
 & < & \infty
\end{eqnarray}
which implies that
\begin{equation}
u^{*}\in L_{\mathcal{F}}^{2}\left(0,T;\mathbb{R}^{d}\right)\label{eq:3.25-1}
\end{equation}
also we have that for any $t\in\left[0,T\right)$ 
\begin{eqnarray*}
\mbox{ }\Lambda_{s}^{t} & = & p_{s}^{t}\theta_{s}+q_{s}^{t}\\
 & = & \left[M_{s}\left(Y_{s}^{*}\right)^{3}-N_{s}\left(Y_{s}^{*}\right)^{2}Y_{t}^{*}+\Gamma_{s}Y_{s}^{*}\left(Y_{t}^{*}\right)^{2}-\Phi_{s}\left(Y_{t}^{*}\right)^{3}\right]\theta_{s}\\
 &  & +\left[3M_{s}\left(Y_{s}^{*}\right)^{2}-2N_{s}Y_{t}^{*}Y_{s}^{*}+\left(Y_{t}^{*}\right)^{2}\Gamma_{s}\right]u_{s}^{*}\\
 & = & \left[M_{s}\left(Y_{s}^{*}\right)^{3}-N_{s}\left(Y_{s}^{*}\right)^{2}Y_{t}^{*}+\Gamma_{s}Y_{s}^{*}\left(Y_{t}^{*}\right)^{2}-\Phi_{s}\left(Y_{t}^{*}\right)^{3}\right]\theta_{s}\\
 &  & +\left[3\alpha_{s}M_{s}\left(Y_{s}^{*}\right)^{3}-2\alpha_{s}N_{s}Y_{t}^{*}\left(Y_{s}^{*}\right)^{2}+\alpha_{s}\left(Y_{t}^{*}\right)^{2}\Gamma_{s}Y_{s}^{*}\right]\theta_{s}\\
 & = & \left[\left(1+3\alpha_{s}\right)M_{s}\left(Y_{s}^{*}\right)^{3}-\left(1+2\alpha_{s}\right)N_{s}Y_{t}^{*}\left(Y_{s}^{*}\right)^{2}+\left(1+\alpha_{s}\right)\Gamma_{s}\left(Y_{t}^{*}\right)^{2}Y_{s}^{*}-\Phi_{s}\left(Y_{t}^{*}\right)^{3}\right]\theta_{s}
\end{eqnarray*}
Since $\alpha,M,N,\Gamma,\theta$ are bounded, it is clearly by (\ref{eq:3.23-1})
that 
\begin{equation}
\mathbb{E}_{t}\int_{t}^{T}\left|\Lambda_{s}^{t}\right|ds<\infty\label{eq:3.27-1}
\end{equation}
and
\begin{equation}
\underset{s\left\downarrow t\right.}{\lim}\mathbb{E}_{t}\left[\Lambda_{s}^{t}\right]=\underset{s\left\downarrow t\right.}{\lim}\theta_{s}\mathbb{E}_{t}\left[\left(1+3\alpha_{s}\right)M_{s}\left(Y_{s}^{*}\right)^{3}-\left(1+2\alpha_{s}\right)N_{s}Y_{t}^{*}\left(Y_{s}^{*}\right)^{2}+\left(1+\alpha_{s}\right)\Gamma_{s}\left(Y_{t}^{*}\right)^{2}Y_{s}^{*}-\Phi_{s}\left(Y_{t}^{*}\right)^{3}\right]
\end{equation}
since we have
\begin{align}
 & \underset{s\left\downarrow t\right.}{\lim}\mathbb{E}_{t}\left[\left(1+3\alpha_{s}\right)M_{s}\left(Y_{s}^{*}\right)^{3}-\left(1+2\alpha_{s}\right)N_{s}Y_{t}^{*}\left(Y_{s}^{*}\right)^{2}+\left(1+\alpha_{s}\right)\Gamma_{s}\left(Y_{t}^{*}\right)^{2}Y_{s}^{*}-\Phi_{s}\left(Y_{t}^{*}\right)^{3}\right]\nonumber \\
= & \underset{s\left\downarrow t\right.}{\lim}\left\{ \left(1+3\alpha_{s}\right)M_{s}\mathbb{E}_{t}\left[\left(Y_{s}^{*}\right)^{3}\right]-\left(1+2\alpha_{s}\right)N_{s}Y_{t}^{*}\mathbb{E}_{t}\left[\left(Y_{s}^{*}\right)^{2}\right]+\left(1+\alpha_{s}\right)\Gamma_{s}\left(Y_{t}^{*}\right)^{2}\mathbb{E}_{t}\left[Y_{s}^{*}\right]-\Phi_{s}\left(Y_{t}^{*}\right)^{3}\right\} \nonumber \\
\overset{(\ref{eq:3.23-1})}{=} & \left(1+3\alpha_{t}\right)M_{t}\mathbb{E}_{t}\left[\underset{s\left\downarrow t\right.}{\lim}\left(Y_{s}^{*}\right)^{3}\right]-\left(1+2\alpha_{t}\right)N_{t}Y_{t}^{*}\mathbb{E}_{t}\left[\underset{s\left\downarrow t\right.}{\lim}\left(Y_{s}^{*}\right)^{2}\right]+\left(1+\alpha_{t}\right)\Gamma_{t}\left(Y_{t}^{*}\right)^{2}\mathbb{E}_{t}\left[\underset{s\left\downarrow t\right.}{\lim}Y_{s}^{*}\right]-\Phi_{t}\left(Y_{t}^{*}\right)^{3}\nonumber \\
 & \mbox{by Dominated Convergence}\nonumber \\
= & \left[\left(1+3\alpha_{t}\right)M_{t}-\left(1+2\alpha_{t}\right)N_{t}+\left(1+\alpha_{t}\right)\Gamma_{t}-\Phi_{t}\right]\left(Y_{t}^{*}\right)^{3}\nonumber \\
= & \left[\frac{M_{t}-N_{t}+\Gamma_{t}-\Phi_{t}}{-3M_{t}+2N_{t}-\Gamma_{t}}-\alpha_{t}\right]\left(-3M_{t}+2N_{t}-\Gamma_{t}\right)\left(Y_{t}^{*}\right)^{3}\nonumber \\
= & 0
\end{align}
and also $\theta$ is bounded, thus we deduce that 
\begin{equation}
\underset{s\left\downarrow t\right.}{\lim}\mathbb{E}_{t}\left[\Lambda_{s}^{t}\right]=0\label{eq:3.30-1}
\end{equation}
we also have $\widehat{r_{t}}+\alpha_{t}\left|\theta_{t}\right|^{2}=0,\forall t\in\left[0,T\right)$
which implies that
\begin{eqnarray}
 & \left(\widehat{r_{t}}+\alpha_{t}\left|\theta_{t}\right|^{2}\right)Y_{t}^{*} & =0,\forall t\in\left[0,T\right)\nonumber \\
\Rightarrow & \widehat{r_{t}}Y_{t}^{*}+\alpha_{t}\theta_{t}^{T}\theta_{t}Y_{t}^{*} & =0,\forall t\in\left[0,T\right)\nonumber \\
\Rightarrow & \widehat{r_{t}}Y_{t}^{*}+\left(u_{t}^{*}\right)^{T}\theta_{t} & =0,\forall t\in\left[0,T\right)
\end{eqnarray}
which means $\left(u^{*},Y^{*}\right)$ satisfies 
\begin{equation}
\widehat{r_{t}}Y_{t}^{*}+\left(u_{t}^{*}\right)^{T}\theta_{t}=0,\forall t\in\left[0,T\right)\label{eq:3.31-1}
\end{equation}
Then (\ref{eq:3.25-1}), (\ref{eq:3.27-1}), (\ref{eq:3.30-1}) and
(\ref{eq:3.31-1}) are exactly the required conditions in (\ref{eq:3.7-1})
and we deduce that $u^{*}$ is an equilibrium control for the family
of problems in (\ref{eq:3.2-1}) by proposition \ref{prop:3.2-1}.
\\

\noindent From the studies on $h(x)=-\frac{x^{3}}{3}$ and $\frac{x^{4}}{4}$,
we could see that in solving the problem for these two power utility
functions there is a sort of regular pattern in the systems of ODEs
(\ref{eq:3.18}) and (\ref{eq:3.18-1}) obtained above which can be
extended to the same problem for higher order power functions.
\end{proof}
\newpage{}

\section{Utility function $h\left(x\right)=x^{-}$ for stronger risk aversion }

Some investors are risk averse to the extent that they hope to make
the under-performs down to the lowest level. Thus they would like
to use the utility function we use here. In this section, we change
our view from the previous sections to a different one which makes
this section look like that it is not related to the previous ones.

\subsection{Notations and definition of equilibrium for this section}

\noindent Firstly we define $\mathbb{P}$ as our physical measure
and $\mathbb{Q}$ as the risk neutral measure with 
\begin{equation}
\left.\frac{d\mathbb{Q}}{d\mathbb{P}}\right|_{\mathcal{F}_{t}}=\mathcal{E}\left(-\int_{0}^{t}\theta_{s}dW_{s}\right)\label{eq:7.1}
\end{equation}
In this section, when we define sets $L_{\mathcal{F}}^{2}\left(t,T;\mathbb{R}^{l}\right)$
and $L_{\mathcal{G}}^{2}\left(\Omega;\mathbb{R}^{l}\right)$ which
contain the elements that satisfy the corresponding conditions under
both of the probability measure $\mathbb{P}$ and $\mathbb{Q}$. That
is

\begin{align*}
\begin{cases}
L_{\mathcal{F}}^{2}\left(t,T;\mathbb{R}^{l}\right): & \mbox{the set of \ensuremath{\left\{  \mathcal{F}_{s}\right\} } }_{s\in\left[t,T\right]}\mbox{-adapted processes }f=\left\{ f_{s}:t\leq s\leq T\right\} \\
 & \mbox{ with }\mathbb{E}^{\mathbb{P}}\left[\int_{t}^{T}\left|f_{s}\right|^{2}ds\right]<\infty\mbox{ and }\mathbb{E^{\mathbb{Q}}}\left[\int_{t}^{T}\left|f_{s}\right|^{2}ds\right]<\infty\\
L_{\mathcal{G}}^{2}\left(\Omega;\mathbb{R}^{l}\right): & \mbox{the set of random variables \ensuremath{\xi:\left(\Omega,\mathcal{G}\right)\rightarrow\left(\mathbb{R}^{l},\mathcal{B}\left(\mathbb{R}^{l}\right)\right)}}\\
 & \mbox{with }\mathbb{E}^{\mathbb{P}}\left[\left|\xi\right|^{2}\right]<\infty\mbox{ and }\mathbb{E}^{\mathbb{Q}}\left[\left|\xi\right|^{2}\right]<\infty
\end{cases}
\end{align*}
For the reason of simplicity, we just write $\mathbb{E}\left[\cdot\right]$
to represent $\mathbb{E}^{\mathbb{P}}\left[\cdot\right]$ in the remaining
part of the section. \\

\noindent In this section we want to solve the family of following
problems
\begin{align}
\underset{\pi}{\min} & \mbox{ \ensuremath{}\ensuremath{}\ensuremath{}\ensuremath{}\mbox{ \ensuremath{}\ensuremath{}\ensuremath{}\ensuremath{}}\mbox{ \ensuremath{}\ensuremath{}\ensuremath{}\ensuremath{}}}\mathbb{E}_{t}\left[\left(X_{T}-X_{t}e^{\int_{t}^{T}\mu_{s}ds}\right)^{-}\right]\nonumber \\
s.t. & \mbox{ \ensuremath{}\ensuremath{}\ensuremath{}\ensuremath{}\mbox{ \ensuremath{}\ensuremath{}\ensuremath{}\ensuremath{}}}\begin{cases}
dX_{s}=\left(r_{s}X_{s}+\pi_{s}^{T}\sigma_{s}\theta_{s}\right)ds+\pi_{s}^{T}\sigma_{s}dW_{s}\\
X_{t}=x_{t}\\
\pi\in\overline{U_{ad}^{\pi}}=\left\{ \pi\left|\right.\pi\in L_{\mathcal{F}}^{2}\left(0,T;\mathbb{R}^{d}\right)\mbox{and }\mathbb{E}_{t}\left[X_{T}\right]=X_{t}e^{\int_{t}^{T}\mu_{s}ds},\forall t\in\left[0,T\right)\right\} 
\end{cases}\label{eq:2.102-1}
\end{align}
where we assume our drift rate of risky assets $\mu^{x}$ is deterministic
and thus based on our assumptions made above we have that $r$, $\sigma$
and $\theta$ are bounded and deterministic. Here $\mu$ as usual
is our required return process which is assumed to be bounded and
deterministic. \\

\noindent Then by letting $u=\sigma^{T}\pi$ the above family of problems
is equivalent to
\begin{align}
\underset{u}{\min} & \mbox{ \ensuremath{}\ensuremath{}\ensuremath{}\ensuremath{}\mbox{ \ensuremath{}\ensuremath{}\ensuremath{}\ensuremath{}}\mbox{ \ensuremath{}\ensuremath{}\ensuremath{}\ensuremath{}}}\mathbb{E}_{t}\left[\left(X_{t}e^{\int_{t}^{T}\mu_{s}ds}-X_{T}\right)^{+}\right]\nonumber \\
s.t. & \mbox{ \ensuremath{}\ensuremath{}\ensuremath{}\ensuremath{}\mbox{ \ensuremath{}\ensuremath{}\ensuremath{}\ensuremath{}}}\begin{cases}
dX_{s}=\left(r_{s}X_{s}+u_{s}^{T}\theta_{s}\right)ds+u_{s}^{T}dW_{s}\\
X_{t}=x_{t}\\
u\in\overline{U_{ad}^{u}}=\left\{ u\left|\right.u\in L_{\mathcal{F}}^{2}\left(0,T;\mathbb{R}^{d}\right)\mbox{and }\mathbb{E}_{t}\left[X_{T}\right]=X_{t}e^{\int_{t}^{T}\mu_{s}ds},\forall t\in\left[0,T\right)\right\} 
\end{cases}\label{eq:4.1-2}
\end{align}
Also in this section we use the following definition of equilibrium
which is different from the one used by previous sections.

\paragraph{\textmd{Given a control $u^{*}$, for any $t\in\left[0,T\right)$,
$\varepsilon>0$ and $v\in L_{\mathcal{F}}^{2}\left(t,t+\varepsilon;\mathbb{R}^{d}\right)$,
we define }}

\begin{equation}
u_{s}^{t,\varepsilon,v}=u_{s}^{*}+v_{s}\mathbf{1}_{s\in\left[t,t+\varepsilon\right)},\mbox{ }s\in\left[t,T\right]\label{eq:4.3}
\end{equation}

\begin{defn}
\label{4.1}Let $u^{*}\in U_{ad}$ be a given control with $U_{ad}$
being the set of admissible controls. Let $X^{*}$ be the state process
corresponding to $u^{*}$. The control $u^{*}$ is called an equilibrium
if for any $t\in\left[0,T\right)$, $\exists$ $\delta>0$, s.t. for
any $\varepsilon\in\left(0,\delta\right)$ and $v\in L_{\mathcal{F}}^{2}\left(t,t+\varepsilon;\mathbb{R}^{d}\right)$
s.t. $u^{t,\varepsilon,v}\in U_{ad}$, we have that 

\begin{equation}
J\left(t,X_{t}^{*};u^{t,\varepsilon,v}\right)-J\left(t,X_{t}^{*};u^{*}\right)\geq0
\end{equation}
where $u^{t,\varepsilon,v}$ is defined by (\ref{eq:4.3}).
\end{defn}

\subsection{Details of finding an equilibrium}

Firstly we consider the family of following problems 
\begin{align}
\underset{u}{\min} & \mbox{ \ensuremath{}\ensuremath{}\ensuremath{}\ensuremath{}\mbox{ \ensuremath{}\ensuremath{}\ensuremath{}\ensuremath{}}\mbox{ \ensuremath{}\ensuremath{}\ensuremath{}\ensuremath{}}}\mathbb{E}_{t}\left[\left(X_{t}e^{\int_{t}^{T}\mu_{s}ds}-X_{T}\right)^{+}\right]\nonumber \\
s.t. & \mbox{ \ensuremath{}\ensuremath{}\ensuremath{}\ensuremath{}\mbox{ \ensuremath{}\ensuremath{}\ensuremath{}\ensuremath{}}}\begin{cases}
dX_{s}=\left(r_{s}X_{s}+u_{s}^{T}\theta_{s}\right)ds+u_{s}^{T}dW_{s}\\
X_{t}=x_{t}\\
u\in U_{ad}^{u}=\left\{ u\left|\right.u\in L_{\mathcal{F}}^{2}\left(0,T;\mathbb{R}^{d}\right)\right\} 
\end{cases}\label{eq:4.2}
\end{align}
Let $X^{t,\varepsilon,v}$ be the state process corresponding to $u^{t,\varepsilon,v}$.
We set
\begin{equation}
\Delta_{s}^{t,\varepsilon,v}=X_{s}^{t,\varepsilon,v}-X_{s}^{*}
\end{equation}
then by letting
\begin{eqnarray*}
\bar{\Delta}_{s}^{t,\varepsilon,v} & = & \Delta_{s}^{t,\varepsilon,v}e^{-\int_{t}^{s}r_{u}du}\\
\bar{X}_{s}^{t,\varepsilon,v} & = & X_{s}^{t,\varepsilon,v}e^{-\int_{t}^{s}r_{u}du}\\
\bar{X}_{s}^{*} & = & X_{s}^{*}e^{-\int_{t}^{s}r_{u}du}\\
\bar{u}_{s}^{*} & = & u_{s}^{*}e^{-\int_{t}^{s}r_{u}du}\\
\bar{u}_{s}^{t,\varepsilon,v} & = & u_{s}^{t,\varepsilon,v}e^{-\int_{t}^{s}r_{u}du}\\
\bar{v}_{s} & = & v_{s}e^{-\int_{t}^{s}r_{u}du}
\end{eqnarray*}
we have
\begin{equation}
d\bar{\Delta}_{s}^{t,\varepsilon,v}=\left(\bar{u}_{s}^{t,\varepsilon,v}-\bar{u}_{s}^{*}\right)^{T}dW_{s}^{\mathbb{Q}}
\end{equation}
and
\begin{eqnarray}
\bar{\Delta}_{T}^{t,\varepsilon,v} & = & \int_{t}^{t+\varepsilon}\bar{v}_{s}^{T}dW_{s}^{\mathbb{Q}}\nonumber \\
 & = & \bar{\Delta}_{t+\varepsilon}^{t,\varepsilon,v}\label{eq:4.5}
\end{eqnarray}
and we have that 
\[
\mathbb{E}_{t}^{\mathbb{Q}}\left[\bar{\Delta}_{t+\epsilon}^{t,\varepsilon,v}\right]=0
\]
and thus
\begin{equation}
\mathbb{E}_{t}\left[\bar{\Delta}_{t+\varepsilon}^{t,\varepsilon,v}\mathcal{E}\left(-\int_{t}^{t+\varepsilon}\theta_{s}dW_{s}\right)\right]=0\label{eq:4.55}
\end{equation}
Then for any $t\in\left[0,T\right)$, $v\in L_{\mathcal{F}}^{2}\left(t,t+\varepsilon;\mathbb{R}^{d}\right)$
and $\varepsilon\in\left(0,\delta\right)$ for some $\delta>0$ we
have 
\begin{eqnarray}
 &  & J\left(t,X_{t}^{*};u^{t,\varepsilon,v}\right)-J\left(t,X_{t}^{*};u^{*}\right)\nonumber \\
 & = & \mathbb{E}_{t}\left[\left(X_{t}^{*}e^{\int_{t}^{T}\mu_{s}ds}-X_{T}^{t,\varepsilon,v}\right)^{+}-\left(X_{t}^{*}e^{\int_{t}^{T}\mu_{s}ds}-X_{T}^{*}\right)^{+}\right]\nonumber \\
 & = & \mathbb{E}_{t}\left[\left(X_{t}^{*}e^{\int_{t}^{T}\mu_{s}ds}-X_{T}^{*}-\Delta_{T}^{t,\varepsilon,v}\right)^{+}-\left(X_{t}^{*}e^{\int_{t}^{T}\mu_{s}ds}-X_{T}^{*}\right)^{+}\right]\label{eq:4.4}
\end{eqnarray}
By definition \ref{4.1} and (\ref{eq:4.4}) , we could deduce that
$u$$^{*}$ is an equilibrium for our problem (\ref{eq:4.2}) if and
only if for any $t\in\left[0,T\right)$, $\exists$ $\delta>0$, s.t.
for any $\varepsilon\in\left(0,\delta\right)$ we have $v=0$ is optimal
to the following problem

\begin{align}
\underset{v}{\min} & \mbox{ \ensuremath{}\ensuremath{}\ensuremath{}\ensuremath{}\mbox{ \ensuremath{}\ensuremath{}\ensuremath{}\ensuremath{}}\mbox{ \ensuremath{}\ensuremath{}\ensuremath{}\ensuremath{}}}\mathbb{E}_{t}\left[\left(X_{t}^{*}e^{\int_{t}^{T}\mu_{s}ds}-X_{T}^{*}-\Delta_{T}^{t,\varepsilon,v}\right)^{+}-\left(X_{t}^{*}e^{\int_{t}^{T}\mu_{s}ds}-X_{T}^{*}\right)^{+}\right]\nonumber \\
s.t. & \mbox{ \ensuremath{}\ensuremath{}\ensuremath{}\ensuremath{}\mbox{ \ensuremath{}\ensuremath{}\ensuremath{}\ensuremath{}}}\begin{cases}
d\bar{\Delta}_{s}^{t,\varepsilon,v}=\left(\bar{u}_{s}^{t,\varepsilon,v}-\bar{u}_{s}^{*}\right)^{T}dW_{s}^{\mathbb{Q}}\\
v\in L_{\mathcal{F}}^{2}\left(t,t+\varepsilon;\mathbb{R}^{d}\right)
\end{cases}\label{eq:4.6}
\end{align}
which is equivalent to say that $v=0$ is optimal to the following
problem for the given $t$ and $\varepsilon$ 
\begin{align}
\underset{v}{\min} & \mbox{ \ensuremath{}\ensuremath{}\ensuremath{}\ensuremath{}\mbox{ \ensuremath{}\ensuremath{}\ensuremath{}\ensuremath{}}\mbox{ \ensuremath{}\ensuremath{}\ensuremath{}\ensuremath{}}}e^{-\int_{t}^{T}r_{s}ds}\mathbb{E}_{t}\left[\left(X_{t}^{*}e^{\int_{t}^{T}\mu_{s}ds}-X_{T}^{*}-\Delta_{T}^{t,\varepsilon,v}\right)^{+}\right]\nonumber \\
s.t. & \mbox{ \ensuremath{}\ensuremath{}\ensuremath{}\ensuremath{}\mbox{ \ensuremath{}\ensuremath{}\ensuremath{}\ensuremath{}}}\begin{cases}
d\bar{\Delta}_{s}^{t,\varepsilon,v}=\bar{v}_{s}^{T}\mathbf{1}_{s\in\left[t,t+\varepsilon\right)}dW_{s}^{\mathbb{Q}}\\
v\in L_{\mathcal{F}}^{2}\left(t,t+\varepsilon;\mathbb{R}^{d}\right)
\end{cases}\label{eq:4.6-2}
\end{align}
which can be written as 
\begin{align}
\underset{v}{\min} & \mbox{ \ensuremath{}\ensuremath{}\ensuremath{}\ensuremath{}\mbox{ \ensuremath{}\ensuremath{}\ensuremath{}\ensuremath{}}\mbox{ \ensuremath{}\ensuremath{}\ensuremath{}\ensuremath{}}}\mathbb{E}_{t}\left[\left(X_{t}^{*}e^{\int_{t}^{T}\mu_{s}-r_{s}ds}-\bar{X}_{T}^{*}-\bar{\Delta}_{T}^{t,\varepsilon,v}\right)^{+}\right]\nonumber \\
s.t. & \mbox{ \ensuremath{}\ensuremath{}\ensuremath{}\ensuremath{}\mbox{ \ensuremath{}\ensuremath{}\ensuremath{}\ensuremath{}}}\begin{cases}
d\bar{\Delta}_{s}^{t,\varepsilon,v}=\bar{v}_{s}^{T}\mathbf{1}_{s\in\left[t,t+\varepsilon\right)}dW_{s}^{\mathbb{Q}}\\
v\in L_{\mathcal{F}}^{2}\left(t,t+\varepsilon;\mathbb{R}^{d}\right)
\end{cases}\label{eq:4.6-2-1}
\end{align}
by (\ref{eq:4.5}) which can also be written as 
\begin{align}
\underset{v}{\min} & \mbox{ \ensuremath{}\ensuremath{}\ensuremath{}\ensuremath{}\mbox{ \ensuremath{}\ensuremath{}\ensuremath{}\ensuremath{}}\mbox{ \ensuremath{}\ensuremath{}\ensuremath{}\ensuremath{}}}\mathbb{E}_{t}\left[\left(X_{t}^{*}e^{\int_{t}^{T}\mu_{s}-r_{s}ds}-\bar{X}_{T}^{*}-\bar{\Delta}_{t+\varepsilon}^{t,\varepsilon,v}\right)^{+}\right]\nonumber \\
s.t. & \mbox{ \ensuremath{}\ensuremath{}\ensuremath{}\ensuremath{}\mbox{ \ensuremath{}\ensuremath{}\ensuremath{}\ensuremath{}}}\begin{cases}
\bar{\Delta}_{t+\varepsilon}^{t,\varepsilon,v}=\int_{t}^{t+\varepsilon}\bar{v}_{s}^{T}dW_{s}^{\mathbb{Q}}\\
v\in L_{\mathcal{F}}^{2}\left(t,t+\varepsilon;\mathbb{R}^{d}\right)
\end{cases}\label{eq:4.11}
\end{align}
then the statement $v=0$ is optimal to the problem (\ref{eq:4.11})
is equivalent to the statement that $\tilde{\triangle}_{t+\varepsilon}=0$
is optimal to the following problem 
\begin{align}
\underset{\tilde{\triangle}_{t+\varepsilon}}{\min} & \mbox{ \ensuremath{}\ensuremath{}\ensuremath{}\ensuremath{}\mbox{ \ensuremath{}\ensuremath{}\ensuremath{}\ensuremath{}}\mbox{ \ensuremath{}\ensuremath{}\ensuremath{}\ensuremath{}}}\mathbb{E}_{t}\left[\left(X_{t}^{*}e^{\int_{t}^{T}\mu_{s}-r_{s}ds}-\bar{X}_{T}^{*}-\tilde{\triangle}_{t+\varepsilon}\right)^{+}\right]\nonumber \\
s.t. & \mbox{ \ensuremath{}\ensuremath{}\ensuremath{}\ensuremath{}\mbox{ \ensuremath{}\ensuremath{}\ensuremath{}\ensuremath{}}}\begin{cases}
\mathbb{E}_{t}^{\mathbb{Q}}\left[\tilde{\triangle}_{t+\varepsilon}\right]=0\\
\tilde{\triangle}_{t+\varepsilon}\in L_{\mathcal{F}_{t+\varepsilon}}^{2}\left(\Omega;\mathbb{R}\right)
\end{cases}\label{eq:4.12}
\end{align}
which means that given any $v\in L_{\mathcal{F}}^{2}\left(t,t+\varepsilon;\mathbb{R}^{d}\right)$
the corresponding $\bar{\Delta}_{t+\varepsilon}^{t,\varepsilon,v}$
is an admissible $\tilde{\triangle}_{t+\varepsilon}$ in problem (\ref{eq:4.12}),
and given any admissible $\tilde{\triangle}_{t+\varepsilon}$ in problem
(\ref{eq:4.12}) we could find a $v\in L_{\mathcal{F}}^{2}\left(t,t+\varepsilon;\mathbb{R}^{d}\right)$
s.t. $\tilde{\triangle}_{t+\varepsilon}=\bar{\Delta}_{t+\varepsilon}^{t,\varepsilon,v}=\int_{t}^{t+\varepsilon}\bar{v}_{s}^{T}dW_{s}^{\mathbb{Q}}$.
This is shown as follows:
\begin{proof}
On the one hand, we have as $\theta$ is bounded that $ $ 

\begin{eqnarray}
\forall v\in L_{\mathcal{F}}^{2}\left(t,t+\varepsilon;\mathbb{R}^{d}\right)\Rightarrow\begin{cases}
\mathbb{E}_{t}^{\mathbb{Q}}\left[\bar{\Delta}_{t+\varepsilon}^{t,\varepsilon,v}\right] & =\mathbb{E}_{t}^{\mathbb{Q}}\left[\int_{t}^{t+\varepsilon}\bar{v}_{s}^{T}dW_{s}^{\mathbb{Q}}\right]\\
 & =0\\
\mathbb{E}\left[\left|\bar{\Delta}_{t+\varepsilon}^{t,\varepsilon,v}\right|^{2}\right] & =\mathbb{E}\left[\left|\int_{t}^{t+\varepsilon}\bar{v}_{s}^{T}\left(dW_{s}+\theta_{s}ds\right)\right|^{2}\right]\\
 & \leq\mathbb{E}\left[\left(\left|\int_{t}^{t+\varepsilon}\bar{v}_{s}^{T}dW_{s}\right|+\left|\int_{t}^{t+\varepsilon}\bar{v}_{s}^{T}\theta_{s}ds\right|\right)^{2}\right]\\
 & \leq2\mathbb{E}\left[\left|\int_{t}^{t+\varepsilon}\bar{v}_{s}^{T}dW_{s}\right|^{2}+\left|\int_{t}^{t+\varepsilon}\bar{v}_{s}^{T}\theta_{s}ds\right|^{2}\right]\\
 & \leq2\left(\mathbb{E}\left[\left|\int_{t}^{t+\varepsilon}\bar{v}_{s}^{T}dW_{s}\right|^{2}\right]+\mathbb{E}\left[\varepsilon\int_{t}^{t+\varepsilon}\left|\bar{v}_{s}^{T}\theta_{s}\right|^{2}ds\right]\right)\\
 & \leq2\left(\mathbb{E}\int_{t}^{t+\varepsilon}\left|\bar{v}_{s}\right|^{2}d{}_{s}+\varepsilon\mathbb{E}\left[\int_{t}^{t+\varepsilon}\left|\bar{v}_{s}\right|^{2}\left|\theta_{s}\right|^{2}ds\right]\right)\\
 & <\infty\\
\mathbb{E}_{t}^{\mathbb{Q}}\left[\left|\bar{\Delta}_{t+\varepsilon}^{t,\varepsilon,v}\right|^{2}\right] & =\mathbb{E}_{t}^{\mathbb{Q}}\left[\left|\int_{t}^{t+\varepsilon}\bar{v}_{s}^{T}dW_{s}^{\mathbb{Q}}\right|^{2}\right]\\
 & =\mathbb{E}_{t}^{\mathbb{Q}}\left[\int_{t}^{t+\varepsilon}\left|\bar{v}_{s}\right|^{2}d{}_{s}\right]\\
 & <\infty
\end{cases}
\end{eqnarray}
which means the corresponding $\bar{\Delta}_{t+\varepsilon}^{t,\varepsilon,v}\in L_{\mathcal{F}_{t+\varepsilon}}^{2}\left(\Omega;\mathbb{R}\right)$
and $\mathbb{E}_{t}^{\mathbb{Q}}\left[\bar{\Delta}_{t+\varepsilon}^{t,\varepsilon,v}\right]=0$,
thus $\bar{\Delta}_{t+\varepsilon}^{t,\varepsilon,v}$ is an admissible
$\tilde{\triangle}_{t+\varepsilon}$ in problem (\ref{eq:4.12}).
\\

\noindent On the other hand, $\forall\tilde{\triangle}_{t+\varepsilon}\in L_{\mathcal{F}_{t+\varepsilon}}^{2}\left(\Omega;\mathbb{R}\right)$
$ $with $\mathbb{E}_{t}^{\mathbb{Q}}\left[\tilde{\triangle}_{t+\varepsilon}\right]=0$
we have $\tilde{\triangle}_{t+\varepsilon}$ is $\mathcal{F}_{t+\varepsilon}$
measurable and that 

\noindent Firstly, $\mathbb{E}^{\mathbb{}}\left[\left|\tilde{\triangle}_{t+\varepsilon}\right|^{2}\right]<\infty$
implies the following Lipschitz BSDE 
\[
\begin{cases}
d\bar{\Delta}_{s}^{t,\varepsilon,v} & =\bar{v}_{s}^{T}\mathbf{1}_{s\in\left[t,t+\varepsilon\right)}\left(dW_{s}+\theta_{s}ds\right)\\
\bar{\Delta}_{t+\varepsilon}^{t,\varepsilon,v} & =\tilde{\triangle}_{t+\varepsilon}
\end{cases}
\]

\noindent $ $admits a unique solution $\left(\bar{\Delta}^{t,\varepsilon,v_{1}},\bar{v}_{1}\right)$
with $\mathbb{E}\left[\int_{t}^{t+\varepsilon}\left|\left(\bar{v}_{1}\right)_{s}\right|^{2}ds\right]<\infty$
under $\mathbb{P}$ s.t. 
\begin{eqnarray*}
\tilde{\triangle}_{t+\varepsilon} & =\bar{\Delta}_{t+\varepsilon}^{t,\varepsilon,v_{1}} & =\int_{t}^{t+\varepsilon}\left(\bar{v}_{1}\right)_{s}^{T}\left(dW_{s}+\theta_{s}ds\right)\\
 &  & =\int_{t}^{t+\varepsilon}\left(\bar{v}_{1}\right)_{s}^{T}dW_{s}^{\mathbb{Q}}
\end{eqnarray*}

\noindent Secondly, $\mathbb{E}^{\mathbb{Q}}\left[\left|\tilde{\triangle}_{t+\varepsilon}\right|^{2}\right]<\infty$
implies by using martingale representation theorem under $\mathbb{Q}$
that there exists a unique $\bar{v}_{2}$ with $\mathbb{E^{\mathbb{Q}}}\left[\int_{t}^{t+\varepsilon}\left|\left(\bar{v}_{2}\right)_{s}\right|^{2}ds\right]<\infty$
s.t. 
\begin{eqnarray*}
\tilde{\triangle}_{t+\varepsilon} & = & \mathbb{E}_{t}^{\mathbb{Q}}\left[\tilde{\triangle}_{t+\varepsilon}\right]+\int_{t}^{t+\varepsilon}\left(\bar{v}_{2}\right)_{s}^{T}dW_{s}^{\mathbb{Q}}\\
 & = & \int_{t}^{t+\varepsilon}\left(\bar{v}_{2}\right)_{s}^{T}dW_{s}^{\mathbb{Q}}
\end{eqnarray*}

\noindent Then the above two equations for $\tilde{\triangle}_{t+\varepsilon}$
implies that $\bar{v}_{1}=\bar{v}_{2}$, which means there exists
a $v\in L_{\mathcal{F}}^{2}\left(t,t+\varepsilon;\mathbb{R}^{d}\right)$
with $\bar{v}=\bar{v}_{1}=\bar{v}_{2}$ s.t. $\tilde{\triangle}_{t+\varepsilon}=\int_{t}^{t+\varepsilon}\bar{v}_{s}^{T}dW_{s}^{\mathbb{Q}}$
\end{proof}
\noindent Then by (\ref{eq:7.1}), the problem (\ref{eq:4.12}) for
the given $t$ and $\varepsilon$ can be written as

\begin{align}
\underset{\tilde{\triangle}_{t+\varepsilon}}{\min} & \mbox{ \ensuremath{}\ensuremath{}\ensuremath{}\ensuremath{}\mbox{ \ensuremath{}\ensuremath{}\ensuremath{}\ensuremath{}}\mbox{ \ensuremath{}\ensuremath{}\ensuremath{}\ensuremath{}}}\mathbb{E}_{t}\left[\left(X_{t}^{*}e^{\int_{t}^{T}\mu_{s}-r_{s}ds}-\bar{X}_{T}^{*}-\tilde{\triangle}_{t+\varepsilon}\right)^{+}\right]\nonumber \\
s.t. & \mbox{ \ensuremath{}\ensuremath{}\ensuremath{}\ensuremath{}\mbox{ \ensuremath{}\ensuremath{}\ensuremath{}\ensuremath{}}}\begin{cases}
\mathbb{E}_{t}\left[\tilde{\triangle}_{t+\varepsilon}\mathcal{E}\left(-\int_{t}^{t+\varepsilon}\theta_{s}dW_{s}\right)\right]=0\\
\tilde{\triangle}_{t+\varepsilon}\in L_{\mathcal{F}_{t+\varepsilon}}^{2}\left(\Omega;\mathbb{R}\right)
\end{cases}\label{eq:4.13}
\end{align}
which could be transformed to the following problem using Lagrangian
multiplier method
\begin{align}
\underset{\tilde{\triangle}_{t+\varepsilon}}{\min} & \mbox{ \ensuremath{}\ensuremath{}\ensuremath{}\ensuremath{}\mbox{ \ensuremath{}\ensuremath{}\ensuremath{}\ensuremath{}}\mbox{ \ensuremath{}\ensuremath{}\ensuremath{}\ensuremath{}}}\mathbb{E}_{t}\left[\left(X_{t}^{*}e^{\int_{t}^{T}\mu_{s}-r_{s}ds}-\bar{X}_{T}^{*}-\tilde{\triangle}_{t+\varepsilon}\right)^{+}-\lambda\tilde{\triangle}_{t+\varepsilon}\mathcal{E}\left(-\int_{t}^{t+\varepsilon}\theta_{s}dW_{s}\right)\right]\nonumber \\
s.t. & \mbox{ \ensuremath{}\ensuremath{}\ensuremath{}\ensuremath{}\mbox{ \ensuremath{}\ensuremath{}\ensuremath{}\ensuremath{}}}\begin{cases}
\tilde{\triangle}_{t+\varepsilon}\in L_{\mathcal{F}_{t+\varepsilon}}^{2}\left(\Omega;\mathbb{R}\right)\end{cases}\label{eq:4.14}
\end{align}
$ $so we have deduced that $u$$^{*}$ is an equilibrium for our
problem (\ref{eq:4.2}) if and only if for any $t\in\left[0,T\right)$,
$\exists$ $\delta>0$, s.t. for any $\varepsilon\in\left(0,\delta\right)$
we have that $\tilde{\triangle}_{t+\varepsilon}=0$ is optimal to
the above problem (\ref{eq:4.14}).\\

\noindent Then problem (\ref{eq:4.14}) can be written as
\begin{align}
\underset{\tilde{\triangle}_{t+\varepsilon}}{\min} & \mbox{ \ensuremath{}\ensuremath{}\ensuremath{}\ensuremath{}\mbox{ \ensuremath{}\ensuremath{}\ensuremath{}\ensuremath{}}\mbox{ \ensuremath{}\ensuremath{}\ensuremath{}\ensuremath{}}}\mathbb{E}_{t}\left[f(\tilde{\triangle}_{t+\varepsilon})-\lambda\tilde{\triangle}_{t+\varepsilon}\mathcal{E}\left(-\int_{t}^{t+\varepsilon}\theta_{s}dW_{s}\right)\right]\nonumber \\
s.t. & \mbox{ \ensuremath{}\ensuremath{}\ensuremath{}\ensuremath{}\mbox{ \ensuremath{}\ensuremath{}\ensuremath{}\ensuremath{}}}\begin{cases}
\tilde{\triangle}_{t+\varepsilon}\in L_{\mathcal{F}_{t+\varepsilon}}^{2}\left(\Omega;\mathbb{R}\right)\end{cases}\label{eq:4.13-1-1-1}
\end{align}
where $f(y)=\left(X_{t}^{*}e^{\int_{t}^{T}\mu_{s}-r_{s}ds}-\bar{X}_{T}^{*}-y\right)^{+}$
and let $f(y,w)=\left[\left(X_{t}^{*}e^{\int_{t}^{T}\mu_{s}-r_{s}ds}-\bar{X}_{T}^{*}-y\right)^{+}\right]\left(w\right),\forall w\in\Omega$\\

\noindent since $f(\tilde{\triangle}_{t+\varepsilon}\left(w\right),w)-\lambda\tilde{\triangle}_{t+\varepsilon}\left(w\right)\mathcal{E}\left(-\int_{t}^{t+\varepsilon}\theta_{s}dW_{s}\right)\left(w\right)$
is $\mbox{}$a convex and differentiable function with respect to
$\tilde{\triangle}_{t+\varepsilon}$$\left(w\right)$ for any $w\in\Omega$,
then we deduce that $\tilde{\triangle}_{t+\varepsilon}=0$ is optimal
if and only if
\begin{equation}
f^{\prime}(0,w)-\lambda\mathcal{E}\left(-\int_{t}^{t+\varepsilon}\theta_{s}dW_{s}\right)\left(w\right)=0,\forall w\in\Omega\label{eq:4.17}
\end{equation}
and we have
\begin{equation}
f^{\prime}(0,w)=\left[-\mathbf{1}{}_{\left(0<X_{t}^{*}e^{\int_{t}^{T}\mu_{s}-r_{s}ds}-\bar{X}_{T}^{*}\right)}\right]\left(w\right)
\end{equation}
since $f^{\prime}(0,w)\in\left[-1,0\right]$, while $\mathcal{E}\left(-\int_{t}^{t+\varepsilon}\theta_{s}dW_{s}\right)\left(w\right)$
could blow up towards $\infty$, we must have $\lambda=0$ to make
(\ref{eq:4.17}) achievable and thus we have
\[
f^{\prime}(0,w)=0,\forall w\in\Omega
\]
which is again achievable if and only if 
\begin{eqnarray*}
X_{t}^{*}e^{\int_{t}^{T}\mu_{s}-r_{s}ds} & \leq & \bar{X}_{T}^{*}\\
 & = & X_{t}^{*}+\int_{t}^{T}\left(\bar{u}_{s}^{*}\right)^{T}dW_{s}^{\mathbb{Q}}
\end{eqnarray*}
which means
\begin{equation}
\int_{t}^{T}\left(\bar{u}_{s}^{*}\right)^{T}dW_{s}^{\mathbb{Q}}\geq X_{t}^{*}\left(e^{\int_{t}^{T}\mu_{s}-r_{s}ds}-1\right)\label{eq:4.19}
\end{equation}
thus we have showed that $\tilde{\triangle}_{t+\varepsilon}=0$ is
optimal to problem (\ref{eq:4.14}) if and only if $ $(\ref{eq:4.19})
is satisfied. Since we have also showed above that $u$$^{*}$ is
an equilibrium for our problem (\ref{eq:4.2}) if and only if for
any $t\in\left[0,T\right)$, $\exists$ $\delta>0$, s.t. for any
$\varepsilon\in\left(0,\delta\right)$ we have that $\tilde{\triangle}_{t+\varepsilon}=0$
is optimal to problem (\ref{eq:4.14}), so we conclude by definition
\ref{4.1} that 
\begin{itemize}
\item $u$$^{*}$ is an equilibrium for the family of problems (\ref{eq:4.2})
if and only if for any $t\in\left[0,T\right)$, $\exists$ $\delta>0$,
s.t. for any $\varepsilon\in\left(0,\delta\right)$ 
\[
\int_{t}^{T}\left(\bar{u}_{s}^{*}\right)^{T}dW_{s}^{\mathbb{Q}}\geq X_{t}^{*}\left(e^{\int_{t}^{T}\mu_{s}-r_{s}ds}-1\right)
\]

\item since $\mathbb{E}_{t}\left[X_{T}^{*}\right]=X_{t}^{*}e^{\int_{t}^{T}\mu_{s}ds}$
is equivalent to $X_{t}^{*}\left(e^{\int_{t}^{T}\mu_{s}-r_{s}ds}-1\right)=\int_{t}^{T}\mathbb{E}_{t}\left[\left(\bar{u}_{s}^{*}\right)^{T}\theta_{s}\right]ds$,
then $u$$^{*}\in\overline{U_{ad}^{u}}$ and thus is an equilibrium
for the family of problems (\ref{eq:4.1-2}) if for any $t\in\left[0,T\right)$,
$\exists$ $\delta>0$, s.t. for any $\varepsilon\in\left(0,\delta\right)$
we have 
\[
\begin{cases}
\int_{t}^{T}\left(\bar{u}_{s}^{*}\right)^{T}dW_{s}^{\mathbb{Q}}\geq X_{t}^{*}\left(e^{\int_{t}^{T}\mu_{s}-r_{s}ds}-1\right)\mbox{\ensuremath{\mbox{ }}and } & \mbox{}\\
\int_{t}^{T}\mathbb{E}_{t}\left[\left(\bar{u}_{s}^{*}\right)^{T}\theta_{s}\right]ds=X_{t}^{*}\left(e^{\int_{t}^{T}\mu_{s}-r_{s}ds}-1\right)
\end{cases}
\]

\end{itemize}
It is clear that $u$$^{*}=0$ is an equilibrium for (\ref{eq:4.2})
when $\mu\leq r$ and $x_{0}\geq0$ with the corresponding state process
$X_{t}^{*}=x_{0}e^{\int_{0}^{t}r_{s}ds}$. And this $u$$^{*}=0$
is also an equilibrium for (\ref{eq:4.1-2}) when $r=\mu$.

\section{Conclusion}

In this paper, we have studied the time inconsistent stochastic control
problems of portfolio selection by using different utility functions
with a moving target that need to be met. And we solve our problem
by finding equilibrium controls under our definition as the optimal
controls. This paper has also posed some open questions during the
procedure of solving our problems such as how to prove the existence
of solutions for our derived system of ODEs when we solve our family
of problems using utility function $h(x)=-\frac{x^{3}}{3}$ and $h(x)=\frac{x^{4}}{4}$,
which could be good further research topics in this area.

\noindent 
\[
\]

\end{document}